\documentclass[aps,twocolumn,prb,amsmath,amssymb,showpacs,notitlepage,superscriptaddress]{revtex4-2}

\usepackage[utf8]{inputenc}
\usepackage[T1]{fontenc}
\usepackage{textcomp}
\usepackage[english]{babel}
\usepackage{color}
\usepackage{amssymb}
\usepackage{amsmath}
\usepackage{amsfonts}
\usepackage{amsopn}
\usepackage{amsbsy}
\usepackage{enumerate}
\usepackage[plainpages=false,pdfpagelabels,breaklinks=true,pdfborderstyle={/S/U/W 1.2}]{hyperref}
\usepackage{graphics}
\usepackage{units}
\usepackage[cal=boondoxo,bb=boondox]{mathalfa}
\usepackage{dsfont}
\usepackage{csquotes}
\numberwithin{equation}{section} 
\numberwithin{table}{section} 
\numberwithin{figure}{section} 
\usepackage[amsmath,thmmarks,thref,hyperref]{ntheorem}

\theoremstyle{plain}
\newtheorem{theorem}{Theorem}[section]
\newtheorem{definition}[theorem]{Definition}
\newtheorem{lemma}[theorem]{Lemma}

\newtheorem{conjecture}[theorem]{Conjecture}

\newtheorem{assumption}[theorem]{Assumption}

\theorembodyfont{\upshape}

\theoremstyle{nonumberplain}

\theoremsymbol{\ensuremath{_\Box}}
\newtheorem{proof}{Proof}
\usepackage[scaled]{beramono}
\usepackage{helvet}
\let\mathcal\undefined
\usepackage[charter]{mathdesign} % math font has ugly \mathcals
\providecommand{\ie}{i.~e.~}
\providecommand{\eg}{e.~g.~}
\providecommand{\cf}{cf.~}

\providecommand{\R}{\mathbb{R}}

\providecommand{\C}{\mathbb{C}}
\renewcommand{\C}{\mathbb{C}}

\providecommand{\T}{\mathbb{T}}
\renewcommand{\T}{\mathbb{T}}
\providecommand{\N}{\mathbb{N}}
\providecommand{\Z}{\mathbb{Z}}
\providecommand{\ii}{\mathrm{i}}
\providecommand{\e}{\mathrm{e}}
\renewcommand{\Re}{\mathrm{Re} \,}
\renewcommand{\Im}{\mathrm{Im} \,}
\providecommand{\Hil}{\mathcal{H}}
\providecommand{\Sone}{\mathbb{S}^1}
\providecommand{\eps}{\varepsilon}

\providecommand{\ker}{\mathrm{ker} \, }
\providecommand{\ran}{\mathrm{ran} \, }

\providecommand{\ker}{\mathrm{ker} \,}

\providecommand{\dd}{\mathrm{d}}
\providecommand{\id}{\mathds{1}}
\providecommand{\sabs}[1]{\lvert #1 \vert}

\providecommand{\norm}[1]{\left \lVert #1 \right \rVert}
\providecommand{\snorm}[1]{\lVert #1 \rVert}

\providecommand{\scpro}[2]{\left \langle #1 , #2 \right \rangle}
\providecommand{\sscpro}[2]{\langle #1 , #2 \rangle}
\providecommand{\bscpro}[2]{\bigl \langle #1 , #2 \bigr \rangle}

\providecommand{\sket}[1]{\vert #1 \rangle}

\providecommand{\sbra}[1]{\langle #1 \vert}

\providecommand{\sopro}[2]{\vert #1 \rangle \langle #2 \vert}

\makeatletter
\newsavebox{\@brx}
\newcommand{\llangle}[1][]{\savebox{\@brx}{\(\m@th{#1\langle}\)}%
	\mathopen{\copy\@brx\kern-0.5\wd\@brx\usebox{\@brx}}}
\newcommand{\rrangle}[1][]{\savebox{\@brx}{\(\m@th{#1\rangle}\)}%
	\mathclose{\copy\@brx\kern-0.5\wd\@brx\usebox{\@brx}}}
\makeatother

\providecommand{\scppro}[2]{\llangle #1 , #2 \rrangle}
\providecommand{\sscppro}[2]{\llangle #1 , #2 \rrangle}
\providecommand{\bscppro}[2]{\llangle[\big] #1 , #2 \rrangle[\big]}

\providecommand{\skket}[1]{\Vert #1 \rrangle}

\providecommand{\sbrbra}[1]{\llangle #1 \Vert}

\providecommand{\soppro}[2]{\Vert #1 \rrangle \llangle #2 \Vert}

\providecommand{\BZ}{\T^*}
\providecommand{\sgn}{\mathrm{sgn} \, }
\usepackage{dcolumn}% Align table columns on decimal point
\usepackage{multirow}
\usepackage{diagxy}
\begin{document}

\title{On Choosing a Physically Meaningful Topological Classification \\ for Non-Hermitian Systems and the Issue of Diagonalizability}
\author{Max Lein}
% \email{max.lein@tohoku.ac.jp}
% \affiliation{Advanced Institute of Materials Research,
	% Tohoku University, Sendai 980-8577, Japan}

% 
\begin{abstract}
	The topological classification of hermitian operators is solely determined by the presence or absence of certain discrete symmetries. For non-hermitian operators we in addition need to specify the type of spectral gap \cite{Kawabata_Shiozaki_Ueda_Sato:classification_non_hermitian_systems:2019,Zhou_Lee:non_hermitian_topological_classification:2019}. They come in the flavor of a point gap or a line gap. Since the presence of a line gap implies the existence of a point gap, there is usually more than one mathematical classification applicable to a physical system. That raises the question: which of these gap-type classifications is physically meaningful? 
	
	To decide this question, I propose a simple criterion, namely the choice of \emph{physically relevant states}. This generalizes the notion of Fermi projection that plays a crucial role in the topological classification of fermionic condensed matter systems, and enters as an \emph{auxiliary quantity} in the bulk classification of photonic \cite{Raghu_Haldane:quantum_Hall_effect_photonic_crystals:2008,DeNittis_Lein:symmetries_electromagnetism:2020} and magnonic crystals \cite{Shindou_et_al:chiral_magnonic_edge_modes:2013,Lein_Sato:topological_classification_magnons:2019}. After that the classification is entirely algorithmic, the system's topology is encoded in (pairs of) projections with symmetries and constraints. 
	A crucial point in my investigation is the relevance of diagonalizability. Even for existing topological classifications of non-hermitian systems diagonalizability needs to be assumed to ensure that continuous deformations of the hamiltonian lead to continuous deformations of the spectra, projections and unitaries. 
\end{abstract}
\maketitle
\tableofcontents

%%% begin content %%% (fold)
%!TEX root = /Users/max/Dropbox/research/topological classification non-selfadjoint hamiltonians/paper/non_hermitian_classification.tex
\section{Introduction} % (fold)
\label{intro}
Engineering systems with non-trivial topology has become a standard tool if one wants to create systems with very robust edge or surface states. Topologically protected boundary modes have been realized in a wide range of quantum \cite{von_Klitzing_Dorda_Pepper:quantum_hall_effect:1980,Thouless_Kohmoto_Nightingale_Den_Nijs:quantized_hall_conductance:1982,Simon:holonomy_Berrys_phase:1983,Hatsugai:Chern_number_edge_states:1993,Hatsugai:edge_states_Riemann_surface:1993,Shindou_et_al:chiral_magnonic_edge_modes:2013,Lein_Sato:topological_classification_magnons:2019,Chiu_Teo_Schnyder_Ryu:classification_topological_insulators:2016,Prodan_Schulz_Baldes:complex_topological_insulators:2016} and classical waves \cite{Raghu_Haldane:quantum_Hall_effect_photonic_crystals:2008,Wang_et_al:edge_modes_photonic_crystal:2008,Rechtsman_Zeuner_et_al:photonic_topological_insulators:2013,DeNittis_Lein:symmetries_electromagnetism:2020,Fleury_et_al:breaking_TR_acoustic_waves:2014,Safavi-Naeini_et_al:2d_phonoic_photonic_band_gap_crystal_cavity:2014,Peano_Brendel_Schmidt_Marquardt:topological_phases_sound_light:2015,Chen_Zhao_Mei_Wu:acoustic_frequency_filter_topological_phononic_crystals:2017,Suesstrunk_Huber:mechanical_topological_insulator:2015,Suesstrunk_Huber:classification_mechanical_metamaterials:2016,Lein_Sato:topological_classification_magnons:2019,DeNittis_Lein:symmetries_electromagnetism:2020,Bliokh_Leykam_Lein_Nori:topological_classification_homogeneous_electromagnetic_media:2019,Ozawa_et_al:review_topological_photonics:2018,Kondo_Akagi_Katsura:Z2_topological_invariant_magnons:2019}. Their existence hinges on the presence of spectral gaps (or, more generally, dynamical localization) and selectively breaking or preserving certain discrete symmetries. 

Topology typically manifests itself through the presence of boundary modes, which are very robust against perturbations and disorder. More precisely, an effect is considered topological if it can be explained by means of a \emph{bulk-boundary correspondence,}
\begin{align}
	O_{\mathrm{bdy}}(t) \approx T_{\mathrm{bdy}} = f(T_{\mathrm{bulk}}) 
	. 
	\label{intro:eqn:bulk_boundary_correspondence}
\end{align}
The first (approximate) equality relates a physical observable $O_{\mathrm{bdy}}(t)$ to a topological invariant $T_{\mathrm{bdy}}$ defined for the semi-infinite system with boundary; it gives the abstract mathematical quantity $T_{\mathrm{bdy}}$ physical significance. The second equality is the “mathematical bulk-boundary correspondence”, which allows me to compute the value of the boundary invariant from the topological bulk invariants $T_{\mathrm{bulk}}$ through a function $f$. For the Quantum Hall Effect at the interface between materials the physical observable $O_{\mathrm{bdy}} = \sigma^{\perp}_{\mathrm{bdy}}$ is the transverse conductivity at the boundary; the topological boundary invariant is the spectral flow, which can be predicted from the difference $f(x,y) = x - y$ of two Chern numbers $T_{\mathrm{bulk}} = (\mathrm{Ch}_1,\mathrm{Ch}_2)$ of the materials \cite{Prodan_Schulz_Baldes:complex_topological_insulators:2016}. As the names suggest, these topological invariants cannot change their values during symmetry- and gap-preserving continuous deformations of the bulk systems, since they typically take values in $\Z$ or $\Z_2$. That, in turn, explains the extraordinary robustness of the boundary modes under perturbations — the only continuous, integer-valued function is the constant function. 

The starting point to finding or deriving bulk-boundary correspondences~\eqref{intro:eqn:bulk_boundary_correspondence} in topological insulators is to \emph{classify the bulk system} and get a complete list of topological bulk invariants for the infinite system. In what follows, when I write topological classification, I shall always mean the classification of \emph{bulk} systems unless explicitly stated otherwise. 

When the system is described by a hermitian operator, the situation is well-understood by now: operators belong to one of 10 Cartan-Altland-Zirnbauer classes \cite{Altland_Zirnbauer:superconductors_symmetries:1997,Chiu_Teo_Schnyder_Ryu:classification_topological_insulators:2016,Ozawa_et_al:review_topological_photonics:2018}. Inside each class, there are inequivalent phases, which are defined by continuous, gap- and symmetry-preserving deformations. These bulk phases can be labeled by topological invariants such as Chern numbers \cite{Panati:triviality_Bloch_bundle:2006,Prodan_Schulz_Baldes:complex_topological_insulators:2016,Hatcher:vector_bundles_K_theory:2009,DeNittis_Lein:exponentially_loc_Wannier:2011} and the Kane-Mele invariant \cite{Kane_Mele:Z2_ordering_spin_quantum_Hall_effect:2005,DeNittis_Gomi:AII_bundles:2014}. Also the interplay with crystallographic symmetries has been analyzed \cite{Gomi:topological_classification_crystallographic_point_groups_2d:2017,Shiozaki_Sato_Gomi:band_topology_3d_crystallographic_groups:2018} recently. 

However, many media for classical waves are described by non-hermitian operators. These differ from hermitian operators in two ways: 
\begin{enumerate}[(1)]
	\item Their spectrum may be complex. 
	\item Non-hermitian operators may possess Jordan blocks, \ie they need not be diagonalizable. 
\end{enumerate}
Not surprisingly, the zoology of non-hermitian operators is much richer. Independently, Kawabata et al. \cite{Gong_et_al:topological_phases_non_hermitian_systems:2018,Kawabata_Shiozaki_Ueda_Sato:classification_non_hermitian_systems:2019} and Zhou and Lee \cite{Zhou_Lee:non_hermitian_topological_classification:2019} have extended the Cartan-Altland-Zirnbauer classification, and they find non-hermitian operators belong to one of 38 topological classes. There have been other noteworthy works in this direction. De~Nittis and Gomi have developed a mathematical framework to classify dynamically stable pseudohermitian systems by means of a suitably adapted $K$-theory \cite{DeNittis_Gomi:K_theoretic_classification_operators_on_Krein_spaces:2019}. And Wojcik et al.\ took the homotopy-theoretic route and related the topology of certain non-hermitian operators to (non-abelian) braid groups \cite{Wojcik_Sun_Bzdusek_Fan:topological_classification_non_hermitian_hamiltonians:2020}. A crucial insight in the 38-Fold Way Classification of Kawabata et al.\ is the distinction between different kinds of gaps. As the spectrum is a subset of the complex plane $\C \simeq \R^2$, spectral gaps — obstacles for continuous deformations of operators — can be $0$- or $1$-dimensional; this is referred to as point gap and line gap classifications, respectively. Very often, the relevant line gaps are the imaginary or real axis, which give rise to the real and imaginary line gap classification (the order is reversed). 

\begin{figure}
	\begin{centering}
		\resizebox{70mm}{!}{\includegraphics{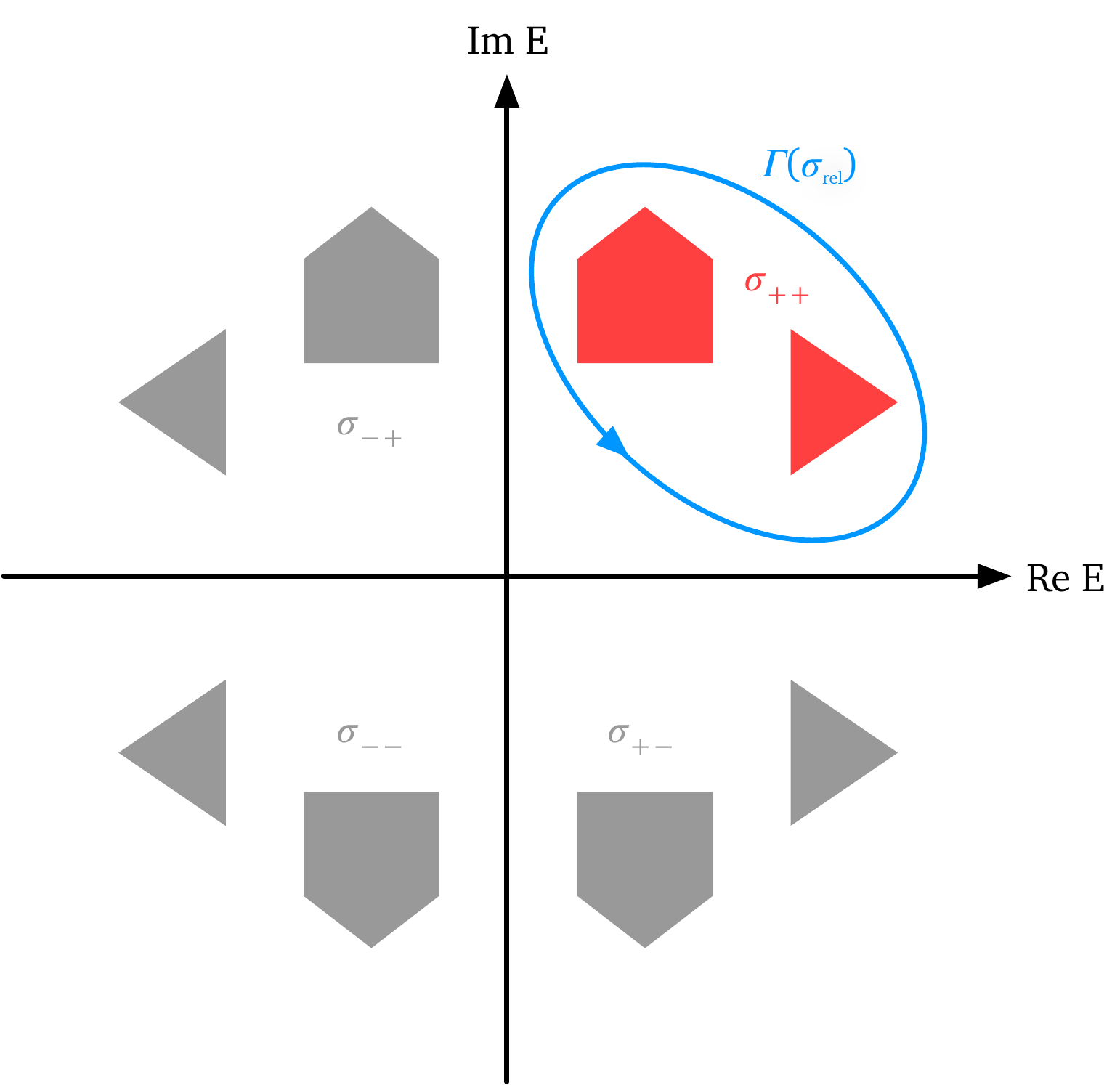}}
	\end{centering}
	\caption{Spectrum with point symmetry and reflection symmetries about real and imaginary axis. Very often such symmetries are due to the presence of discrete symmetries such as time-reversal symmetries and chiral symmetries. Part of the spectrum needs to be designated as “relevant”; for example, $\sigma_{\mathrm{rel}} = \sigma_{++}$ is one choice. The relevant spectrum can be encircled by a contour, which enters the relevant projection via equation~\eqref{intro:eqn:definition_P_rel}. }
	\label{intro:figure:highly_symmetric_spectrum}
\end{figure}

However, any operator with a line gap also possesses a point gap. Take an operator whose spectrum is sketched by Figure~\ref{intro:figure:highly_symmetric_spectrum}. It has a point gap and both, real and imaginary line gaps, leading to potentially three distinct topological classifications. Generally, the three different classifications will disagree, so not only do I have a choice, the choice matters. Which of these three mathematical classifications is physically meaningful? What physical data decide how to pick one classification over another? 

In short, this is the topic of this work. The proposal I put forth is as follows: 
\begin{enumerate}[(A)]
	\item Based on the physics of the system I pick the \emph{relevant part of the spectrum} $\sigma_{\mathrm{rel}}$, \ie I single out states associated to one particular spectral region. As is usual in the theory of topological insulators, $\sigma_{\mathrm{rel}}$ needs to be separated from the rest of the spectrum by a gap. 
	\item I define the spectral projection $P_{\mathrm{rel}}$ onto the states located inside of $\sigma_{\mathrm{rel}}$. 
	\item Symmetries of the hamiltonian $H$ and the relevant spectrum $\sigma_{\mathrm{rel}}$ will lead to symmetries and constraints of the \emph{relevant projection} $P_{\mathrm{rel}}$ and potentially a second projection $P_{\mathrm{rel},\dagger}$ defined from $H^{\dagger}$. 
	\item Then I classify (pairs of) projections with symmetries and constraints. 
\end{enumerate}
Let me give you a little more detail on each of these points to motivate the material covered in the main body of this paper. The following also furnishes an outline of this paper.

\subsection{The notion of physically relevant states} % (fold)
\label{intro:relevant_states}
The notion of \emph{physically relevant states} is motivated from the Fermi projection for hermitian systems, which describes the state of many fermionic condensed matter systems at zero temperature \cite{Grosso_Parravicini:solid_state_physics:2003}. Here all states below the characteristic Fermi energy are filled whereas those in the conduction band are empty. The word “insulators” in topological insulators points to the fact that the Fermi energy lies in a spectral gap or, more generally, in a region of dynamical localization \cite{Anderson:Anderson_localization_absence_diffusion:1958,Bellissard_van_Elst_Schulz_Baldes:noncommutative_geometry_quantum_hall_effect:1994,John:Anderson_localization_light:1991}. Hence, there is no direct conductivity and the bulk is insulating. 

The “Fermi projection” need not always be interpreted as the state of the system. For classical and many bosonic waves the “Fermi projection” is just an \emph{auxiliary quantity} that enters the bulk-boundary correspondence 
\cite{Raghu_Haldane:quantum_Hall_effect_photonic_crystals:2008,Shindou_et_al:chiral_magnonic_edge_modes:2013,DeNittis_Lein:symmetries_electromagnetism:2020,Lein_Sato:topological_classification_magnons:2019}. Given that non-hermitian systems are easily realized with classical waves, this subtlety of distinguishing mathematics from physical interpretation is essential. 

For the operator from Figure~\ref{intro:figure:highly_symmetric_spectrum} there are essentially $4$ distinct scenarios: I could pick a single spectral island such as $\sigma_{++}$, which corresponds to a point gap. This also takes care of the case of three spectral islands (the spectral complement). When I select two spectral islands, I have essentially three choices, I could choose states to the left of the imaginary axis, below the real axis or point symmetrically (\eg $\sigma_{++}$ and $\sigma_{--}$); three of these choices corresponds to point gap and real and imaginary line gap. The point-symmetric case does not seem to be covered by Kawabata et al. 
% subsection the_notion_of_physically_relevant_states (end)

\subsection{Defining the projection onto the relevant states} % (fold)
\label{intro:defining_P_rel}
Having picked $\sigma_{\mathrm{rel}}$, I then define the projection onto the relevant states — or \emph{relevant projection} for short 
\begin{align}
	P_{\mathrm{rel}} = \frac{\ii}{2\pi} \int_{\Gamma(\sigma_{\mathrm{rel}})} \dd z \, (H - z)^{-1} 
	\label{intro:eqn:definition_P_rel}
\end{align}
as a countour integral where $\Gamma(\sigma_{\mathrm{rel}})$ encloses only $\sigma_{\mathrm{rel}}$. Other definitions are possible, \eg via functional calculus~\eqref{diagonalizable_operators:eqn:relevant_projection_functional_calculus}, but the projection operator I obtain is independent of that — if $P_{\mathrm{rel}}$ exists. In some cases I will also need to include $P_{\mathrm{rel},\dagger}$ in the analysis, which is defined via equation~\eqref{intro:eqn:definition_P_rel} after replacing $H$ with its adjoint $H^{\dagger}$. 

When $H$ is hermitian, the existence is well-known: the resolvent only has first-order poles and therefore is well-defined. But in case $\sigma_{\mathrm{rel}}$ contains Jordan blocks, \ie when $H$ is \emph{not diagonalizable}, then the contour encloses higher-order poles of the resolvent and this formula is not well-defined. To ensure $P_{\mathrm{rel}}$ is well-defined, I need to make the following 
\begin{assumption}\label{intro:assumption:diagonalizability}
	The Hamiltonian $H$ is a \emph{diagonalizable}, bounded operator that possesses a bounded inverse. 
\end{assumption}
The precise mathematical definition and ramifications of diagonalizability are covered in Section~\ref{diagonalizable_operators}. The boundedness of $H$ is not an essential assumption, only diagonalizability is. Assuming boundedness just simplifies many arguments and is certainly satisfied for tight-binding operators, which make up the bulk of all effective models for periodic and many disordered systems. 

After reviewing the 38-Fold Classification in Section~\ref{38_fold_recap}, I will give arguments in Section~\ref{why_diagonalizability_matters} why diagonalizability is an \emph{essential} assumption and why \cite{Kawabata_Shiozaki_Ueda_Sato:classification_non_hermitian_systems:2019,Zhou_Lee:non_hermitian_topological_classification:2019} can only classify diagonalizable operators. The critical point is that even if I \emph{continuously} deform a non-hermitian operator $H(\lambda)$, the spectrum, spectral projections and many other associated quantities that enter the topological classification \emph{may have discontinuities.} 
% subsection defining_the_projection_onto_the_relevant_states (end)

\subsection{Symmetries and constraints of the relevant projections} % (fold)
\label{intro:symmetries_constraints}
Sections~\ref{my_way_simplified} and \ref{my_way} are dedicated to explaining how symmetries of the hamiltonian $H$ and the relevant spectrum $\sigma_{\mathrm{rel}}$ lead to symmetries and constraints of the relevant projection $P_{\mathrm{rel}}$. That is because discrete symmetries — if present — relate states from different spectral islands. Three scenarios emerge: (1)~symmetries preserve the relevant states, 
\begin{align*}
	U \, P_{\mathrm{rel}} \, U^{-1} = P_{\mathrm{rel}}
	, 
\end{align*}
(2)~constraints exchange relevant and irrelevant states, 
\begin{align*}
	U \, P_{\mathrm{rel}} \, U^{-1} = \id_{\Hil} - P_{\mathrm{rel}}
	, 
\end{align*}
or (3)~symmetries may be broken, 
\begin{align*}
	U \, P_{\mathrm{rel}} \, U^{-1} \neq P_{\mathrm{rel}} , \; \id_{\Hil} - P_{\mathrm{rel}}
	. 
\end{align*}
Scenarios (1) and (2) also come in a flavor that involves spectral projections of $H^{\dagger}$, either $P_{\mathrm{rel},\dagger}$ or its complement. 
% subsection symmetries_and_constraints_of_the_relevant_projections (end)

\subsection{Classifying projections with symmetries and constraints} % (fold)
\label{intro:classification}
The last step consists of classifying (pairs of) projections with symmetries and constraints. I will not present a generic recipe of my own here, but purposefully use two different standard techniques, vector bundle theory and $K$-theory, to emphasize that my ideas are independent of how to attack the actual classification problem. Indeed, this is why I have picked non-trivial examples that can nevertheless be treated with existing tools to obtain a topological classification (\cf Section~\ref{my_way_simplified:non_hermitian_examples}). These points will be reiterated in the summary (\cf Section~\ref{discussion}), where I will critically compare my classification to results from the literature and make some comments about future developments. 
% subsection classifying_projections_with_symmetries_and_constraints (end)
% section introduction (end)
%!TEX root = /Users/max/Dropbox/research/topological classification non-selfadjoint hamiltonians/paper/non_hermitian_classification.tex
\section{Diagonalizable operators are normal operators with respect to the biorthogonal scalar product} % (fold)
\label{diagonalizable_operators}
Diagonalizability is intrinsically connected to the \emph{absence of Jordan blocks.} For a diagonalizable $N \times N$ matrix $H$ this means that the number of linearly independent \emph{proper} eigenvectors $v_n$ is exactly $N$. So the \emph{eigenvectors form a basis} of my vector space, I can expand any vector in terms of proper eigenvectors, namely 
\begin{align*}
	u = \sum_{n = 1}^N c_n \, v_n
	. 
\end{align*}
When I collect the eigenvectors into a matrix 
\begin{align*}
	G^{-1} = \bigl ( v_1 \vert \cdots \vert v_N \bigr )
	, 
\end{align*}
then adjoining with $G$ diagonalizes $H$, 
\begin{align}
	G \, H \, G^{-1} = D 
	= \mathrm{diag}(E_1 , \ldots , E_N) 
	. 
	\label{diagonalizable_operators:eqn:diagonalizable_matrix}
\end{align}
Put another way, a matrix is diagonalizable exactly when there exists a \emph{similarity transform} $G$ which diagonalizes $H$. 

Succinctly, the completeness relation of the eigenvectors can be rewritten as a \emph{resolution of the identity}
\begin{align}
	\id_{\C^N} = \sum_{n = 1}^N \sopro{v_n}{v_n} 
	. 
	\label{diagonalizable_operators:eqn:resolution_identity_matrix}
\end{align}
The above formula uses the scalar product~\eqref{diagonalizable_operators:eqn:weighted_scalar_product} on $\C^N$, which \emph{declares} $\{ v_1 , \ldots , v_N \}$ to be orthonormal; this is better known in the physics community as the scalar product from the biorthogonal formalism \cite{Brody:biorthogonal_quantum_mechanics:2014}. Mathematically speaking, I am free to choose a “non-standard” scalar product and I emphasize that this places no additional restrictions on $H$. 

Unlike for matrices — there seemingly exists no universally accepted generalization of diagonalizability to operators on infinite-dimensional normed vector spaces in the literature; my extension mimics~\eqref{diagonalizable_operators:eqn:diagonalizable_matrix}. To fix terminology, a \emph{similarity transform} $G \in \mathcal{B}(\Hil)^{-1} = \mathrm{GL}(\Hil)$ is bounded invertible map with bounded inverse. 
\begin{definition}[Diagonalizable operator]
	A bounded operator $H \in \mathcal{B}(\Hil)$ on a Hilbert space is called diagonalizable if there exists a \emph{similarity transform} $G \in \mathcal{B}(\Hil)^{-1}$ for which 
	\begin{align}
		G \, H \, G^{-1} = \int_{\C} E \, \dd \sopro{\psi_E}{\psi_E}
		\label{diagonalizable_operators:eqn:spectral_decomposition}
	\end{align}
	admits a spectral decomposition. 
\end{definition}
The right-hand side of \eqref{diagonalizable_operators:eqn:spectral_decomposition} features what mathematicians refer to as projection-valued measure (\cf Definition~\ref{appendix:functional_calculus_normal_operators:defn:projection_valued_measure}). Physicists are familiar with it as well since it enters the resolution of the identity 
\begin{align}
	\id_{\Hil} &= \int_{\C} \dd \sopro{\psi_E}{\psi_E}
	. 
	\label{diagonalizable_operators:eqn:resolution_identity}
\end{align}
This is in analogy to hermitian operators like position 
\begin{align*}
	\id_{\Hil} &= \int_{\R} \dd x \, \sopro{\psi_x}{\psi_x}
	. 
\end{align*}
Note that the first integral is over the complex plane $\C$ whereas the second is over $\R$. 

While physics text books tend to use sums rather than integrals, equation~\eqref{diagonalizable_operators:eqn:resolution_identity} is a necessary generalization if one wants to accommodate operators on \emph{infinite}-dimensional vector spaces. Even when they are diagonalizable, most do not possess a complete basis of \emph{proper} eigenfunctions, and I have to include “infinitesimal linear combinations” of \emph{generalized} eigenfunctions. In contrast to proper eigenfunctions, these are not elements of the Hilbert space itself. Plane waves $\e^{+ \ii k \cdot x}$ and Bloch waves on $\R^d$ as well as the delta “function” $\delta(y - x)$ are the most prominent examples. But this integral notation also works if the resolution of identity equals a discrete sum, I merely have to insert suitably weighted Dirac measures, \eg 
\begin{align*}
	\sum_{k \in \Z^d} = \int_{\R} \sum_{k \in \Z^d} \delta( \, \cdot \, - k )
	. 
\end{align*}
This difference is not just mathematical nitpickery. Proper eigenstates like are localized bound states whereas generalized eigenstates are delocalized, ionized or scattering states. The hydrogen atom is a good example, the states below the ionization threshold $E = 0$ are discrete bound states, above $E = 0$ one has a continuum of ionized states. So proper and generalized eigenvectors have very different physical behaviors. The distinction between bulk and boundary states may also be understood in these terms: in a periodic system with $(d-1)$-dimensional boundary I can study the spectrum of the operator $H(k_{\parallel})$, where $k_{\parallel}$ is the momentum parallel to the surface. Then delocalized bulk states — Bloch waves — contribute continuous spectrum. Boundary states are associated to eigenvalues of $H(k_{\parallel})$: they are localized near the boundary and typically decay exponentially as I get farther and farther away from the boundary. 

To be able to cope with continuous spectrum, mathematicians connect the validity of equation~\eqref{diagonalizable_operators:eqn:resolution_identity} to the existence of a so-called \emph{projection-valued measure} (\cf Appendix~\ref{appendix:functional_calculus_normal_operators} and \cite[Chapter~3.1]{Teschl:quantum_mechanics:2009}) that gives precise meaning to $\dd \sopro{\psi_E}{\psi_E}$ in integral~\eqref{diagonalizable_operators:eqn:resolution_identity}.

\subsection{The upshot} % (fold)
\label{diagonalizable_operators:the_upshot}
Diagonalizability of operators can be characterized in several equivalent ways: 
\begin{theorem}\label{diagonalizable_operators:thm:characterizations_diagonalizable_operators}
	Let me denote the biorthogonal adjoint with ${}^{\ddagger}$ (\cf equation~\eqref{diagonalizable_operators:eqn:weighted_scalar_product}). The following are equivalent characterizations of diagonalizability: 
	\begin{enumerate}[(1)]
		\item $H$ is diagonalizable. 
		\item There exists a similarity transform $G$ so that $G \, H \, G^{-1}$ is normal with respect to $\dagger$ (\cf Section~\ref{diagonalizable_operators:definition}). 
		\item There exists a similarity transform $G$ so that $G \, H \, G^{-1}$ admits a functional calculus $f \mapsto f \bigl ( G \, H \, G^{-1} \bigr )$, \ie a systematic way to associate an operator $f(H)$ to suitable functions $f : \C \longrightarrow \C$ so that $f \bigl ( G \, H \, G^{-1} \bigr )^{\dagger} = \bar{f} \bigl ( G \, H \, G^{-1} \bigr )$ holds (\cf Appendix~\ref{appendix:functional_calculus_normal_operators} for details). 
		\item $H$ is normal with respect to the biorthogonal scalar product $\scppro{\, \cdot \,}{\, \cdot \,}$ on the vector space $\Hil$, \ie $[H , H^{\ddagger}] = 0$. 
		\item $H$ admits a functional calculus $f \mapsto f(H)$, \ie a systematic way to associate an operator $f(H)$ to suitable functions $f : \C \longrightarrow \C$ so that $f(H)^{\ddagger} = \bar{f}(H)$ holds (\cf Appendix~\ref{appendix:functional_calculus_normal_operators} for details).
	\end{enumerate}
\end{theorem}
Because these five characterizations are mathematically equivalent, I am free to pick any of them to actually \emph{define} diagonalizability; I refer the interested readers to Appendix~\ref{appendix:diabonalizable_operators} for proofs and additional details. 

A central ingredient in the topological classification are carteisan and polar decompositions of diagonalizable operators, which mimic the decompositions of complex numbers 
\begin{align*}
	z = \Re z + \ii \, \Im z = \e^{\ii \vartheta} \, \sabs{z} 
	. 
\end{align*}
into real and imaginary parts, and phase and modulus, respectively. For the latter I need to assume $z \neq 0$ in order to avoid ambiguities in the phase. 
\begin{theorem}[Cartesian and polar decomposition]
	Let me denote the biorthogonal adjoint with ${}^{\ddagger}$ (\cf equation~\eqref{diagonalizable_operators:eqn:weighted_scalar_product}). 
	\begin{enumerate}[(1)]
		\item $H$ is diagonalizable if and only if it is possible to write 
		\begin{align*}
			H = H_{\Re} + \ii H_{\Im} 
			\label{diagonalizable_operators:eqn:cartesian_decomposition}
		\end{align*}
		for two hermitian operators $H_{\Re,\Im} = H_{\Re,\Im}^{\ddagger}$ that \emph{commute}, $[H_{\Re} , H_{\Im}] = 0$. 
		\item $H$ is diagonalizable with bounded inverse if and only if there exist a unitary $V_H$ and a hermitian, strictly positive operator $\sabs{H} = \sabs{H}^{\ddagger}$ so that 
		\begin{align}
			H = V_H \, \sabs{H} 
			. 
			\label{diagonalizable_operators:eqn:polar_decomposition}
		\end{align}
		and the two operators \emph{commute}, $[V_H , \sabs{H}] = 0$. 
	\end{enumerate}
\end{theorem}
Insisting on \eg $[ H_{\Re} \, , \, H_{\Im} ] = 0$ in the cartesian decomposition is absolutely \emph{essential:} there are infinitely many ways to split \emph{any} operator, diagonalizable or not, 
\begin{align*}
	H = H_{\Re} + \ii H_{\Im} 
\end{align*}
into two hermitian operators if I do not insist that real and imaginary part operators commute. Given a scalar product with adjoint ${}^{\dagger}$ I may \emph{always} decompose 
\begin{align}
	H = \tfrac{1}{2} \bigl ( H + H^{\dagger} \bigr ) + \ii \, \tfrac{1}{\ii 2} \bigl ( H - H^{\dagger} \bigr )
	\label{diagonalizable_operators:eqn:naive_decomposition_real_imaginary_parts}
\end{align}
into the sum of two hermitian operators. Unfortunately, the two summands generally fail to commute and hence, \emph{cannot be diagonalized simultaneously.} 

Furthermore, when $[H_{\Re} , H_{\Im}] \neq 0$, I may not conclude from $H_{\Im} \neq 0$ that $H$ has spectrum away from the real line, for example. The Maxwell operator is an explicit example I will discuss in Section~\ref{diagonalizable_operators:choice_of_scalar_product} below: it is hermitian in the biorthogonal scalar product but where $H_{\Im} \neq 0$ if I define the imaginary part with respect to the naïve scalar product. 

\emph{Only} when $H$ is diagonalizable may I pick real and imaginary part operators that \emph{commute}. It turns out if I define $H_{\Re}$ and $H_{\Im}$ with respect to the biorthogonal scalar product~\eqref{diagonalizable_operators:eqn:weighted_scalar_product}, then real and imaginary part operators commute with one another. Consequently, I may simultaneously diagonalize them both. In this way, I may think of diagonalizable operators as “hermitian operators with complex spectrum”, that is normal operators (see Section~\ref{diagonalizable_operators:definition} below). 

Similarly, diagonalizability is not necessary for operators to admit a polar decomposition, there are as many ways to write $H = V_H \, \sabs{H}$ as there are scalar products on $\Hil$. What singles out diagonalizable operators is that I may chose a scalar product so that $V_H$ and $\sabs{H}$ \emph{commute} with one another, 
\begin{align*}
	\bigl [ V_H , \sabs{H} \bigr ] = 0 
	. 
\end{align*}
Once again, this ensures I can simultaneously diagonalize the modulus $\sabs{H}$ and the phase $V_H$. 

\emph{Consequently, from a mathematical point of view the topological classification of diagonalizable, bounded operators with bounded inverses reduces to classifying pairs of commuting, hermitian operators; alternatively, it is equivalent to classifying a commuting pair consisting of a unitary and a positive, hermitian operator.} 
% subsection the_upshot (end)

\subsection{Definition of normal operators} % (fold)
\label{diagonalizable_operators:definition}
An operator $H$ is normal if and only if it commutes with its hermitian adjoint, 
\begin{align*}
	\mbox{$H$ normal}
	\; \; \overset{\mathrm{def}}{\Longleftrightarrow} \; \; 
	[H , H^{\dagger}] = 0 
	. 
\end{align*}
Unitary operators are probably the best-known example of a normal operator that is not hermitian, since indeed the commutator $[U , U^{\dagger}] = \id_{\Hil} - \id_{\Hil} = 0$ always vanishes. 

Normal operators $H = H_{\Re} + \ii H_{\Im}$ can be decomposed into two \emph{hermitian} operators, 
\begin{align*}
	H_{\Re} &= \tfrac{1}{2} \bigl ( H + H^{\dagger} \bigr ) 
	, 
	\\
	H_{\Im} &= \tfrac{1}{\ii 2} \bigl ( H - H^{\dagger} \bigr ) 
	, 
\end{align*}
which commute with one another, $[H_{\Re} , H_{\Im}] = 0$. Hence, $H_{\Re}$ and $H_{\Im}$ can be \emph{diagonalized simultaneously}. Each contributes real and imaginary part to the spectrum of $H$, respectively. Consequently, every normal operator is diagonalizable. Hermitian operators are exactly those diagonalizable operators whose imaginary part $H_{\Im} = 0$ vanishes. 

Similarly, while also non-diagonalizable operators that are bounded and have a bounded inverse possess a polar decomposition~\eqref{diagonalizable_operators:eqn:polar_decomposition}, only when $H$ is normal do the unitary phase and the absolute value \emph{commute}. This is most easily seen when solving~\eqref{diagonalizable_operators:eqn:polar_decomposition} for 
\begin{align*}
	V_H = H \, \sabs{H}^{-1} = H \; \bigl ( H \, H^{\dagger} \bigr )^{-\nicefrac{1}{2}} 
	= \sabs{H}^{-1} \, H 
\end{align*}
and exploiting that $H$ and $H^{\dagger}$ commute. 

In summary, normal operators inherit most of the nice properties that hermitian operators have, the only difference being that their spectra may be complex. Next, I will show the converse is also true, \ie diagonalizable operators are normal with respect to a suitably chosen scalar product. 
% subsection definition_of_normal_operators (end)

\subsection{Mathematically, I have a choice of scalar product} % (fold)
\label{diagonalizable_operators:choice_of_scalar_product}
The previous subsection established that any normal operator is diagonalizable. Now I will explain why converse is also true, although that takes a bit more work. At first glance, it tempting to come up with a counter example, \ie an operator which is diagonalizable but seemingly not normal. In Appendix~\ref{appendix:biorthogonal_calculus_weighted_Hilbert_spaces:2x2_matrix} I walk the readers through the arguments for the $2 \times 2$ matrix 
\begin{align}
	H = \left (
	\begin{matrix}
		1 & -1+\ii \\
		0 & \ii
	\end{matrix}
	\right )
	. 
	\label{diagonalizable_operators:eqn:2_times_2_matrix_example}
\end{align}
$H$ does not commute with its adjoint taken with respect to the Euclidean scalar product. Nevertheless, has two distinct eigenvalues, $1$ and $\ii$, and is therefore diagonalizable. How can this seeming contradiction be resolved? 

Objects like eigenvalues, the spectrum, eigenvectors — and consequently, diagonalizability — are defined through purely \emph{algebraic} relations. Algebraically defined notions are \emph{more fundamental} than \emph{geometric} notions like orthogonality and length. Indeed, I employ different notions of angles and length, depending on my (mathematical or physical) needs. This is accomplished by picking another, perhaps “non-standard” scalar product, \ie scalar products one may not be used to seeing. However, from a mathematical perspective, any sesquilinear map on a vector space that satisfies the axioms of a scalar product (\cf \cite[Chapter~0.3]{Teschl:quantum_mechanics:2009}) \emph{is} a scalar product. 

Mathematically, an operator that is hermitian with respect to a “non-standard” scalar product is not “hermitian” (with quotation marks), but just \emph{hermitian}, period. No scalar product is better than another. Such an operator possesses all the properties I expect hermitian operators to have. For instance, $\e^{- \ii t H}$ is unitary, which leads to conserved quantities such as $\scpro{\e^{- \ii t H} \varphi}{\e^{- \ii t H} \psi} = \scpro{\varphi}{\psi}$. Conserved quantities usually come with a good physical interpretation, so “unusual” scalar products usually have a cogent and straightforward physical interpretation. That extends to normal operators — as long as there exists a scalar product with respect to which $H$ is normal, I can tap into the general theory for normal operators — and that is all that counts here. 

One of the better known examples in theoretical physics comes from classical electromagnetism. The so-called \emph{Maxwell operator} 
\begin{align}
	H = W \, D 
	= \left (
	\begin{matrix}
		\eps^{-1} & 0 \\
		0 & \mu^{-1} \\
	\end{matrix}
	\right ) \left (
	\begin{matrix}
		0 & + \ii \nabla^{\times} \\
		- \ii \nabla^{\times} & 0 \\
	\end{matrix}
	\right )
	% H = \eps^{-1} \, \nabla^{\times} \, \mu^{-1} \, \nabla^{\times}
	\label{diagonalizable_operators:eqn:Maxwell_operator}
\end{align}
arises from rewriting Maxwell's equations describing a dielectric medium (\cf \eg \cite{Raghu_Haldane:quantum_Hall_effect_photonic_crystals:2008} or \cite[Section~3]{DeNittis_Lein:Schroedinger_formalism_classical_waves:2017}). It acts on complex electromagnetic fields fields with square integrable amplitudes. The electric permittivity $\eps$ and the magnetic permeability $\mu$ describe the properties of the medium. This operator is not hermitian if I use the standard scalar product. However, when $\eps$ and $\mu$ take values in the hermitian $3 \times 3$ matrix-valued functions whose eigenvalues are always positive, it \emph{is} hermitian with respect to the scalar product  
\begin{align}
	\scppro{\Psi}{\Phi} &= \scpro{\psi^E}{\eps \, \phi^E} + \scpro{\psi^H}{\mu \, \phi^H}
	\label{diagonalizable_operators:eqn:electromagnetic_energy_scalar_product}
	\\
	&= \int \dd x \, \Bigl ( \psi^E(x) \cdot \eps(x) \phi^E(x) 
	\, + \Bigr . 
	\notag \\
	&\qquad \qquad \; \Bigl . 
	+ \, \psi^H(x) \cdot \mu(x) \phi^H(x) \Bigr )
	, 
	\notag 
\end{align}
because it satisfies $\scppro{\Psi}{H \Phi} = \scppro{H \Psi}{\Phi}$ or $H^{\ddagger} = H$ for short. Complex conjugation is implicit in the dot product on $\C^3$. Mathematically, $H^{\ddagger} = H$ is just as good as $H^{\dagger} = H$: the spectrum of $H$ is real, eigenfunctions to different eigenvalues will be $\scppro{\, \cdot \,}{\, \cdot \,}$-orthogonal and so forth \footnote{Strictly speaking, I have not taken the unboundedness of $H$ into account, but this is not essential for my arguments here. Fortunately, I have already paid this debt in an earlier publication and shown that not only the operating prescription, but also the domains of $H$ and $H^{\ddagger}$ coincide (\cf \cite[Proposition~6.2]{DeNittis_Lein:Schroedinger_formalism_classical_waves:2017}). }. And the conserved quantity $\scppro{\Psi}{\Psi}$ is nothing but twice the electromagnetic field energy. So this “unusual, non-standard” scalar product~\eqref{diagonalizable_operators:eqn:electromagnetic_energy_scalar_product} has neat physical interpretation. 

The example of the Maxwell operator also teaches us that calling \eg 
\begin{align*}
	H_{\Im} &= \frac{1}{\ii 2} \bigl ( H - H^{\dagger} \bigr ) 
	= \frac{1}{\ii 2} \bigl ( W \, D - D \, W \bigr ) 
	\neq 0 
\end{align*}
the imaginary part of the Maxwell operator makes no sense: since the Maxwell operator is hermitian with respect to the energy scalar product~\eqref{diagonalizable_operators:eqn:electromagnetic_energy_scalar_product}, its spectrum is real and the imaginary part should vanish. The interpretation of $H_{\Re}$ and $H_{\Im}$ as real and imaginary parts of $H$ hinge on the condition $[ H_{\Re} , H_{\Im} ] = 0$ that they commute. 

Both for pragmatic and conceptual reasons, it is better to study the Maxwell operator~\eqref{diagonalizable_operators:eqn:Maxwell_operator} in its “natural” Hilbert space, \ie using the scalar product~\eqref{diagonalizable_operators:eqn:electromagnetic_energy_scalar_product} that makes it hermitian. More generally, the same pragmatic reasons apply for arbitrary diagonalizable operators: since I can always \emph{construct} a scalar product $\scppro{\, \cdot \,}{\, \cdot \,}$ with respect to which the operator is normal, I can tap into the general theory of hermitian and normal operators. 

Let me illustrate the construction with a simple example first. In case of the $2 \times 2$ matrix example from earlier, equation~\eqref{diagonalizable_operators:eqn:2_times_2_matrix_example}, the two eigenvectors are $v_1 = (1 , 0)^{\mathrm{T}}$ and $v_2 = (1 , 1)^{\mathrm{T}}$. Therefore, I can write 
\begin{align*}
	H &= G^{-1} \, D \, G
	= \left (
	\begin{matrix}
		1 & 1 \\
		0 & 1 \\
	\end{matrix}
	\right ) \, \left (
	\begin{matrix}
		1 & 0 \\
		0 & \ii \\
	\end{matrix}
	\right ) \, \left (
	\begin{matrix}
		1 & -1 \\
		0 & 1 \\
	\end{matrix}
	\right )
\end{align*}
as the product of the invertible matrix $G$ and a diagonal matrix with the two eigenvalues in its diagonal. Picking 
\begin{align}
	\scppro{\varphi}{\psi} \overset{\mathrm{def}}{=} \scpro{G \varphi \, }{ \, G \psi}
	\label{diagonalizable_operators:eqn:weighted_scalar_product}
\end{align}
not only declares the two eigenvectors to be orthonormal, the adjoint matrix 
\begin{align}
	H^{\ddagger} &= (G^{\dagger} G)^{-1} \, H^{\dagger} \, (G^{\dagger} G)
	\label{diagonalizable_operators:eqn:G_adjoint}
	\\
	&= G^{-1} \, \bigl ( G \, H \, G^{-1} \bigr )^{\dagger} \, G
	\notag \\
	&= G^{-1}  \, \overline{D} \, G
	\notag 
\end{align}
commutes with $H$ as claimed. We refer to the Appendix~\ref{appendix:biorthogonal_calculus_weighted_Hilbert_spaces:2x2_matrix} for details. Note that I could have chosen \eg $v_1 = (20 , 0)^{\mathrm{T}}$ and $v_2 = (4 , 4)^{\mathrm{T}}$ instead, which measures lengths in the two directions with respect to “different units”. 

Very often it will be useful to express the scalar product 
\begin{align*}
	\sscppro{\varphi}{\psi} &= \scpro{\varphi \, }{ \, G^{\dagger} G \psi}
	= \scpro{\varphi}{W \psi}
\end{align*}
in terms of the weight 
\begin{align*}
	W = G^{\dagger} G
	. 
\end{align*}
Not only does that make some formulas more succinct, it becomes clear that I may replace $G$ with $\sqrt{W}$. By my assumptions on $G$, the operator $W$ is automatically a \emph{strictly positive} similarity transform, \ie $W$ is positive and bounded, and its inverse $W^{-1}$ exists and is bounded as well. 

Clearly, these arguments immediately extend to diagonalizable $N \times N$ matrices. When I am dealing with diagonalizable operators on infinite-dimensional Hilbert spaces, the arguments become more technical. Yet in essence, it still follows the same basic outline; I refer the interested readers to Appendix~\ref{apprendix:characterizations_diagonalizability}. 
% subsection choice_of_scalar_product (end)

\subsection{Equivalence to biorthongonal formalism} % (fold)
\label{diagonalizable_operators:biorthongonal_formalism}
The so-called biorthogonal formalism is a standard tool in the physics community when dealing with non-hermitian operators \cite{Brody:biorthogonal_quantum_mechanics:2014}. However, I feel the way it is typically presented obscures its mathematical underpinnings and makes it harder to exploit general mathematical facts to their fullest extent. One of the main points of this article, that all diagonalizable operators are normal with respect to a suitably chosen scalar product, is a mathematical triviality, but is obscured by the notation and terminology frequently used in much of the physics community. 

For example, the standard presentation does not make it obvious that the “unusual” scalar product~\eqref{diagonalizable_operators:eqn:electromagnetic_energy_scalar_product} is nothing but the scalar product from the biorthogonal formalism. The language used in many physics publications falsely suggests a hierarchy of scalar products where the “standard” or “quantum mechanical” scalar product is more fundamental than \eg the biorthogonal scalar product. Mathematically, this is false: as long as a sesquilinear form $\scpro{\, \cdot \,}{\, \cdot \,}$ satisfies all the axioms of a scalar product, it \emph{is} a scalar product. And one scalar product is as good as any other. The same applies to all derived notions like hermiticity or unitarity: for instance, Brody puts hermitian and unitary in quotation marks when these are defined with respect to the biorthogonal scalar product \cite{Brody:biorthogonal_quantum_mechanics:2014}. This is misleading and unnecessary, for mathematically these operators \emph{truly are} hermitian and unitary and not hermitian and unitary in some second-class sense. 

Of course, physics frequently \emph{does} single out one scalar product by giving it a cogent physical interpretation — although that need not always be the “standard” scalar product. For instance, when people refer to the standard scalar product as the “quantum mechanical” scalar product, what they actually mean is that this is the scalar product with which I compute transition probabilities. Hermiticity of the Hamiltonian $H = H^{\dagger}$ then leads to unitarity of the time evolution $\e^{- \ii \frac{t}{\hbar} H}$, and the unitarity in turn implies conservation of probability. But many applications make no reference to quantum mechanics: when expressing Maxwell's equations in the form of a Schrödinger equation via \eqref{diagonalizable_operators:eqn:Maxwell_operator}, then the standard scalar product has no particular physical meaning. Instead, it is the “non-standard” biorthogonal scalar product~\eqref{diagonalizable_operators:eqn:electromagnetic_energy_scalar_product} that is of physical significance: $\scppro{(\mathbf{E},\mathbf{H})}{(\mathbf{E},\mathbf{H})}$ is twice the electromagnetic field energy. Indeed, the “biorthogonal” scalar product~\eqref{diagonalizable_operators:eqn:weighted_scalar_product} may very well be the one that carries physical significance rather than the originally given scalar product $\scpro{\, \cdot \,}{\, \cdot \,}$. 

That being said, let me recap the standard recipe of the biorthogonal formalism and connect it to the previous section. To free us of mathematical technicalities, let me suppose that $H$ is not just diagonalizable but in addition possesses a complete basis made up of proper (right-)\linebreak eigenvectors. The idea of the biorthogonal formalism is that in addition to the set of right-eigenkets obtained from the eigenvalue equation 
\begin{align*}
	H \sket{\psi_{R,n}} = E_n \, \sket{\psi_{R,n}}
	, 
\end{align*}
there exists a set of left-eigenbras, which solve 
\begin{align*}
	\sbra{\psi_{L,n}} H = E_n \, \sbra{\psi_{L,n}} 
	. 
\end{align*}
Denoting the antilinear duality between bras and kets also with $\dagger$, I get the ordinary eigenvalue equation 
\begin{align*}
	H^{\dagger} \sket{\psi_{L,n}} = \overline{E_n} \; \sket{\psi_{L,n}}
\end{align*}
for the left-eigenvectors of the adjoint operator to the complex conjugate eigenvalues. In principle, I can normalize left- and right-eigenvectors independently, but I may choose to normalize them 
\begin{align*}
	\scpro{\psi_{L,n}}{\psi_{R,n}} = 1 
\end{align*}
by convention (see the discussion of \cite[equation~(19)]{Brody:biorthogonal_quantum_mechanics:2014}). 

The right-eigenvectors of $H$ then make up the “columns” of $G^{-1}$ just like in the $2 \times 2$ matrix example from Section~\ref{diagonalizable_operators:choice_of_scalar_product}. Similarly, the “column vectors” of $G^{\dagger}$ are the eigenvectors of the adjoint matrix $H^{\dagger}$ to the complex conjugate eigenvalues, \ie the left-eigenvectors.

What is the relation between left- and right-eigenvectors then? Well, for both, left- and right-eigenvectors I have 
\begin{align*}
	G \psi_{R,n} &= e_n 
	= \bigl ( G^{\dagger} \bigr )^{-1} \psi_{L,n}
	, 
\end{align*}
where $\{ e_n \}_n$ is any reference basis of $\Hil$ that is orthonormal with respect to the original scalar product $\scpro{\, \cdot \,}{\, \cdot \,}$. That means the two types of eigenvectors are related by $G G^{\dagger}$, 
\begin{align*}
	\psi_{L,n} = G^{\dagger} G \psi_{R,n}
	\; \; \Longleftrightarrow \; \; 
	\psi_{R,n} = \bigl ( G^{\dagger} G \bigr )^{-1} \psi_{L,n}
	. 
\end{align*}
Plugging this in reveals that the inner product of right- and left-eigenvalue is nothing but the $\scppro{\, \cdot \,}{\, \cdot \,}$-scalar product from equation~\eqref{diagonalizable_operators:eqn:weighted_scalar_product} of two left-eigenvectors,  
\begin{align}
	\scpro{\psi_{L,j}}{\psi_{R,n}} &= \scpro{G^{\dagger} \, G \psi_{R,j} \, }{ \, \psi_{R,n}}
	\notag \\
	&= \scpro{\psi_{R,j} \, }{ \, G^{\dagger} \, G \psi_{R,n}}
	\notag \\
	&= \scpro{G \psi_{R,j}}{G \psi_{R,n}} 
	= \bscppro{\psi_{R,j}}{\psi_{R,n}} 
	. 
	\label{diagonalizable_operators:eqn:relation_two_scalar_products}
\end{align}
Moreover, it also explains the orthonormality of left- and right-eigenvectors, 
\begin{align*}
	\scpro{\psi_{L,j}}{\psi_{R,n}} &= \bscppro{\psi_{R,j}}{\psi_{R,n}} 
	= \scpro{G \psi_{R,j} \, }{ \, G \psi_{R,n}} 
	\\
	&= \sscpro{e_j}{e_n}
	= \delta_{jn}
	, 
\end{align*}
which corresponds to \cite[equation~(18)]{Brody:biorthogonal_quantum_mechanics:2014} for the special case $d_k = \delta_{jk}$ and $c_k = \delta_{nk}$. 

Then the completeness relation
\begin{align*}
	\id_{\Hil} &= \sum_n \sopro{\psi_{R,n}}{\psi_{L,n}}
	\\
	&= \sum_n \sopro{\psi_{R,n}}{G^{\dagger} G \psi_{R,n}}
	\\
	&= \sum_n \skket{\psi_{R,n}} \sbrbra{\psi_{R,n}}
\end{align*}
can be equivalently written in a symmetric way with the $\sscppro{\, \cdot \,}{\, \cdot \,}$-scalar product~\eqref{diagonalizable_operators:eqn:weighted_scalar_product}. Alternatively, I may view it as a fancy way of writing 
\begin{align*}
	\bigl ( G^{\dagger} \bigr )^{\dagger} \, G^{-1} = G \, G^{-1} = \id_{\Hil}
	. 
\end{align*}
These arguments extend to the case when the spectrum does not consist solely of eigenvalues; I refer the interested readers to Appendix~\ref{apprendix:characterizations_diagonalizability}. 

To summarize, from the perspective of mathematics, the biorthogonal calculus is just a very cumbersome way of using the adapted $\sscppro{\, \cdot \,}{\, \cdot \,}$-scalar product~\eqref{diagonalizable_operators:eqn:weighted_scalar_product}. Moreover, the biorthogonal calculus obscures certain fundamental facts about the setting. For example, it is not necessary to keep track of two sets of eigenvectors, right-eigenvectors contain all the information. And not least it obscures that diagonalizable operators are exactly those that are $\sscppro{\, \cdot \,}{\, \cdot \,}$-normal, \ie normal with respect to the scalar product~\eqref{diagonalizable_operators:eqn:weighted_scalar_product}. 
% subsection equivalence_to_biorthongonal_formalism (end)

\subsection[The topological classification does not depend on the choice of scalar product]{The topological classification does not \linebreak depend on the choice of scalar product} % (fold)
\label{diagonalizable_operators:topological_classification_algebraic}
So far all of the arguments suggest the topological classification is independent of the geometry and \emph{only depends on algebraic relations.} And this will be indeed the case: objects like the spectrum, eigenvectors, inverses and so forth do not depend on my choice of scalar product. 

Only one subtlety should be briefly addressed: I mentioned that geometry gives meaning to the notion of length via the norm $\norm{\psi} = \sqrt{\scpro{\psi}{\psi}}$ on $\Hil$, and length is used to measure the distance between two operators, 
\begin{align}
	\snorm{H} \overset{\mathrm{def}}{=} \sup_{\snorm{\varphi} = 1} \snorm{H \varphi}
	. 
	\label{diagonalizable_operators:eqn:operator_norm}
\end{align}
This, in turn, is a crucial ingredient when defining what continuity means in the space of operators. And given that homotopies are maps that interpolate between two operators in a \emph{continuous} fashion, continuity — and hence, geometry — \emph{does} enter the topological classification. 

Fortunately, though, the notions of continuity that arise from the originally given scalar product $\scpro{\, \cdot \,}{\, \cdot \,}$ on $\Hil$ and the adapted scalar product $\sscppro{\, \cdot \,}{\, \cdot \,}$ (the scalar product from biorthogonal calculus) are one and the same. That is because the two norms are equivalent, 
\begin{align*}
	\norm{G^{-1}}^{-2} \, \scpro{\psi}{\psi} \leq \scppro{\psi}{\psi} \leq \norm{G}^2 \, \scpro{\psi}{\psi} 
	, 
\end{align*}
where $\norm{G}$ denotes the operator norm~\eqref{diagonalizable_operators:eqn:operator_norm} of $G$. So indeed, I am free to use the adapted scalar product $\sscppro{\, \cdot \,}{\, \cdot \,}$ without altering my definition of topological phase. 
% subsection topological_classification_is_algebraic_rather_than_geometric (end)
% section diagonalizable_operators_are_normal_operators (end)
%!TEX root = /Users/max/Dropbox/research/topological classification non-selfadjoint hamiltonians/paper/non_hermitian_classification.tex
\section{Review of the 38-fold classification of \cite{Kawabata_Shiozaki_Ueda_Sato:classification_non_hermitian_systems:2019,Zhou_Lee:non_hermitian_topological_classification:2019}} % (fold)
\label{38_fold_recap}
For the benefit of the reader and to give me the opportunity to introduce some basic notions and notation, I will summarize the main points of the 38-fold classification of non-hermitian operators.

\subsection{The homotopy definition of topological phases} % (fold)
\label{38_fold_recap:homotopy_definition}
At the very basis of most topological classifications is the homotopy definition of topological phases. The idea is that the “topology of my system” must not change under continuous deformations as long as the relevant spectral gap remains open and all essential symmetries are preserved. Let me call the set of operators with certain properties (\eg hermiticity) that possess the relevant symmetries and have the right type of spectral gap $\mathcal{X}$; the types of symmetries and the precise nature of the spectral gap will be introduced below. Mathematically, this translates to considering two operators $H_0$ and $H_1$ equivalent if there exists a continuous path $H(\lambda)$ in $\mathcal{X}$ that connects $H_0 = H(0)$ with $H_1 = H(1)$; such a path is called a \emph{homotopy}, and $H_0$ and $H_1$ are considered \emph{homotopically equivalent.} In this mathematical dialect a topological phase is a \emph{homotopy equivalence class of operators,} \ie I identify all operators connected by a homotopy. Put another way, the topological phases make up the \emph{set} of connected components $\pi_0(\mathcal{X})$ inside the set of operators $\mathcal{X}$. Two operators can only be in different topological phases if there is some barrier to them being connected by a continuous path. For the purpose of \cite{Kawabata_Shiozaki_Ueda_Sato:classification_non_hermitian_systems:2019,Zhou_Lee:non_hermitian_topological_classification:2019} the barriers are the regions where the spectral gap closes. Note that there is no natural group structure on $\pi_0(\mathcal{X})$ (\cf \cite[Section~2]{Thiang:K_theoretic_classification_topological_insulators:2016}); most approaches to classifying operators amount to establishing relations between elements of the \emph{set} $\pi_0(\mathcal{X})$ and certain groups (such as $K$-groups). 

When the operators are periodic or depend on other, suitable parameters, homotopy \emph{groups} will enter the discussion; a recent preprint \cite{Wojcik_Sun_Bzdusek_Fan:topological_classification_non_hermitian_hamiltonians:2020} starts with this premise and develops a classification entirely from homotopy theory. Unfortunately, homotopy groups are in general not \emph{algorithmically} computable, and mathematicians had had to find other ways to characterize and distinguish topological phases. 

One way out is $K$-theory, which associates \emph{abelian groups} to topological spaces. And these groups capture some essential topological features of these topological spaces. Its group elements act as \emph{labels} for the topological phases of $\mathcal{X}$, \ie group elements are house numbers or coordinates for the \emph{set} of path-connected components $\pi_0(\mathcal{X})$. To give one example, for periodic hermitian operators of class~A in $d = 4$, the six first Chern numbers, the second Chern number and the rank label the phase of the system \emph{uniquely}; typically, the rank (the number of filled bands) is disregarded, though, as it is being kept fixed. In the context of topological insulators, these groups are typically products and sums of $\Z$ and $\Z_2$. The tremendous advantage of $K$-groups is that these are \emph{algorithmically} computable, \ie I can (in principle) write a computer program that spits out the $K$-group once I give it a CW complex $\mathcal{X}$ (like the Brillouin torus $\T^d$). 

The $K$-theoretic approach has its shortcomings. While none of these shortcomings take away from the success of applying $K$-theory to problems from topological insulators, it is nevertheless useful to keep them in mind. For one, \emph{a priori} it is not clear whether the list of invariants is (or even can be!) exhaustive. Many works only list the so-called strong or top invariants; if there are other weak invariants, then the strong invariants are not enough to uniquely label topological phases. To continue my example of hermitian class~A operators in $d = 4$, the strong invariant is the second Chern number whereas the weak invariants are the six first Chern numbers. So while two operators with different second Chern numbers must lie in different topological phases, just because their second Chern numbers agree does not automatically mean they are homotopic. Even in the hermitian case and only for low dimension do we have proofs that the list of topological invariants is complete; I am currently aware only of proofs for classes~A, AI, AII and AIII \cite{Hatcher:vector_bundles_K_theory:2009,DeNittis_Gomi:AI_bundles:2014,DeNittis_Gomi:AII_bundles:2014,DeNittis_Gomi:AIII_bundles:2015}. Even for the best-understood class, class~A, there are cases when knowing all topological invariants is not enough; an explicit example is constructed in \cite[Section~V.G]{DeNittis_Lein:exponentially_loc_Wannier:2011} for $d = 5$ and rank~$2$. 

Lastly, the recent preprint \cite{Wojcik_Sun_Bzdusek_Fan:topological_classification_non_hermitian_hamiltonians:2020} emphasizes another relevant point: first homotopy groups can be non-abelian. Specifically, Wojcik et al.\ show explicitly that braid groups appear in the classification of certain classes of non-hermitian operators. $K$-groups, on the other hand, are always abelian, so they are unable to resolve some of these finer details, which may be important to properly understand the physics. 
% subsection homotopy_definition (end)

\subsection{Relevant symmetries} % (fold)
\label{38_fold_recap:symmetries}
In a nutshell, Kawabata et al.\ \cite{Kawabata_Shiozaki_Ueda_Sato:classification_non_hermitian_systems:2019} have computed the relevant $K$-groups for gapped non-hermitian operators with particular types of symmetries. I shall make a list of them now. They fall into two distinct classes, either they relate $H$ with itself, 
\begin{align}
	U \, H \, U^{-1} &= \pm H
	, 
	\label{symmetries:eqn:usual_CAZ_symmetries}
\end{align}
or $H$ with its adjoint $H^{\dagger}$, 
\begin{align}
	U_{\dagger} \, H \, U_{\dagger}^{-1} &= \pm H^{\dagger}
	. 
	\label{symmetries:eqn:dagger_CAZ_symmetries}
\end{align}
$U$ and $U_{\dagger}$ (abbreviated as $U_{(\dagger)}$ in what follows) are either a linear or an antilinear. Antilinear maps come in the even or odd variety, depending on whether $U_{(\dagger)}^2 = \pm \id$, whereas linear ones are always assumed to square to $+\id$. Following the convention of \cite{Kawabata_Shiozaki_Ueda_Sato:classification_non_hermitian_systems:2019}, this will either give rise to ordinary, chiral, time-reversal and particle-hole symmetries or their daggered counterparts. I have listed them in Table~\ref{symmetries:table:overview_symmetries}. 

\begin{table*}
	\begin{centering}
		\newcolumntype{A}{>{\centering\arraybackslash\normalsize}m{18mm}}
		\newcolumntype{B}{>{\centering\arraybackslash\normalsize}m{20mm}}
		\newcolumntype{C}{>{\centering\arraybackslash\normalsize}m{30mm}}
		\newcolumntype{D}{>{\centering\arraybackslash\normalsize}m{15mm}}
		\renewcommand{\arraystretch}{1.5}
		\begin{tabular}{A | A | B | C | D | A}
			\multicolumn{2}{c|}{\normalsize\emph{Notation used}} & \emph{(Anti)linear} & \emph{Condition on $H$} & $U^2 =$  & $\sigma(H) = $ \\
			\emph{here} & \emph{in \cite{Kawabata_Shiozaki_Ueda_Sato:classification_non_hermitian_systems:2019}} &  &  &  & \\ \hline \hline 
			ordinary & ordinary & linear & $V \, H \, V^{-1} = + H$ & $+\id$ & $+\sigma(H)$ \\ \hline 
			chiral   & $\mathrm{CS}^{\dagger} = \mathrm{SLS}$ & linear & $S \, H \, S^{-1} = - H$ & $+\id$ & $-\sigma(H)$ \\ \hline 
			$\pm \mathrm{TR}$  & $\pm \mathrm{TRS}$ & antilinear & $T \, H \, T^{-1} = + H$ & $\pm \id$ & $+ \overline{\sigma(H)}$ \\ \hline 
			$\pm \mathrm{PH}$  & $\pm \mathrm{PHS}^{\dagger}$ & antilinear & $C \, H \, C^{-1} = - H$ & $\pm \id$ & $- \overline{\sigma(H)}$ \\ \hline \hline 
			pseudo   & $\mathrm{pH}$ & linear & $V_{\dagger} \, H \, V_{\dagger}^{-1} = + H^{\dagger}$ & $+\id$ & $+\overline{\sigma(H)}$ \\ \hline 
			chiral${}^{\dagger}$ & $\mathrm{CS}$ & linear & $S_{\dagger} \, H \, S_{\dagger}^{-1} = - H^{\dagger}$ & $+\id$ & $-\overline{\sigma(H)}$ \\ \hline 
			$\pm \mathrm{TR}^{\dagger}$ & $\pm \mathrm{TRS}^{\dagger}$ & antilinear & $T_{\dagger} \, H \, T_{\dagger}^{-1} = + H^{\dagger}$ & $\pm \id$ & $+ \sigma(H)$ \\ \hline 
			$\pm \mathrm{PH}^{\dagger}$ & $\pm \mathrm{PHS}$ & antilinear & $C_{\dagger} \, H \, C_{\dagger}^{-1} = - H^{\dagger}$ & $\pm \id$ & $- \sigma(H)$ \\ 
		\end{tabular}
	\end{centering}
	\caption{This table lists the types of symmetries considered for this classification. The two naming schemes are compared in the first two columns. The presence of discrete symmetries of the form \eqref{symmetries:eqn:usual_CAZ_symmetries} or \eqref{symmetries:eqn:dagger_CAZ_symmetries} leads to symmetries in the spectrum. }
	\label{symmetries:table:overview_symmetries}
\end{table*}

My labeling convention is completely equivalent to the one adopted in \eg \cite{Kawabata_Shiozaki_Ueda_Sato:classification_non_hermitian_systems:2019}, who factor out complex conjugation $K$. In my notation, a time-reversal symmetry is an antiunitary $T = \widetilde{T} \, K$ that commutes with $H$. Kawabata et al.\ would instead focus on the unitary operator $\widetilde{T}$ as the time-reversal symmetry since it satisfies 
\begin{align*}
	\widetilde{T} \, H \, \widetilde{T}^{-1} = + \overline{H}
	, 
\end{align*}
where $\overline{H} = K \, H \, K$ is the complex conjugate of $H$. Similarly, using that the transpose 
\begin{align*}
	H^{\mathrm{T}} \overset{\mathrm{def}}{=} \overline{H}^{\dagger}
\end{align*}
of an operator is \emph{defined} as the adjoint of the complex conjugate operator, I can translate the definitions of Kawabata et al.\ and compare them with mine. The result is summarized in the first two columns of Table~\ref{symmetries:table:overview_symmetries}. Even though my labeling convention to extend the symmetries differs from \cite{Kawabata_Shiozaki_Ueda_Sato:classification_non_hermitian_systems:2019}, both choices are logically consistent: while my notation is guided by mathematical simplicity, Kawabata et al.\ motivate their choices by physics. 

The presence of these symmetries leads to symmetries in the spectrum of the operator 
\begin{align}
	\sigma(H) \overset{\mathrm{def}}{=} \bigl \{ E \in \C \; \; \vert \; \; \mbox{$H - E$ not invertible} \bigr \} 
	, 
	\label{38_fold_recap:eqn:definition_spectrum}
\end{align}
\eg those visible in Figures~\ref{intro:figure:highly_symmetric_spectrum} or \ref{38_fold_recap:figure:nested_Cs_spectrum}. 

For example, assume $H$ possesses a chiral${}^{\dagger}$ symmetry, \ie a unitary $S_{\dagger}$ that satisfies 
\begin{align*}
	S_{\dagger} \, H \, S_{\dagger}^{-1} &= - H^{\dagger} 
	. 
\end{align*}
Suppose $\psi_E$ is an eigenvector to the complex eigenvalue $E = \lambda + \ii \mu \in \sigma(H)$. Then a straightforward computation, 
\begin{align*}
	S_{\dagger} \, H \psi_E &= S_{\dagger} \bigl ( E \, \psi_E \bigr ) 
	= E \, S_{\dagger} \psi_E
	\\
	&\overset{!}{=} - H^{\dagger} \, S_{\dagger} \psi_E 
	, 
\end{align*}
shows us that $-E = - \lambda - \ii \mu$ lies in the spectrum of the adjoint $H^{\dagger}$. While not all energies from the spectrum must correspond to eigenstates, these arguments can be made rigorous with the help of Weyl sequences, \ie sequences of approximate eigenvectors (\cf \cite[Lemma~2.1.6]{Teschl:quantum_mechanics:2009}). Given that the spectra $\sigma(H^{\dagger}) = \overline{\sigma(H)}$ of $H$ and $H^{\dagger}$ are related by complex conjugation (see \cite[Theorem~VI.7]{Reed_Simon:M_cap_Phi_1:1972}), this leads to a symmetry in the spectrum of $H$ itself, \ie $\sigma(H) = - \overline{\sigma(H)}$. Namely, whenever $E \in \sigma(H)$, then also 
\begin{align*}
	- \overline{E} = - \lambda + \ii \mu \in \sigma(H)
\end{align*}
must lie in the spectrum. Visually, the spectrum is symmetric by reflection about the imaginary axis. 

Repeating this argument 7 more times gives all the other cases; I have listed them all in the last column of Table~\ref{symmetries:table:overview_symmetries}. All symmetries come in “spectral symmetry pairs”, \eg the presence of chiral and $\pm \mathrm{PH}^{\dagger}$ symmetries lead to the same symmetry in the spectrum $\sigma(H)$. 
% subsection the_relevant_symmetries (end)

\subsection{Point gap vs.\ line gaps} % (fold)
\label{38_fold_recap:point_gap_vs_line_gaps}
One of the central points of \cite{Kawabata_Shiozaki_Ueda_Sato:classification_non_hermitian_systems:2019} is that the classification crucially depends on the type of spectral gap that one chooses to preserve during deformations. This becomes necessary, because non-hermitian operators may have complex spectrum. Deformations of operators may move the spectrum in the complex plane $\C \simeq \R^2$, \ie a space with two real dimensions. So I can think of several types of obstacles, $0$- and $1$-dimensional obstacles (with codimensions $2$ and $1$, respectively). 

In contrast, hermitian operators have real spectrum, the real line $\R$ is $1$-dimensional and gaps can only be $0$-dimensional (codimension $1$). Nevertheless, in both (hermitian \cite{Altland_Zirnbauer:superconductors_symmetries:1997} and non-hermitian \cite{Kawabata_Shiozaki_Ueda_Sato:classification_non_hermitian_systems:2019,Zhou_Lee:non_hermitian_topological_classification:2019}) formalisms a $1$-codimensional barrier appears in the classification of hermitian operators. 
\begin{figure}
	\begin{centering}
		\resizebox{70mm}{!}{\includegraphics{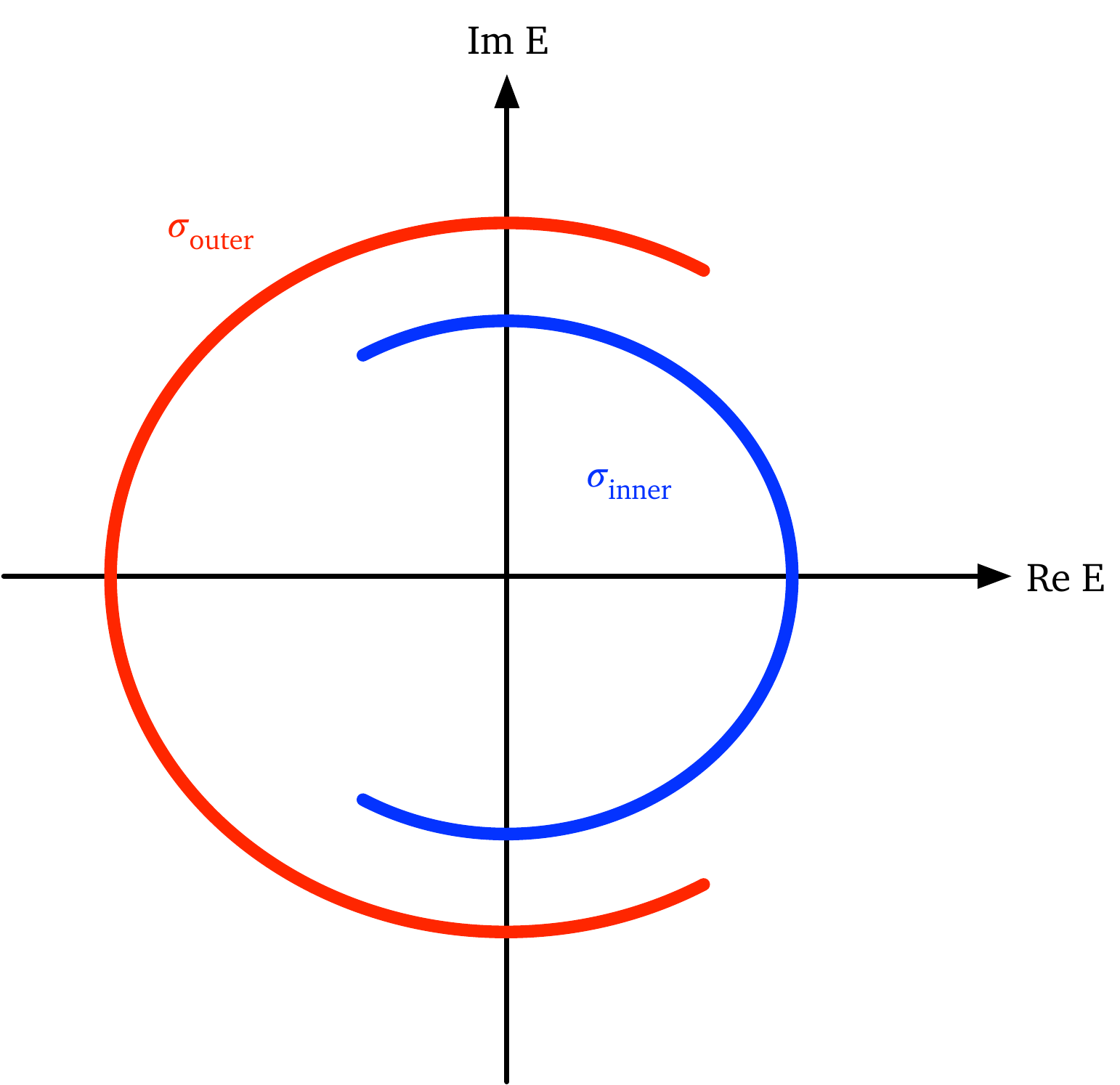}}
	\end{centering}
	\caption{Spectrum with reflection symmetries about the real axis. While this spectrum consists of two distinct subsets that are separated by a gap, unlike the spectrum in Figure~\ref{intro:figure:highly_symmetric_spectrum}, this spectrum does not possess a real line gap in the sense of \cite{Kawabata_Shiozaki_Ueda_Sato:classification_non_hermitian_systems:2019}. Nevertheless, my construction works and the topological classification corresponds to that of a real line gap. }
	\label{38_fold_recap:figure:nested_Cs_spectrum}
\end{figure}

\subsubsection{Point gap} % (fold)
\label{38_fold_recap:point_gap_vs_line_gaps:point_gap}
First of, they assume that the relevant spectral gap always includes the point $E = 0$. And indeed, the simplest type of gap is a so-called \emph{point gap}, which assumes that I can draw a disc of positive radius around $E = 0$ and not meet any spectrum of $H$. For example, the spectra in Figures~\ref{intro:figure:highly_symmetric_spectrum} and \ref{38_fold_recap:figure:nested_Cs_spectrum} possess point gaps. Here, the gap condition means I am imposing that during deformations $H$ remains a bounded operator with \emph{bounded} inverse. 

For operators with point gaps, Kawabata et al.\ exploit the polar decomposition of operators 
\begin{align}
	H = V_H \, \sabs{H} 
	, 
	\label{38_fold_recap:eqn:polar_decomposition}
\end{align}
where $V_H$ is unitary “phase” and $\sabs{H} = \sqrt{H \, H^{\dagger}}$ is the hermitian absolute value. They exploit that invertibility and unitarity are topologically speaking equivalent (read: homotopic) in many contexts, including \eg Kuiper's Theorem \cite{Kuiper:U_infty_contractible:1965}. So rather than study subgroups of the general linear group $\mathrm{GL}(\Hil)$ composed of bounded invertible operators with bounded inverse, it suffices to study subgroups of the unitary operators $\mathrm{U}(\Hil)$. 

Figure~2~(b) in \cite{Kawabata_Shiozaki_Ueda_Sato:classification_non_hermitian_systems:2019} nicely illustrates how to graphically obtain the spectrum of the unitary from the spectrum of the original operator. Note that unitary operators are normal, \ie $U$ commutes with $U^{\dagger}$, so they behave much like hermitian operators and can always be diagonalized; and their spectrum is always a subset of the unit circle $\Sone \subseteq \C$. Nevertheless, since the polar decomposition is not unique (\cf the discussion in Section~\ref{diagonalizable_operators:the_upshot}), this “graphical construction” is not telling the whole story and could lead to issues when $H$ is not diagonalizable. 
% subsubsection point_gap (end)

\subsubsection{Line gaps} % (fold)
\label{38_fold_recap:point_gap_vs_line_gaps:line_gaps}
The other type of gap that Kawabata et al.\ introduce are line gaps, where an infinite line drawn through the origin separates two spectral regions from one another. The presence of a line gap forbids rotations of spectrum in the complex plane, each spectral region is confined to its side of the dividing line. Real and imaginary line gaps are of particular importance, \ie the cases where I can choose the imaginary and real axis (the order is reversed!) as the dividing line. The operator with spectrum given by Figure~\ref{intro:figure:highly_symmetric_spectrum} has a real \emph{and} imaginary line gap. In contrast, the operator from Figure~\ref{38_fold_recap:figure:nested_Cs_spectrum} has only a point gap but in Kawabata et al.'s definition neither a real nor an imaginary line gap. 

Operators with line gaps can be deformed — spectrally flattened — into hermitian or antihermitian operators where the spectrum on either side of the dividing line eventually coalesces at either $\pm 1$ or $\pm \ii$. These spectrally flattened hamiltonians $Q$ can then be considered as a grading on the Hilbert space $\Hil = \mathcal{E}_+ \oplus \mathcal{E}_-$ with additional symmetries, where $\mathcal{E}_{\pm} = \mathrm{Eig}(Q,\pm 1)$ or $\mathcal{E}_{\pm} = \mathrm{Eig}(Q,\pm \ii)$ (depending on the line gap type). 

The reason why real and imaginary line gap classifications differ from one another is that the relevant symmetries will either respect the line gap or break it. For example, consider an operator with a time-reversal symmetry $T$, a particle-hole symmetry $C$ and a chiral symmetry $S$ that possesses a real line gap. Indeed, it may have the spectrum pictured in Figure~\ref{intro:figure:highly_symmetric_spectrum}. The time-reversal symmetry only flips the sign of the imaginary part of the eigenvalue, so it maps the spectral region to the right of the imaginary axis (where $\Re E > 0$) onto itself. Particle-hole and chiral symmetries, though, map states with $\Re E > 0$ onto states with $\Re E < 0$. Had I chosen the imaginary line gap (the real axis $\Im E = 0$), then time-reversal and particle-hole symmetries have opposite behaviors — $T$ maps $\Im E > 0$ onto $\Im E < 0$ states whereas $C$ preserves $\pm \Im E > 0$. Put another way, symmetries may either be even or odd with respect to the grading, depending on whether they map the subspaces $\mathcal{E}_{\pm}$ onto themselves (even) or onto $\mathcal{E}_{\mp}$ (odd). 
% subsubsection line_gaps (end)
% subsection point_gap_vs_line_gaps (end)

\subsection{Classification through the extended operator} % (fold)
\label{38_fold_recap:extended_operator}
The starting point of their topological classification is the homotopy definition of topological phases. Rather than homotopies $H(\lambda)$, Kawabata et al.\ consider homotopies on the level of the associated \emph{extended Hamilton operators} 
\begin{align}
	\widetilde{H}(\lambda) = \left (
	\begin{matrix}
		0 & H(\lambda) \\
		H(\lambda)^{\dagger} & 0 \\ 
	\end{matrix}
	\right )
	= \widetilde{H}(\lambda)^{\dagger}
	, 
	\label{38_fold_recap:eqn:extended_operator}
\end{align}
which act on the \emph{extended Hilbert space} $\widetilde{\Hil} = \Hil \oplus \Hil$. The overarching idea is that homotopies $H(\lambda)$ are in one-to-one correspondence with homotopies of extended operators $\widetilde{H}(\lambda)$, both subject to symmetry constraints and suitable gap conditions. By design, the extended operator is always hermitian, and thus, diagonalizable — independently of whether $H$ is. 

As a consequence of the doubling, a chiral constraint $\widetilde{\Gamma} = \sigma_3 \otimes \id_{\Hil}$ emerges, 
\begin{align*}
	\widetilde{\Gamma} \, \widetilde{H} \, \widetilde{\Gamma} = - \widetilde{H} 
	. 
\end{align*}
(Anti)commuting symmetries of $H$ can be extended to (anti)commuting symmetries of the hermitian extended operator (as defined in equations~(30)–(35) in \cite{Kawabata_Shiozaki_Ueda_Sato:classification_non_hermitian_systems:2019}). When they relate $H$ to itself, they are block-diagonal, $\dagger$-symmetries are implemented block-offdiagonally. That is captured by whether these extended symmetries commute (block-diagonal) or anticommute (block-offdiagonal) with the chiral constraint $\widetilde{\Gamma}$. 

The classification procedure of Kawabata et al.\ can now be summarized as follows: 
\begin{enumerate}[(1)]
	\item List all symmetries of $H$ and choose the appropriate gap-type (point gap or, if applicable, real, imaginary or generic line gap) for the classification. 
	\item Depending on the gap type “normalize” the operator $H$ (or, equivalently, $\widetilde{H}$). That is, continuously deform $H$ to a unitary (point gap) or an (anti)hermitian spectrally flattened hamiltonian $Q^{\dagger} = \pm Q$ (real ($+$), imaginary ($-$) or generic ($\pm$) line gap). 
	\item Classify the normalized operator with existing theory for unitary or hermitian operators with symmetries; Kawabata et al.\ use twisted equivariant $K$-theory developed by Freed and Moore \cite{Freed_Moore:twisted_equivariant_matter:2013,Gomi:twisted_equivariant_K_theory:2017}. (When the spectrally flattened hamiltonian $Q^{\dagger} = -Q$ is antihermitian, then the hermitian operator $\ii Q = (\ii Q)^{\dagger}$  is classified.)
	\item \emph{By definition,} the topological phase of $H$ is the topological phase of the normalized operator. 
\end{enumerate}
The details for point and line gap classifications are explained in Sections~IV.A.\ and IV.B.\ in \cite{Kawabata_Shiozaki_Ueda_Sato:classification_non_hermitian_systems:2019}. 

To better see the links between \cite{Kawabata_Shiozaki_Ueda_Sato:classification_non_hermitian_systems:2019} and this work, 
I will point out two pertinent facts: in case of a point gap, the deformation that Kawabata et al.\ describe does nothing more than map $H$ onto its phase operator,  
\begin{align*}
	H = H_0 = V_H \, \sabs{H} \mapsto H_1 = V_H 
	. 
\end{align*}
This also works perfectly well when $H$ is not diagonalizable, but importantly, in that case $V_H$ and $\sabs{H}$ necessarily do \emph{not} commute. So I need not make a homotopy argument and can define the “normalized” operator directly. 

Similarly, the homotopy argument can be skipped for the line gap classifications, at least when $P_{\mathrm{rel}}$ can be defined by means of \eg the contour integral~\eqref{intro:eqn:definition_P_rel}: after getting rid of the factor $\pm \ii$ when necessary (\eg in the imaginary line gap case), the (now hermitian) spectrally flattened hamiltonian 
\begin{align*}
	Q = \id_{\Hil} - 2 P_{\mathrm{rel}}' 
\end{align*}
is related to the \emph{unique} orthogonal projection ${P_{\mathrm{rel}}'}^2 = P_{\mathrm{rel}}' = {P_{\mathrm{rel}}'}^{\dagger}$ that maps onto the subspace 
\begin{align*}
	\Hil_{\mathrm{rel}} \overset{\mathrm{def}}{=} \ran P_{\mathrm{rel}} 
	= \ran P_{\mathrm{rel}}' 
	. 
\end{align*}
In general, $P_{\mathrm{rel}} \neq P_{\mathrm{rel}}'$ will disagree even though both are \emph{oblique} projections onto the same subspace since $P_{\mathrm{rel}}$ is hermitian exactly when $H$ is normal. 

These two insights will lead to straight-forward generalizations of Theorems~1 and 2 of \cite{Kawabata_Shiozaki_Ueda_Sato:classification_non_hermitian_systems:2019} to operators which lack periodicity, and could be useful to include disorder and models for aperiodic materials \cite{Bourne_Prodan:Chern_numbers_aperiodic_systems:2018}. 
% subsection classification_through_the_extended_operator (end)
% section the_38_fold_classification_of_cite_kawabata_shiozaki_ueda_sato_classification_non_hermitian_systems_2019_zhou_lee_non_hermitian_topological_classification_2019 (end)
%!TEX root = /Users/max/Dropbox/research/topological classification non-selfadjoint hamiltonians/paper/non_hermitian_classification.tex
\section[Why diagonalizability matters for the topological classification]{Why diagonalizability matters \linebreak for the topological classification} % (fold)
\label{why_diagonalizability_matters}
At first glance, Kawabata et al.'s classification procedure directly applies to non-diagonalizable operators. Unfortunately, upon closer inspection one finds that this classification scheme is inconsistent unless one insists that all operators are diagonalizable and remain diagonalizable during all deformations. The purpose of this section is to first give a simple counterexample that shows the inconsistency and then unpack the mathematical mechanism behind it. 

Generic non-hermitian, \ie non-diagonalizable operators are not as nicely behaved as hermitian and unitary operators, which will matter when I want to classify non-hermitian topological insulators. And because much of our intuition for the behavior of operators is developed from the study of matrices, that is operators on finite-dimensional vector spaces as well as hermitian and unitary operators. No doubt this is thanks to the emphasis placed on quantum mechanics in the education of physicists.

\subsection[Spectral gaps of non-diagonalizable operators may suddenly close]{Spectral gaps of non-diagonalizable \linebreak operators may suddenly close} % (fold)
\label{why_diagonalizability_matters:spectral_gaps_may_close_suddenly}
While this first point is not a shortcoming of \cite{Kawabata_Shiozaki_Ueda_Sato:classification_non_hermitian_systems:2019}, I believe it could nevertheless be very relevant in practical applications. Suppose $H(\lambda)$ is a non-hermitian operator that depends continuously on the parameter $\lambda$. The spectrum $\sigma \bigl ( H(\lambda) \bigr )$ given by~\eqref{38_fold_recap:eqn:definition_spectrum} is the generalization of the set of eigenvalues for operators defined on infinite-dimensional Hilbert spaces $\Hil$ such as $\ell^2(\Z^d)$. 

\begin{figure*}
	\begin{centering}
		\resizebox{140mm}{!}{\includegraphics{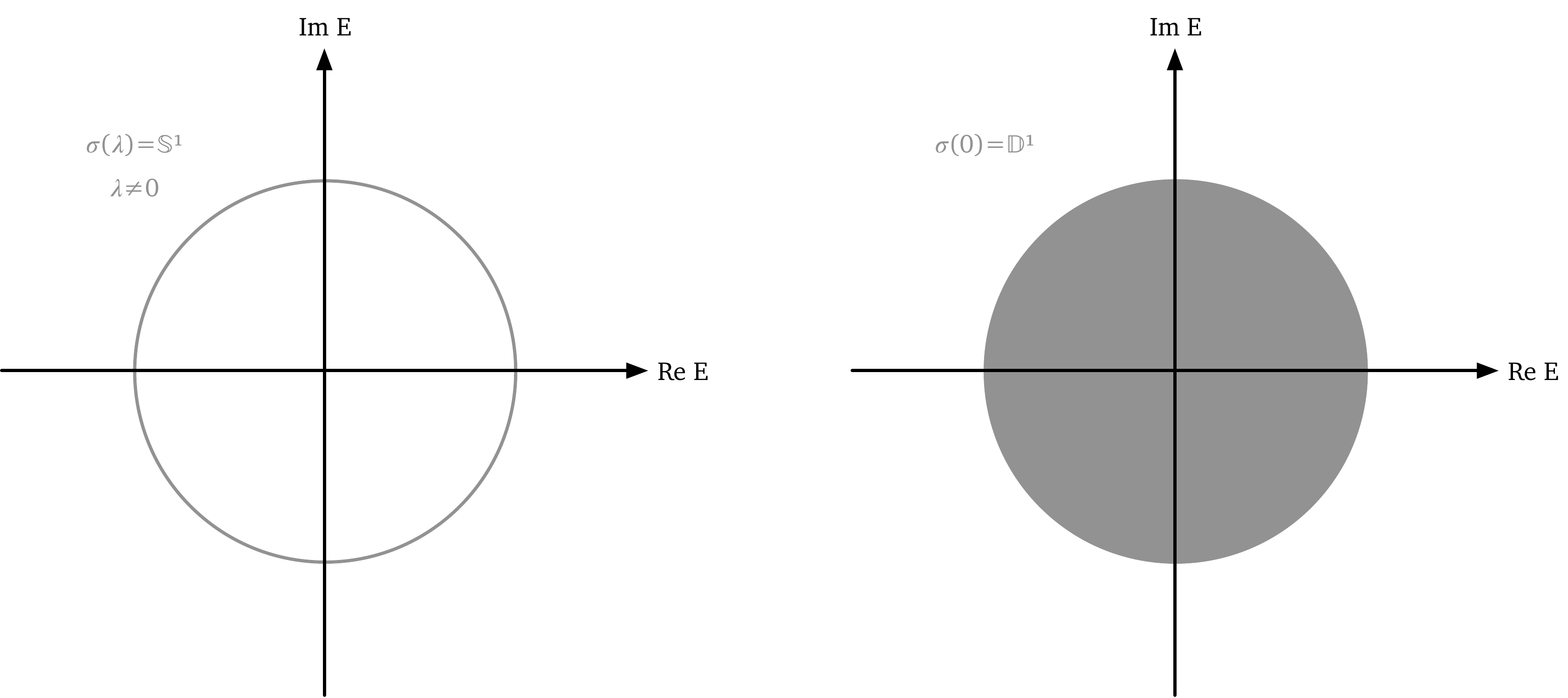}}
	\end{centering}
	\caption{The spectrum of generic non-hermitian operators need not be inner continuous: even when $H(\lambda)$ is perturbed in a continuous fashion, spectrum may suddenly appear. Kato gives an explicit example where for $\lambda \neq 0$ the spectrum is the unit circle $\Sone$; when $\lambda = 0$ the spectrum is the entire unit disc $\mathbb{D}^1$. Kato's example shows that spectral gaps may collapse without warning.}
	\label{intro:figure:highly_symmetric_spectrum}
\end{figure*}

One may expect that the spectrum $\sigma \bigl ( H(\lambda) \bigr )$ depends continuously on $\lambda$ (as is the case for hermitian and indeed, diagonalizable operators), but this is \emph{false.} Kato gives an explicit example in \cite[Chapter~IV, §3, pp.~208–210]{Kato:perturbation_theory:1995} of his excellent text book; the relevant operator is a modification of the unitary shift operator on $\Z$ at a single lattice site. 

Mathematically, one distinguishes between two types of continuities of sets: \emph{outer}
continuity translates to \emph{“spectrum may not suddenly disappear”} whereas \emph{inner} continuity can be thought of as \emph{“stability of gaps”}; the interested reader may look up the precise mathematical definitions in Appendix~\ref{appendix:continuity_spectra}. For non-hermitian operators only outer continuity of the spectrum is guaranteed (\cf \cite[Chapter~IV, §3, 1, Theorem~3.1]{Kato:perturbation_theory:1995}), but when $\Hil$ is infinite-dimensional spectral gaps may suddenly collapse (\cf \cite[Chapter~IV, §3, 2,  pp.~209–210]{Kato:perturbation_theory:1995}). In Kato's example of the perturbed shift operator, away from the bad point the spectrum of the perturbed shift is a subset of the unit circle $\Sone$, but at the singular point where the operator fails to be diagonalizable the spectrum becomes the entire unit disc. This example illustrates that I may not always be able to anticipate the closing of a spectral gap by looking at the spectra near the point where $H(\lambda)$ fails to be diagonalizable. 

In contrast, during deformations in the set of diagonalizable operators, which includes hermitian and unitary operators, the spectra are outer \emph{and} inner continuous. 

Consequently, unless I can categorically exclude the presence of Jordan blocks (\eg by having model operators that do not have band degeneracies or band crossings), this complicates the topological classification in two ways. While the mathematical definition of a topological phase is unaffected by this — points where the relevant spectral gap closes are excluded by definition, it can nevertheless affect numerical studies of tight-binding toy models. Computers can only cover finitely many parameter values, and I may miss the points where the spectral gap suddenly closes. So it may look as if two operators lie inside the same topological phase, because I cannot be sure that the gap closing is anticipated by a continuous shrinking of the gap. 

But secondly, even if the relevant spectral gap is unaffected, the relevant spectrum may change all of a sudden. That leads us to the next point. 
% subsection spectral_gaps_may_suddenly_close (end)

\subsection{Deformations of spectral projections and spectrally flattened hamiltonians need not be continuous} % (fold)
\label{why_diagonalizability_matters:spectral_projections_spectrally_flattened_hamiltonians}
We have just seen that even if $H(\lambda)$ depends continuously on the parameter, the spectrum can suddenly change in a discontinuous fashion at points where $H(\lambda)$ is not diagonalizable. Similarly, other operators constructed from $H(\lambda)$ need not be continuous at points where $H(\lambda)$ is not diagonalizable. And this goes directly to the heart of the matter. 

To illustrate this point, suppose I am interested in the states with eigenvalues less than $0$ of 
\begin{align*}
	H(\lambda) &= \left (
	\begin{matrix}
		1 & 0 & 0 \\
		0 & -2 - \lambda & 1 \\
		0 & 0 & -3 + \lambda \\
	\end{matrix}
	\right )
\end{align*}
when $\sabs{\lambda} \leq 1$. Let me denote the projection onto the associated (proper!) eigenspaces with $P_{\mathrm{rel}}(\lambda)$. Away from $\lambda = \nicefrac{1}{2}$ the matrix $H(\lambda)$ has three distinct eigenvalues, and the range of $P_{\mathrm{rel}}(\lambda)$ is two-dimensional. At $\lambda \neq \nicefrac{1}{2}$ the matrix is not diagonalizable, and I need to distinguish between the algebraic multiplicity of the eigenvalue $E = - 2 \nicefrac{1}{2}$, which is $2$, and the geometric multiplicity, \ie the dimensionality of the eigenspace, which is $1$. For all other values of $\lambda$, the sums of the algebraic and geometric multiplicities for the negative eigenvalues agree and are $1 + 1 = 2$. That means the rank of the projection $P_{\mathrm{rel}}(\lambda)$ changes across the path, and $P_{\mathrm{rel}}(\lambda)$ is discontinuous at $\lambda = \nicefrac{1}{2}$. 

When $\lambda \neq \nicefrac{1}{2}$ there are several equivalent ways to define the relevant projection $P_{\mathrm{rel}}(\lambda)$. For example, I can use the Cauchy integral~\eqref{intro:eqn:definition_P_rel} where the contour only encloses the two negative eigenvalues. The discontinuity of $P_{\mathrm{rel}}(\lambda)$ at $\lambda = \nicefrac{1}{2}$ is due to a \emph{higher-order pole} of the resolvent at $z = -2 \nicefrac{1}{2}$; the resolvent operator $\bigl ( H(\lambda) - z \bigr )^{-1}$ stays perfectly continuous in $\lambda$, though. 

Consequently, also the spectrally flattened hamiltonian 
\begin{align*}
	Q(\lambda) &= \bigl ( \id_{\C^3} - P_{\mathrm{rel}}(\lambda) \bigr ) - P_{\mathrm{rel}}(\lambda) 
	\\
	&= \id_{\Hil} - 2 P_{\mathrm{rel}} 
\end{align*}
has a discontinuity at $\lambda = \nicefrac{1}{2}$. 

The critically-minded reader may object that I am using the “wrong” definition for $P_{\mathrm{rel}}(\lambda)$ at $\lambda = \nicefrac{1}{2}$. They have a point: given that I am dealing with a $3 \times 3$ matrix, I can of course fix the discontinuity by \emph{defining} $P_{\mathrm{rel}}(\nicefrac{1}{2})$ to be the projection onto the \emph{generalized} eigenspace for the eigenvalue $E = - 2 \nicefrac{1}{2}$. However, there exists no such simple, generic fix that applies to arbitrary non-diagonalizable operators on infinite-dimensional Hilbert spaces, especially if I want to include random operators that model systems with disorder. So unless I add extra assumptions on my class of operators, there is no unambiguous way to define $P_{\mathrm{rel}}$ when $H$ is not diagonalizable. And this leads to discontinuities when deforming $P_{\mathrm{rel}}$ across regions of non-diagonalizability. 

A second valid objection is that my example clearly does not have any meaningful topology associated with it. But I can easily construct a $k$-dependent operator to model a periodic system with non-trivial topology that features a Jordan block for some values of $k$ and $\lambda$. Once I break periodicity by including disorder or some other perturbation, it becomes hard to judge whether an operator even \emph{is} diagonalizable in the first place. Also here spectral projections such as $P_{\mathrm{rel}}(\lambda)$ and the spectrally flattened Hamiltonian $Q(\lambda)$ need not depend continuously on $\lambda$ at points where $\sigma_{\mathrm{rel}}$ contains Jordan blocks (and therefore, $H(\lambda)$ is not diagonalizable). In contrast, a commonly used, \emph{sufficient} criterion for periodic operators to be diagonalizable is to insist that all eigenvalues of $H(k)$ are non-degenerate, \ie none of the bands cross and they are not degenerate. 

That directly impacts any topological classification procedure. Most classification techniques do not solely rely on homotopies $H(\lambda)$ of the hamiltonian, but also of \emph{derived quantities} in their definition of topological class. For the two complex Cartan-Altland-Zirnbauer classes, in the simplest $K$-theoretic approach I need to deal with the $K$-groups $K^0(\T^d)$ and $K^1(\T^d)$ (\cf \eg \cite{Prodan_Schulz_Baldes:complex_topological_insulators:2016}). The topology of the system is then encoded in the equivalence class of the relevant projection $[P_{\mathrm{rel}}(\lambda)]_0 \in K^0(\T^d)$ and equivalence class of relevant unitary $[U(\lambda)]_1 \in K^1(\T^d)$. When $H(\lambda)$ is hermitian, $P_{\mathrm{rel}}(\lambda)$ and $U(\lambda)$ depend continuously on the parameter — and therefore, continuous deformations cannot change the equivalence class. When $H(\lambda)$ is not diagonalizable, this is false, continuity of the derived operators $P_{\mathrm{rel}}(\lambda)$ and $U(\lambda)$ is \emph{not} automatic. Instead, this is something that needs to be checked on a case-by-case basis. Below, I will discuss that these issues persist for twisted equivariant $K$-groups \cite{Freed_Moore:twisted_equivariant_matter:2013}. 
% subsection spectral_projections_and_spectrally_flattened_hamiltonians_need_not_be_continuous (end)

\subsection{Formulas for topological invariants may be ill-defined} % (fold)
\label{why_diagonalizability_matters:topological_invariants}
That explains why \emph{formulas} for topological invariants may be ill-defined at points where $H(\lambda)$ is not diagonalizable. Of course, this is well-known in the literature (\cf \eg the orange region in \cite[Figure~4]{Yokomizo_Murakami:non_Bloch_band_theory_non_hermitian_systems:2019}). But there are derivations in existing works which fail in subtle ways in case $\sigma_{\mathrm{rel}}$ contains Jordan blocks; an example can be found in \cite[Appendix~H.1]{Kawabata_Shiozaki_Ueda_Sato:classification_non_hermitian_systems:2019} that relates Chern classes to the resolvent operator (whose operator kernel is the Green's function) via a version of the Riesz-Dunford formula~\eqref{intro:eqn:definition_P_rel}. At parameter values where the hamiltonian has a Jordan block inside the relevant spectrum (in this case states with $\Re E < 0$), the complex integral is no longer well-defined since the resolvent has higher-order poles in the complex plane. At these singular points in parameter space the projection given by \cite[equation~(H9)]{Kawabata_Shiozaki_Ueda_Sato:classification_non_hermitian_systems:2019} can no longer be expressed as an integral in the complex plane, and in principle, the value of the various Chern numbers can change. 

A subtle, but important question is whether it is the topological \emph{classification} that fails at these points or if it is only the \emph{formulas} for the invariants that fails. That subtle distinction sometimes arises, because the abstract definition of topological invariants typically only requires continuity rather than smoothness or analyticity. Formulas that compute these topological invariants, however, often involve quantities from differential geometry (\eg connections and curvatures) that are only defined once I can guarantee differentiability or smoothness. On balance, I think the evidence points to the \emph{classification failing.} 

One way to introduce Chern numbers in periodic class~A systems is as topological invariants that characterize vector bundles up to isomorphisms, which are obtained by gluing the subspaces $\ran P_{\mathrm{rel}}(k)$ together in a \emph{continuous} fashion. The formula for the Chern number requires us to compute first-order derivatives of $P_{\mathrm{rel}}(k)$, and if all I require is continuity, then this is not necessarily a given. In the context of physical systems, $P_{\mathrm{rel}}(k)$ is usually even analytic in $k$ so this is almost never a problem in practice; although there are systems with slowly decaying, long-range interactions \cite{Aharony_Fisher:critical_behavior_dipolar_interactions:1973,Yamamoto_et_al:topological_origin_magnetostatic_waves:2019} that are the exception to the rule. Conceptually, though, it is nevertheless true that the classification is well-defined for projections that are only continuous in $k$ even though the standard formulae require differentiability; mathematically, the fact that continuous and analytic equivalence of vector bundles are one and the same is known as the \emph{Oka principle} (\cf \eg \cite[p.~268, Satz~I]{Grauert:analytische_Faserungen:1958} and \cite[Section~II.F]{DeNittis_Lein:exponentially_loc_Wannier:2011}). 
% subsection formulas_for_topological_invariants_may_be_ill_defined (end)

\subsection{For some simple example systems regions with Jordan blocks are topologically non-trivial obstacles} % (fold)
\label{why_diagonalizability_matters:jordan_blocks_matter_example}
There is at least one work, which explicitly identifies parameter regions where the hamiltonian becomes non-diagonalizable as a topologically non-trivial obstacle to deformations. Wojcik et al.'s approach \cite{Wojcik_Sun_Bzdusek_Fan:topological_classification_non_hermitian_hamiltonians:2020} is as simple as it is elegant: they opt to work with the homotopy definition of topological phases directly and compute the homotopy groups for periodic two-band class~A hamiltonians in $d = 2$ explicity. (Usually this is the insurmountable obstacle and the reason why few works take the direct approach: unlike, say, $K$-groups homotopy groups are not algorithmically computable.) Energy bands of non-hermitian operators can be braided in the complex plane, and Wojcik et al.\ show from first principles that I can classify hamiltonians in terms of the \emph{non-abelian} braid group (\cf Section~IV and Fig.~5) and its interaction with the standard hermitian invariants. While Wojcik et al.\ explain how to generalize their classification to $N > 2$ bands, incorporating symmetries seems more ambitious — at least if I want to compute the relevant homotopy groups explicitly. 

All in all, their results suggest two important conclusions: first of all, for the purpose of topological classifications regions where $H$ acquires Jordan blocks are obstacles like the gap-closing regions. Crossing any of these forbidden regions allows us to change the topological phase. 

A second important conclusion from their result is that any topological classification of non-hermitian — indeed, non-hermitian, \emph{diagonalizable} operators in terms of \emph{commutative} groups do not necessarily capture all topological features of the system. Since $K$-groups are commutative by design, it seems that \emph{any} approach to classify suitable non-hermitian operators relying solely on $K$-groups cannot resolve all features of the various inequivalent topological phases in a topological class. 
% subsection for_some_simple_example_systems_regions_with_jordan_blocks_are_topologically_no_trivial_obstacles (end)

\subsection{Inconsistency with the classification in Kawabata et al.\ when including non-diagonalizable operators} % (fold)
\label{why_diagonalizability_matters:gap_Kawabata_classification}
As explained in Section~\ref{38_fold_recap:extended_operator} Kawabata et al.\ \emph{by definition} identify the topological phase of a non-hermitian operator with the topological phase of a “normalized” operator $\hat{H}$ that is homotopic to $H$. Depending on the gap type, this normalized operator is a unitary in the point gap case and a spectrally flattened hamiltonian (up to possibly a factor of $\pm \ii$) when considering line gaps. Implicit in this definition is the claim that the classification map 
\begin{align*}
	H \mapsto [H] \overset{\mathrm{def}}{=} [\hat{H}] 
\end{align*}
which assigns to each operator the topological phase (homotopy equivalence class) $[H]$ is well-defined. Well-definedness is a term of art in mathematics and means that the object being defined exists and its definition is self-consistent. In the context of homotopy equivalence classes it means that $[H] = [\hat{H}]$ must be independent of the path in operator space (homotopy) connecting $H$ with its normalized operator $\hat{H}$. What is more, since when $H$ is diagonalizable the normalized operators can be obtained directly, the classification map needs to be consistent with that as well.

\subsubsection{Spectral splitting that enters the classification technique in \cite{Kawabata_Shiozaki_Ueda_Sato:classification_non_hermitian_systems:2019} need not be continuous} % (fold)
\label{why_diagonalizability_matters:gap_Kawabata_classification:discontinuities_spectral_splitting}
To locate the exact spot where things go wrong, one really needs to go into the nitty-gritty. Kawabata et al.\ classify hamiltonians by computing twisted equivariant $K$-groups, which were developed by Freed and Moore \cite{Freed_Moore:twisted_equivariant_matter:2013,Gomi:twisted_equivariant_K_theory:2017}. These $K$-groups characterize spectrally flattened hamiltonians with symmetries and not only require a hermitian operator $H$ as initial data, but \emph{also a splitting of the Hilbert space} $\Hil = \Hil_+ \oplus \Hil_-$ (\cf \cite[Chapter~10]{Freed_Moore:twisted_equivariant_matter:2013}). In the context of condensed matter physics, the splitting corresponds to valence and conduction bands; put another way, the “negative energy” subspace $\Hil_- = \ran P_{\mathrm{rel}}$ is the range of the relevant projection. Hence, twisted equivariant $K$-theory not only requires the continuity of $H(\lambda)$ in the deformation parameter $\lambda$, but also the relevant projection $P_{\mathrm{rel}}(\lambda)$ (or, equivalently, $\Hil = \Hil_+(\lambda) \oplus \Hil_-(\lambda)$). 

For the sake of clarity, let me focus on the real line gap classification. Although my arguments can easily be modified for the other gap-type classifications. Given that the classification is supposed to hold for non-hermitian operators, it must also apply to normal operators. So from now on let me assume that $H$ is normal, \ie $[H , H^{\dagger}] = 0$. Consequently, any spectral projection, including $P_{\mathrm{rel}} = P_{\mathrm{rel}}^{\dagger}$, is \emph{hermitian} out of the box. 

Kawabata et al.\ construct the “normalized” operator, \ie the spectrally flattened hamiltonian $Q$ (\cf \cite[Figure~2~(c)]{Kawabata_Shiozaki_Ueda_Sato:classification_non_hermitian_systems:2019}), by means of a homotopy. Alternatively, I can tread the path from Section~\ref{38_fold_recap:extended_operator} and define $P_{\mathrm{rel}}$ and therefore, $Q = \id_{\Hil} - 2 P_{\mathrm{rel}} = Q^{\dagger}$ directly, \eg as a contour integral~\eqref{intro:eqn:definition_P_rel} or via functional calculus. In this case, $H$ and $Q$ need to lie in the same topological phase. 

Now assume that I am given two normal operators $H_0$ and $H_1$ that are homotopic in the larger set of non-hermitian operators, but \emph{not} homotopic in the smaller set of diagonalizable operators (each with the relevant symmetries). In that case, $P_{\mathrm{rel}}(\lambda)$ and $Q(\lambda)$ are not even well-defined for values of $\lambda$ where $H(\lambda)$ is not diagonalizable (\cf Section~\ref{why_diagonalizability_matters:spectral_projections_spectrally_flattened_hamiltonians}). That means I have no idea whether the topological phase of $Q_0 = Q(0)$ and $Q(1) = Q_1$ are one and the same or not! The example covered in Section~\ref{why_diagonalizability_matters:jordan_blocks_matter_example} not only tells us that this really occurs in practice, but that there are examples when $Q_0$ and $Q_1$ must lie in different topological phases as they are labeled by different topological invariants. 

However, if I impose that homotopies have to lie in the smaller set of \emph{diagonalizable} operators, order in the universe is restored: as long as the homotopy $H(\lambda)$ stays diagonalizable, the corresponding relevant projection $P_{\mathrm{rel}}$ and the spectrally flattened operator $Q(\lambda)$ inherit the continuity in $\lambda$. For diagonalizable operators, this definition is consistent. 
% subsubsection spectral_splitting_that_enters_the_classification_technique_in_cite_kawabata_shiozaki_ueda_sato_classification_non_hermitian_systems_2019_need_not_be_continuous (end)

\subsubsection{Existence of different normalized operators} % (fold)
\label{why_diagonalizability_matters:gap_Kawabata_classification:consistency_different_normalizations}
The inconsistency covered in the previous subsection is conclusive. Nevertheless, it helps to look at it from yet another perspective. Namely, the normalized operator is \emph{not unique}, and therefore to ensure well-definedness, one would have to prove that the topological phase is independent of our choice of normalization. 

Mathematically speaking, I have a choice of scalar product even when I insist that symmetries need to be implemented (anti)unitarily; I will give an explicit example from electromagnetism where this occurs below. In class~A, \ie in the absence of any symmetries, there are as many such decompositions as there are scalar products on my vector space $\Hil$. That immediately raises the question: if 
\begin{align*}
	H &= V_H \, \sabs{H} 
	= V_H' \, \sabs{H}' 
	\label{why_diagonalizability_matters:eqn:two_polar_decompositions}
\end{align*}
are two polar decompositions of my operator where \eg $V_H$ and $V_H'$ are unitary with respect to $\scpro{\varphi}{\psi}$ and the second scalar product 
\begin{align}
	\scpro{\varphi}{\psi}' = \scpro{\varphi}{W \, \psi} 
	, 
	\label{why_diagonalizability_matters:eqn:non_prime_prime_scalar_products}
\end{align}
which I assume are related by a bounded operator $W$ with bounded inverse. 

The above formula~\eqref{why_diagonalizability_matters:eqn:two_polar_decompositions} directly shows that $H$, $V_H$ and $V_H'$ are all homotopic. So for the point gap classification to be self-consistent (\ie well-defined), on would need to show that the topological classes of $V_H$ and $V_H'$ must always agree. 

Of course, I can play this game with spectrally flattened hamiltonians, too, which enter the line gap classifications. Ignoring the factor $\pm \ii$, my choice of scalar product selects one of the two spectrally flattened hamiltonians, 
\begin{align*}
	Q &= \id_{\Hil} - 2 P_{\mathrm{rel}} 
	= Q^{\dagger} 
	, 
	\\
	Q' &= \id_{\Hil} - 2 P_{\mathrm{rel}}' 
	= {Q'}^{\dagger'} 
	, 
\end{align*}
both of which are defined in terms of two different relevant projections. The condition that their ranges agree,  
\begin{align*}
	\ran P_{\mathrm{rel}} = \Hil_- = \ran P_{\mathrm{rel}}' 
	,
\end{align*}
does not uniquely single out one of them, though. Indeed, there are \emph{infinitely many} oblique projections onto a given subspace $\Hil_-$. However, once I ask the projection to be orthogonal does this association become unique, and I am led to either $P_{\mathrm{rel}} = P_{\mathrm{rel}}^{\dagger}$ or $P_{\mathrm{rel}}' = {P_{\mathrm{rel}}'}^{\dagger'}$, depending on my choice of scalar product. 

By the same token, for each scalar product I obtain an extended hamiltonian 
\begin{align*}
	\widetilde{H}' \overset{\mathrm{def}}{=} \left (
	\begin{matrix}
		0 & H \\
		H^{\dagger'} & 0 \\
	\end{matrix}
	\right )
	= \widetilde{H}^{\prime \; \dagger'} 
\end{align*}
that differs from $\widetilde{H}$ only by how I “hermitianize” the operator $H$. Again, the unanswered question is whether the topological classifications of $\widetilde{H}$ and $\widetilde{H}'$ must always agree. 

Even if I add symmetries to the discussion and insist that the symmetries must be (anti)unitary in any scalar product I use, I still have plenty of scalar products to choose from. In fact, any weight for \eqref{why_diagonalizability_matters:eqn:non_prime_prime_scalar_products} that commutes with all the symmetries $[W , U_{(\dagger)}] = 0$ will do. Then equation~\eqref{why_diagonalizability_matters:eqn:non_prime_prime_scalar_products} gives me a second scalar product with respect to which all symmetries are (anti)unitary. 

While this scenario may seem quite artificial, it actually does occur in applications. Many classical waves can be recast as a Schrödinger equation where a Maxwell-type operator $M = W^{-1} D$ plays the role of the quantum hamiltonian \cite{DeNittis_Lein:Schroedinger_formalism_classical_waves:2017}. Maxwell-type operators are \emph{hermitian} with respect to a weighted scalar product $\scppro{\varphi}{\psi} = \scpro{\varphi}{W \, \psi}$ and hence, diagonalizable (\cf \cite[Proposition~6.2]{DeNittis_Lein:Schroedinger_formalism_classical_waves:2017}). For electromagnetic waves, material symmetries (\cf \cite[Section~3]{DeNittis_Lein:symmetries_electromagnetism:2020}) are (anti)unitary with respect to the usual (vacuuum) \emph{and} the weighted scalar product; conceptually, this is really important since the idea is that materials selectively break or preserve vacuum symmetries. In principle, the recipe of Kawabata et al.\ can be implemented with the vacuum adjoint ${}^{\dagger}$ or the biorthogonal adjoint ${}^{\ddagger}$. Do these two classifications of the Maxwell operator agree? And are they consistent with the classification obtained in \cite[Section~4]{DeNittis_Lein:symmetries_electromagnetism:2020}? 
% subsubsection existence_of_different_normalized_operators (end)
% subsection problems_with_the_classification_via_the_extended_operator_eqref (end)
% section why_diagonalizability_matters (end)
%!TEX root = /Users/max/Dropbox/research/topological classification non-selfadjoint hamiltonians/paper/non_hermitian_classification.tex
% \section{Alternative approach in a simplified setting: symmetries and constraints of relevant states} % (fold)
\section[Classifying the projection onto relevant states in a simplified setting]{Classifying the projection onto \linebreak relevant states in a simplified setting} % (fold)
\label{my_way_simplified}
Now it is my turn to explain the ins and outs of the classfication scheme I outlined in the introduction. The diagonalizability assumption that I have discussed at length ensures $P_{\mathrm{rel}}$ is well-defined and continuous, symmetry- and gap-preserving deformations of $H$ lead to continuous deformations of the relevant projection. While I could start with the fully generic case, the following assumption allows me to simplify some of my arguments: 
\begin{assumption}\label{my_way_simplified:assumption:simplifying_assumption}
	Throughout this section, I will assume that $H$ is normal with respect to the initially given scalar product, \ie $[H , H^{\dagger}] = 0$ holds true. 
\end{assumption}
As a consequence, equation~\eqref{diagonalizable_operators:eqn:naive_decomposition_real_imaginary_parts} splits $H$ into \emph{commuting} real and imaginary parts. And symmetries are assumed to be (anti)unitary with respect to the \emph{same} scalar product that makes $H$ normal. This streamlines many arguments, because \eg $\dagger$-symmetries exchange $H = H_{\Re} + \ii H_{\Im}$ with its \emph{commuting} adjoint $H^{\dagger} = H_{\Re} - \ii H_{\Im}$, 
\begin{align*}
	U_{\dagger} \, H \, U_{\dagger}^{-1} &= \pm H^{\dagger} 
	\; \; \Longleftrightarrow \; \; 
	U_{\dagger} \, H^{\dagger} \, U_{\dagger}^{-1} = \pm H 
	. 
\end{align*}
The action of ($\dagger$-)symmetries can then be rephrased in terms of real and imaginary part operators. 

Furthermore, the normality of $H$ implies all spectral projections — including $P_{\mathrm{rel}}$ — are hermitian out of the box. 

Later on in Section~\ref{my_way} I will explain how to modify the arguments made here in case $H$ is diagonalizable, but not normal with respect to the initially given scalar product.

\subsection{Revisiting the classification of hermitian operators} % (fold)
\label{my_way_simplified:hermitian_case}
Let me illustrate my approach with a familiar example. Suppose $H = H^{\dagger}$ models a time-reversal symmetric fermionic system from condensed matter, and possesses an even particle-hole symmetry $C$ and an odd time-reversal symmetry $T$. Furthermore, I suppose there is a spectral gap around $0$. 

At zero temperature, the state of such systems is given by the Fermi projection where all states up to the Fermi energy $E_{\mathrm{F}} = 0$ are completely filled and those above are all empty. Here, the projection onto the relevant states is nothing but the Fermi projection, 
\begin{align*}
	P_{\mathrm{rel}} = P_{\mathrm{F}} = 1_{(-\infty,0]}(H) 
	. 
\end{align*}
The presence of symmetries of $H$ will lead to symmetries and constraints of $P_{\mathrm{rel}}$. The \emph{distinction between symmetries and constraints} is essential for a physically meaningful topological classification. Clearly, the time-reversal symmetry $T \, H \, T^{-1} = + H$ manifests itself on the level of the relevant projection as 
\begin{align}
	T \, P_{\mathrm{rel}} \, T^{-1} = P_{\mathrm{rel}} 
	, 
	\label{my_way_simplified:eqn:symmetry}
\end{align}
and $T$ is a \emph{symmetry} of $P_{\mathrm{rel}}$. 

In contrast, a particle-hole symmetry $C \, H \, C^{-1} = - H$ leads to a \emph{constraint} imposed on the projection, 
\begin{align}
	C \, P_{\mathrm{rel}} \, C^{-1} = \id_{\Hil} - P_{\mathrm{rel}} 
	. 
	\label{my_way_simplified:eqn:constraint}
\end{align}
Its presence forces that deformations of states below the Fermi energy (“electron states”) are in lock-step with matching deformations above the Fermi energy (“hole states”). Moreover, electrons and holes can never mingle as that would require a gap-closing deformation. 

Lastly, their product — possibly garnished with a factor $\pm \ii$ — gives a chiral symmetry, which gives rise to another constraint, 
\begin{align}
	S \, P_{\mathrm{rel}} \, S^{-1} = \id_{\Hil} - P_{\mathrm{rel}} 
	. 
	\label{my_way_simplified:eqn:constraint_2}
\end{align}
Equivalently, symmetries and constraints may be formulated in terms of the \emph{spectrally flattened hamiltonian} 
\begin{align}
	Q &= (\id_{\Hil} - P_{\mathrm{rel}}) - P_{\mathrm{rel}} 
	\notag \\
	&= \id_{\Hil} - 2 P_{\mathrm{rel}}
	. 
	\label{my_way_simplified:eqn:spectrally_flattened_hamiltonian}
\end{align}
In my terminology, symmetries commute with $Q$ whereas constraints anticommute. 

In short, I can reduce the topological classification of hermitian operators to the classification involving only projections. At first, this seems to imply that only $K_0$-groups play a role here (which in the operator-theoretic approach consist of equivalence classes of projections), but unitaries (whose equivalence classes make up the relevant $K_1$-group) emerge naturally in this context. 

Ordinarily, unitaries enter through the spectrally flattened hamiltonian: after choosing a basis for the eigenspaces of $S$, I can write 
\begin{align*}
	Q \simeq \left (
	\begin{matrix}
		0 & u \\
		u^{\dagger} & 0 \\
	\end{matrix}
	\right )
\end{align*}
in terms of a unitary $u$. The choice of basis is important, since class~AIII admits only a \emph{relative} classification and I have to fix a reference system, which I regard as trivial. Then equation~\eqref{my_way_simplified:eqn:spectrally_flattened_hamiltonian} implies that the unitary can be recovered directly from the relevant projection,
\begin{align*}
	P_{\mathrm{rel}} &\simeq \frac{1}{2} \left (
	\begin{matrix}
		\id & -u \\
		-u^{\dagger} & \id \\
	\end{matrix}
	\right )
	, 
\end{align*}
rather than the spectrally flattened hamiltonian $Q$ (\cf the discussion in Chapters~2.3 and 7.3 in \cite{Prodan_Schulz_Baldes:complex_topological_insulators:2016}). Hence, knowing $P_{\mathrm{rel}}$ suffices to perform a $K$-theoretic classification. 

One central difference to the philosophy of \cite{Kawabata_Shiozaki_Ueda_Sato:classification_non_hermitian_systems:2019} is that it is entirely unnecessary to find a homotopy that deforms the original hamiltonian to $Q$. That is because the relevant symmetries of $H$ are encoded in symmetries and constraints of $P_{\mathrm{rel}}$ — and hence, $Q$. 

Based on the distinction between symmetries and constraints, Kennedy and Zirnbauer performed a topological classification of fermionic systems \cite{Kennedy_Zirnbauer:Bott_periodicity_Z2_symmetric_ground_states:2016} using homotopy theory. When the system is periodic, then the Fermi projection gives rise to the so-called Bloch vector bundle that comes furnished with the odd time-reversal symmetry~\eqref{my_way_simplified:eqn:symmetry}. These class~AII vector bundles have been classified by De~Nittis and Gomi \cite{DeNittis_Gomi:AII_bundles:2014}, and not surprisingly the $\Z_2$-valued Kane-Mele invariant arises in the classification. 

The careful distinction between symmetries and constraints has also been essential for obtaining a physically meaningful classification in classical waves \cite{DeNittis_Lein:symmetries_electromagnetism:2020,Lein_Sato:topological_classification_magnons:2019}. In topological photonic and magnonic crystals complex conjugation gives rise to a particle-hole-type \emph{constraint} that stems from the real-valuedness of electromagnetic fields (\cf \cite[Section~3.3]{DeNittis_Lein:Schroedinger_formalism_classical_waves:2017}) and classical spin waves (\cf \cite[Section~III.C.1]{Lein_Sato:topological_classification_magnons:2019}), respectively. Since the real-valuedness is a fundamental tenet of classical waves, in these contexts the particle-hole constraints are unbreakable. 
% subsection illustration_for_the_hermitian_case (end)

\subsection{Definition of $P_{\mathrm{rel}}$ for normal operators} % (fold)
\label{my_way_simplified:definition_P_rel}
Mathematically, the definition immediately extends to diagonalizable operators in a straightforward fashion: suppose I am given a diagonalizable operator $H$ whose spectrum is as in \eg Figure~\ref{intro:figure:highly_symmetric_spectrum} or \ref{38_fold_recap:figure:nested_Cs_spectrum}. On physical grounds I identify the relevant states. What is important for this is that the relevant part of the spectrum $\sigma_{\mathrm{rel}}$ needs to be separated from the remainder by a gap. At this stage I \emph{do not} assume the existence of a line gap or so, all I care about is that I can enclose $\sigma_{\mathrm{rel}}$ with a closed contour that does not enclose or intersect with any other part of the spectrum of $H$. 

When $H$ is also periodic, choosing the relevant spectrum $\sigma_{\mathrm{rel}}$ amounts to choosing relevant energy or frequency bands $\bigl \{ E_{n_j}(k) \bigr \}_{j = 1 , \ldots N_{\mathrm{rel}}}$. Consequently, I can expand
\begin{align}
	P_{\mathrm{rel}}(k) &= \sum_{j = 1}^{N_{\mathrm{rel}}} \sopro{\varphi_{n_j}(k)}{\varphi_{n_j}(k)}
	\label{my_way_simplified:eqn:P_rel_periodic_Bloch_functions}
\end{align}
in terms the relevant Bloch (eigen)functions $\varphi_n(k)$. The diagonalizability condition enforces that the spectral projections give rise to a resolution of the identity, 
\begin{align*}
	\id = \sum_{n = 1}^N \sopro{\varphi_n(k)}{\varphi_n(k)}
\end{align*}
Alternatively, I may equivalently write the relevant projection as a Riesz-Dunford integral 
\begin{align}
	P_{\mathrm{rel}}(k) &= \frac{\ii}{2\pi} \int_{\Gamma(\sigma_{\mathrm{rel}})} \dd z \, \bigl ( H(k) - z \bigr )^{-1} 
	, 
	\label{my_way_simplified:eqn:P_rel_periodic_complex_integral}
\end{align}
where the contour $\Gamma(\sigma_{\mathrm{rel}})$ encloses only the relevant part of the spectrum. As explained in Section~\ref{why_diagonalizability_matters:spectral_projections_spectrally_flattened_hamiltonians} diagonalizability ensures the integral~\eqref{my_way_simplified:eqn:P_rel_periodic_complex_integral} is well-defined. 

If all I am interested in are periodic systems, then either \eqref{my_way_simplified:eqn:P_rel_periodic_Bloch_functions} or \eqref{my_way_simplified:eqn:P_rel_periodic_complex_integral} will do just fine and I may proceed in my analysis. However, this is not entirely satisfactory since topological phenomena are known to exist also in disordered systems. In fact, in some insulators, disorder is the proximate cause behind the absence of conducting, delocalized states \cite{Anderson:Anderson_localization_absence_diffusion:1958}. So I think it is important to offer definitions for $P_{\mathrm{rel}}$ that do not rely on periodicity. For diagonalizable operators, I have two options, functional calculus or expressing it as a complex integral akin to \eqref{my_way_simplified:eqn:P_rel_periodic_complex_integral}. 

For normal operators functional calculus systematically assigns an operator 
\begin{align}
	f(H) &\overset{\mathrm{def}}{=} \int_{\C} f(E) \, \dd \sopro{\psi_E}{\psi_E}
	\label{diagonalizable_operators:eqn:functional_calculus}
\end{align}
to each suitable functions $f : \C \longrightarrow \C$ (\cf Appendix~\ref{appendix:functional_calculus_normal_operators} and references therein). Probably the best-known example is $f(E) = \e^{- \ii t E}$, which gives rise to the time-evolution. Diagonalizability enters again as a crucial assumption through the guise of normality. 

One fact will be important in just a moment: the adjoint of $f(H)$ is 
\begin{align*}
	f(H)^{\dagger} = \bar{f}(H)
	, 
\end{align*}
and \emph{not} $\bar{f}(H^{\dagger})$, where $\bar{f}(E) = \overline{f(E)}$ is the complex conjugate function. In particular, if $f = \bar{f}$ is a real-valued function, then the operator $f(H)^{\dagger} = f(H)$ is automatically hermitian — even if $H$ is not. 

Spectral projections, including the projection onto the relevant states, fall into that category, since it can be defined through functional calculus 
\begin{align}
	P_{\mathrm{rel}} \overset{\mathrm{def}}{=} 1_{\sigma_{\mathrm{rel}}}(H)
	\label{diagonalizable_operators:eqn:relevant_projection_functional_calculus}
\end{align}
for the indicator function $1_{\sigma_{\mathrm{rel}}}(E)$, which equals $1$ when $E$ lies in $\sigma_{\mathrm{rel}}$ and $0$ otherwise. Since $1_{\Lambda}(H) = 0$ is the trivial projection whenever $\Lambda$ lies outside of the spectrum, $\Lambda \cap \sigma(H) = \emptyset$, I can in fact enlarge $\sigma_{\mathrm{rel}}$ and \eg replace $\sigma_{\mathrm{rel}} = \sigma_{++}$ with the larger set $\sigma_{\mathrm{rel}} = (0,\infty) \times (0,\infty)$ (real and imaginary parts need to be positive); as long as there is no other spectrum in the upper-right quadrant of the complex plane, the two resulting projections will coincide. I will exploit this fact below to greatly simplify computations. 

When $\sigma_{\mathrm{rel}} = \sigma_{\mathrm{rel},\Re} \times \sigma_{\mathrm{rel},\Im}$ can be taken as a product (\eg a square in the complex plane), I can exploit that normal operators $H = H_{\Re} + \ii H_{\Im}$ split into real and imaginary parts, and simplify~\eqref{diagonalizable_operators:eqn:relevant_projection_functional_calculus} to 
\begin{align}
	P_{\mathrm{rel}} &= 1_{\sigma_{\mathrm{rel},\Re}}(H_{\Re}) \; 1_{\sigma_{\mathrm{rel},\Im}}(H_{\Im}) 
	\label{diagonalizable_operators:eqn:P_rel_product_set}
	\\
	&= 1_{\sigma_{\mathrm{rel},\Im}}(H_{\Im}) \; 1_{\sigma_{\mathrm{rel},\Re}}(H_{\Re})
	. 
	\notag 
\end{align}
Since real and imaginary parts commute, the order did not matter in the above definition. For instance, assume I consider an operator whose spectrum is given by Figure~\ref{intro:figure:highly_symmetric_spectrum}, and I would like to construct the relevant projection for $\sigma_{++}$. 

The second option I have is to express 
\begin{align}
	P_{\mathrm{rel}} = \frac{\ii}{2 \pi} \int_{\Gamma(\sigma_{\mathrm{rel}})} \dd z \, (H - z)^{-1} 
	\label{diagonalizable_operators:eqn:relevant_projection_Riesz_Dunford_formula}
\end{align}
via the Riesz-Dunford formula as a contour integral, where the contour $\Gamma(\sigma_{\mathrm{rel}})$ encloses $\sigma_{\mathrm{rel}}$ in a counterclockwise fashion, but does not intersect with any other spectrum of $H$. Note that also here, the diagonalizability is essential. The operator above gives rise to an ordinary integral of the complex-valued function
\begin{align*}
	f_{\psi}(z) = \scpro{\psi \, }{ \, (H - z)^{-1} \psi}
\end{align*}
in the complex plane after taking expectation values; here, $\psi \in \Hil$ is a vector in the Hilbert space that is a parameter. I can think of the collection $\bigl \{ f_{\varphi_n}(z) \bigr \}_{n \in \mathcal{I}}$ as matrix elements of the resolvent on the diagonal with respect to a basis $\{ \varphi_n \}_{n \in \mathcal{I}}$ of $\Hil$. The offdiagonal matrix elements can be recovered from the polarization formula, 
\begin{align*}
	\scpro{\varphi}{\psi} &= \frac{1}{4} \Bigl ( \scpro{\varphi + \psi}{\varphi + \psi} - \scpro{\varphi - \psi}{\varphi - \psi} 
	\, \Bigr . + \\
	&\qquad \quad \Bigl . 
	- \ii \, \scpro{\varphi + \ii \, \psi}{\varphi + \ii \, \psi} + \ii \, \scpro{\varphi - \ii \, \psi}{\varphi - \ii \, \psi} \Bigr ) 
	\, . 
\end{align*}
For diagonalizable, including normal operators, one can show that all poles of $f_{\psi}(z)$ are first order. But if $\sigma_{\mathrm{rel}}$ contains Jordan blocks, then the poles inside $\sigma_{\mathrm{rel}}$ for some $\psi$ are higher-order and the contour integral is ill-defined. 

Because $H$ is normal, the spectral projection $P_{\mathrm{rel}} = P_{\mathrm{rel}}^{\dagger}$ is hermitian (\cf Theorem~\ref{appendix:functional_calculus_normal_operators:thm:functional_calculus}~(1)). Once I introduce $P_{\mathrm{rel}}^{\perp} = \id_{\Hil} - P_{\mathrm{rel}}$, the Hilbert space 
\begin{align*}
	\Hil = \ran P_{\mathrm{rel}} \oplus \ran P_{\mathrm{rel}}^{\perp}
\end{align*}
then splits neatly into the relevant states and all other states. This is the decomposition that enters as a datum in the twisted equivariant $K$-theory (\cf \cite[Chapter~10]{Freed_Moore:twisted_equivariant_matter:2013}) that is used in \cite{Kawabata_Shiozaki_Ueda_Sato:classification_non_hermitian_systems:2019}. 

I shall also introduce the projection 
\begin{align}
	P_{\mathrm{rel},\dagger} \overset{\mathrm{def}}{=} 1_{\sigma_{\mathrm{rel}}}(H^{\dagger}) 
	\label{diagonalizable_operators:eqn:P_rel_dagger}
\end{align}
for the operator $H^{\dagger}$. When $H \neq H^{\dagger}$ and $\sigma_{\mathrm{rel}} \neq \overline{\sigma_{\mathrm{rel}}}$, the operators $P_{\mathrm{rel}} = P_{\mathrm{rel}}^{\dagger}$ and $P_{\mathrm{rel},\dagger} = P_{\mathrm{rel},\dagger}^{\dagger} \neq P_{\mathrm{rel}}$ are \emph{not} adjoints of one another! 

Nevertheless, there is a simple relation between $P_{\mathrm{rel},\dagger}$ and spectral projections of $H = H_{\Re} + \ii H_{\Im}$: given that $H^{\dagger} = H_{\Re} - \ii H_{\Im}$ differs from $H$ only by a $-$ sign in front of the imaginary part, 
\begin{align}
	P_{\mathrm{rel},\dagger} &= 1_{\overline{\sigma_{\mathrm{rel}}}}(H) 
	\label{diagonalizable_operators:eqn:P_rel_dagger_as_spectral_projection_of_H}
\end{align}
is just the spectral projection for the spectral region $\overline{\sigma_{\mathrm{rel}}}$ obtained by reflecting $\sigma_{\mathrm{rel}}$ about the real axis. When the set $\sigma_{\mathrm{rel}} = \sigma_{\mathrm{rel},\Re} \times \sigma_{\mathrm{rel},\Im}$ has product form, a computation analogous to \eqref{diagonalizable_operators:eqn:P_rel_product_set} easily allows me to confirm this directly, 
\begin{align*}
	P_{\mathrm{rel},\dagger} &= 1_{\sigma_{\mathrm{rel}}}(H^{\dagger})
	\\
	&= 1_{\sigma_{\mathrm{rel},\Re}}(H_{\Re}) \; 1_{\sigma_{\mathrm{rel},\Im}}(- H_{\Im}) 
	\\
	&= 1_{\sigma_{\mathrm{rel},\Re}}(H_{\Re}) \; 1_{-\sigma_{\mathrm{rel},\Im}}(H_{\Im})
	\\
	&= 1_{\overline{\sigma_{\mathrm{rel}}}}(H)
	. 
\end{align*}
%
% subsection definition_of_p__mathrm_rel_for_diagonalizable_operators (end)

\subsection[Extension of the topological classification to certain non-hermitian examples]{Extension of the topological classification \linebreak to certain non-hermitian examples} % (fold)
\label{my_way_simplified:non_hermitian_examples}
These ideas can be extended in a natural way from the hermitian to the non-hermitian case. Physics decides what states should be regard as relevant for the classification, which fixes $\sigma_{\mathrm{rel}}$. From there, I proceed algorithmically: the relevant states give rise to a spectral projection $P_{\mathrm{rel}}$. Depending on the symmetries of $H$ and how symmetric I have chosen the spectral region $\sigma_{\mathrm{rel}}$, this gives rise to symmetries and constraints of $P_{\mathrm{rel}}$. And given that some symmetries relate $H$ to $H^{\dagger}$, \ie I am classifying \emph{pairs} $(H_{\Re},H_{\Im})$ of commuting hermitian operators, it is possible that a second projection, $P_{\mathrm{rel},\dagger}$ enters the game. To move from the abstract to the concrete, let me discuss two examples. Only afterwards, will I detail the general scheme in Section~\ref{my_way_simplified:general_principles}.

\subsubsection[Example 1: a hamiltonian whose spectrum is point and reflection symmetric]{Example 1: a hamiltonian whose \linebreak spectrum is point and reflection symmetric} % (fold)
\label{my_way_simplified:non_hermitian_examples:1}
%
% I will not attempt to see whether the classification proposed here coincides with that of \cite{Zhou_Lee:non_hermitian_topological_classification:2019,Kawabata_Shiozaki_Ueda_Sato:classification_non_hermitian_systems:2019}. Nevertheless, let us illustrate the general principle by means of an example.
Assume the spectrum of my non-hermititan operator is as in Figure~\ref{intro:figure:highly_symmetric_spectrum}, that is it breaks up into four symmetric components which I label $\sigma_{\pm \pm}$, where each sign indicates whether real and imaginary parts are positive or negative. 

The symmetry in the spectrum suggests the presence of (at least) three symmetries. For the sake of argument, let me suppose that $H$ comes furnished with an odd time-reversal symmetry, 
\begin{align*}
	T \, H \, T^{-1} &= + H
	, 
	&&
	T^2 = - \id_{\Hil} 
	, 
\end{align*}
and is pseudohermitian, 
\begin{align*}
	V_{\dagger} \, H \, V_{\dagger}^{-1} &= + H^{\dagger}
	. 
\end{align*}
Moreover, let me suppose $T$ and $V_{\dagger}$ commute, 
\begin{align*}
	[T , V_{\dagger}] = 0 
	. 
\end{align*}
Consequently, their product $T_{\dagger} = V_{\dagger} \, T$ is an odd time-reversal-$\dagger$ symmetry, 
\begin{align*}
	T_{\dagger} \, H \, T_{\dagger}^{-1} &= + H^{\dagger}
	, 
	&&
	T_{\dagger}^2 = - \id_{\Hil} 
	. 
\end{align*}
I can translate the symmetry conditions on $H = H_{\Re} + \ii H_{\Im}$ to symmetry conditions on real and imaginary parts, 
\begin{subequations}\label{my_way_simplified:eqn:example_1_symmetries_real_imaginary_part}
	\begin{align}
		T \, H \, T^{-1} = + H
		\; \; &\Longleftrightarrow \; \; 
		\begin{cases}
			T \, H_{\Re} \, T^{-1} &= + H_{\Re} \\
			T \, H_{\Im} \, T^{-1} &= - H_{\Im} \\
		\end{cases}
		, 
		\label{my_way_simplified:eqn:example_1_symmetries_real_imaginary_part:TR}
		\\
		V_{\dagger} \, H \, V_{\dagger}^{-1} = + H^{\dagger}
		\; \; &\Longleftrightarrow \; \; 
		\begin{cases}
			V_{\dagger} \, H_{\Re} \, V_{\dagger}^{-1} &= + H_{\Re} \\
			V_{\dagger} \, H_{\Im} \, V_{\dagger}^{-1} &= - H_{\Im} \\
		\end{cases}
		,
		\label{my_way_simplified:eqn:example_1_symmetries_real_imaginary_part:pseudo_hermiticity}
		\\
		T_{\dagger} \, H \, T_{\dagger}^{-1} = + H^{\dagger}
		\; \; &\Longleftrightarrow \; \; 
		\begin{cases}
			T_{\dagger} \, H_{\Re} \, T_{\dagger}^{-1} &= + H_{\Re} \\
			T_{\dagger} \, H_{\Im} \, T_{\dagger}^{-1} &= + H_{\Im} \\
		\end{cases}
		\label{my_way_simplified:eqn:example_1_symmetries_real_imaginary_part:TR_dagger}
		. 
	\end{align}
\end{subequations}
A peek at Table~\ref{symmetries:table:overview_symmetries} tells us that the presence of $T$ leads to the spectral symmetry $\sigma(H) = + \overline{\sigma(H)}$ (reflection symmetric about the real axis), and $V_{\dagger}$ is responsible for $\sigma(H) = + \overline{\sigma(H)}$ (reflection symmetric about the real axis). Their product $T_{\dagger} = T V_{\dagger}$ is a third symmetry, an odd time-reversal${}^{\dagger}$ symmetry to be precise; its presence does not give rise to any symmetries in the spectrum. 

Point gap, real line gap, imaginary line gap, this operator possesses them all. So which is the relevant classification for physics? This depends on what states I deem relevant. 
\medskip

\noindent
\paragraph{Choosing the relevant states symmetrically with respect to reflections about the real axis} % (fold)
\label{my_way_simplified:extracting_symmetries_constraints:symmetric_real}
Suppose I designate the states with positive real part to be physically relevant, \ie 
\begin{align*}
	\sigma_{\mathrm{rel}} = \sigma_{++} \cup \sigma_{+-} = + \overline{\sigma_{\mathrm{rel}}}
	. 
\end{align*}
Accordingly, the projection onto the relevant states 
\begin{align*}
	P_{\mathrm{rel}} = 1_{[0,\infty)}(H_{\Re}) \, 1_{\R}(H_{\Im}) 
	= 1_{[0,\infty)}(H_{\Re}) 
\end{align*}
just involves the real part operator — $1_{\R}(H_{\Im}) = \id_{\Hil}$ is just a fancy way of writing the identity. 

Taking a quick peek at equation~\eqref{my_way_simplified:eqn:example_1_symmetries_real_imaginary_part}, I see that both symmetries as well as their product leave the sign of the real part operator $H_{\Re}$ untouched. Consequently, all three operators are symmetries of $P_{\mathrm{rel}}$, 
\begin{subequations}\label{my_way_simplified:eqn:example_1_symmetries_P_rel}
	\begin{align}
		V_{\dagger} \, P_{\mathrm{rel}} \, V_{\dagger}^{-1} &= P_{\mathrm{rel}} 
		, 
		\label{my_way_simplified:eqn:example_1_symmetries_P_rel:U_dagger}
		\\
		T \, P_{\mathrm{rel}} \, T^{-1} &= P_{\mathrm{rel}} 
		, 
		\label{my_way_simplified:eqn:example_1_symmetries_P_rel:TR}
		\\
		T_{\dagger} \, P_{\mathrm{rel}} \, T_{\dagger}^{-1} &= P_{\mathrm{rel}} 
		. 
		\label{my_way_simplified:eqn:example_1_symmetries_P_rel:TR_dagger}
	\end{align}
\end{subequations}
In none of my arguments was it important that $\sigma_{++}$ and $\sigma_{+-}$ were separated by a gap. Indeed, the relevant gap that needs to be maintained is between the spectra to the left and to the right of the imaginary axis. In the parlance of \cite{Kawabata_Shiozaki_Ueda_Sato:classification_non_hermitian_systems:2019}, my system is characterized by a \emph{real line gap}. Hence, I may instead consider the spectrally flattened hamiltonian~\eqref{my_way_simplified:eqn:spectrally_flattened_hamiltonian}. All symmetries and constraints of $P_{\mathrm{rel}}$ manifest themselves as symmetries of $Q = \id_{\Hil} - 2 P_{\mathrm{rel}}$. The time-reversal symmetry $T$ translates to a time-reversal symmetry of $Q = + T \, Q \, T^{-1}$. Importantly, the pseudo-hermiticity condition translates to an \emph{ordinary symmetry} of $Q$, 
\begin{align*}
	V_{\dagger} \, Q \, V_{\dagger}^{-1} &= Q 
	. 
\end{align*}
Similarly, the odd time-reversal-$\dagger$ symmetry acts as an odd (regular, non-$\dagger$!) time-reversal symmetry, 
\begin{align*}
	T_{\dagger} \, Q \, T_{\dagger}^{-1} &= Q 
	. 
\end{align*}
This definition gives the same spectrally flattened operator as the deformation procedure indicated in \cite[Figure~2~(b)]{Kawabata_Shiozaki_Ueda_Sato:classification_non_hermitian_systems:2019}. However, unlike Kawabata et al.\ the spectrally flattened hamiltonian is not constructed through a deformation of $H$, that is completely unnecessary. The procedure described here defines $P_{\mathrm{rel}}$ \emph{directly} via functional calculus and is \emph{inherently more general}. For instance, my definition of $Q$ still applies to an operator whose spectrum is as in Figure~\ref{38_fold_recap:figure:nested_Cs_spectrum}: I could still declare one of the (inverted) C-shaped parts of the spectrum as relevant and then define $P_{\mathrm{rel}} = 1_{\sigma_{\mathrm{rel}}}(H)$. But since the two Cs overlap in the angular direction, such an operator does not possess a real line gap in the sense of \cite{Kawabata_Shiozaki_Ueda_Sato:classification_non_hermitian_systems:2019}. And unitary flattening will give an operator whose spectrum is the entire circle line without any gaps. 

To summarize, I am left with the problem of giving a topological classification for a projection $P_{\mathrm{rel}}$ that possesses two antilinear symmetries (which act like odd time-reversal-type symmetries on the “Fermi projection”) and one linear, commuting symmetry. One way to obtain a complete topological classification for periodic operators is to combine \cite[Theorem~4.9]{DeNittis_Lein:symmetries_electromagnetism:2020} with the classification of class~AII vector bundles by De~Nittis and Gomi \cite{DeNittis_Gomi:AII_bundles:2014}. 

For the benefit of the readers, I will outline the arguments from \cite{DeNittis_Lein:symmetries_electromagnetism:2020}; for this part of the classification, periodicity is not necessary. The operator $V_{\dagger} = V_{\dagger}^{\dagger} = V_{\dagger}^{-1}$ is a hermitian unitary, and given that I am considering a non-trivial symmetry $V_{\dagger} \neq \pm \id_{\Hil}$, the spectrum of $V_{\dagger}$ equals $\{ -1 , +1 \}$. Since $V_{\dagger}$ commutes with the projection, 
\begin{align*}
	P_{\mathrm{rel}} = P_{\mathrm{rel},+} + P_{\mathrm{rel},-} 
\end{align*}
splits into two parts, one that projects onto relevant states that are also eigenvectors of $V_{\dagger}$ to $+1$ and those to $-1$. Moreover, given that $V_{\dagger}$ commutes with the time-reversal symmetry $T$ as well as $T_{\dagger} = V_{\dagger} \, T$, I deduce that $T$ and $T_{\dagger}$ are block-diagonal and are therefore symmetries for the two components separately, 
\begin{align*}
	T \, P_{\mathrm{rel},\pm} \, T^{-1} &= P_{\mathrm{rel},\pm}
	, 
	\\
	T_{\dagger} \, P_{\mathrm{rel},\pm} \, T_{\dagger}^{-1} &= P_{\mathrm{rel},\pm}
	. 
\end{align*}
And according to \cite[equation~(4.14)]{DeNittis_Lein:symmetries_electromagnetism:2020} $T$ and $T_{\dagger}$ are identical on the subspaces $\ran P_{\mathrm{rel},\pm}$ up to a global sign. Consequently, I should not think of $T$ and $T_{\dagger}$ as two distinct symmetries. Now I am left with \emph{two} projections $P_{\mathrm{rel},\pm}$ and a \emph{single} odd time-reversal-type symmetry, say, $T$. Consequently, I will obtain two sets of topological invariants for class~AII systems that are independent of one another, one for the $+1$ sub bundle of $V_{\dagger}$ and the other to the $-1$ sub bundle. 

One way to obtain a complete list of invariants is to add the assumption that $H$ is periodic and then use vector bundle theory; however, the classification is expected to remain valid even when weak disorder is present. Thanks to the periodicity, the projections give rise to a family of projections $P_{\mathrm{rel},\pm}(k)$ indexed by Bloch momentum; the time-reversal symmetry manifests itself as $T \, P_{\mathrm{rel},\pm}(k) \, T^{-1} = P_{\mathrm{rel},\pm}(-k)$. From the family of projections, I can construct vector bundles from them by gluing together the ranges of $P_{\mathrm{rel},\pm}(k)$ over the entire Brillouin torus, which inherits the odd time-reversal symmetry. The precise definition is a bit technical and I refer to De~Nittis and Gomi \cite[Section~2]{DeNittis_Gomi:AII_bundles:2014}. Then I can read off the topological classification from Theorems~1.4 and 1.7 in \cite{DeNittis_Gomi:AII_bundles:2014} for dimensions $d \leq 4$. The result is summarized in Table~\ref{my_way_simplified:table:example_1_real_line_gap}. I can compare this with \cite[Table~IX]{Kawabata_Shiozaki_Ueda_Sato:classification_non_hermitian_systems:2019}, more specifically the real line gap classification for the case AII, $\eta_+$. Keep in mind that Kawabata et al.\ only list the strong invariants. Here, I see that — at least as far as the strong invariants are concerned — the two classifications agree. In fact, my classification is finer, because I am able to give the weak invariants as well. 

\begin{table}[t]
	\begin{center}
		% \newcolumntype{A}{>{\centering\arraybackslash\normalsize}m{25mm}}
		% \newcolumntype{B}{>{\centering\arraybackslash\normalsize}m{15mm}}
		\renewcommand{\arraystretch}{1.5}
		\begin{tabular}{c | c | c | c | c}
			\emph{Classification} & $d = 1$ & $d = 2$ & $d = 3$ & $d = 4$ \\ \hline \hline 
			$P_{\mathrm{rel},\pm}$ class~AII & $0$ & $\Z_2 \oplus \Z_2$ & $\Z_2^4 \oplus \Z_2^4$ & $\Z_2^{10} \oplus \Z \oplus \Z_2^{10} \oplus \Z$ \\ \hline
			\cite{Kawabata_Shiozaki_Ueda_Sato:classification_non_hermitian_systems:2019} for AII, $\eta_+$, $\mathrm{L_r}$ & $0$ & $\Z_2 \oplus \Z_2$ & $\Z_2 \oplus \Z_2$ & $\Z \oplus \Z$ \\
		\end{tabular}
	\end{center}
	\caption{Classification obtained here for $\sigma_{\mathrm{rel}} = \sigma_{++} \cup \sigma_{+-} = \overline{\sigma_{\mathrm{rel}}}$ in dimensions $d \leq 4$ and comparison with the real line gap classification obtained in \cite[Table~IX]{Kawabata_Shiozaki_Ueda_Sato:classification_non_hermitian_systems:2019} for the case AII, $\eta_+$. One-dimensional class AII vector bundles are all trivial. In dimension $2$, the two $\Z_2$-valued invariants are the Kane-Mele invariants for the two vector bundles generated from $P_{\mathrm{rel},+}$ and $P_{\mathrm{rel},-}$. When $d = 3$, I likewise obtain the two top Kane-Mele invariants as well as two sets of three $\Z_2$-valued weak invariants. And for $d = 4$, the two $\Z$-valued invariants are the second Chern numbers of the two vector bundles, which are complemented with 10 $\Z_2$-valued weak invariants. Note that the classification obtained here includes both, strong and weak invariants whereas Kawabata et al.\ only give strong invariants. }
	\label{my_way_simplified:table:example_1_real_line_gap}
\end{table}
%
% paragraph symmetrically_chosen_relevant_states (end)
\medskip

\noindent
\paragraph{Asymmetrically chosen relevant states} % (fold)
\label{my_way_simplified:extracting_symmetries_constraints:asymmetric}
So what happens when I choose a different set of relevant states? Suppose only states associated with $\sigma_{\mathrm{rel}} = \sigma_{++}$ are deemed relevant. Clearly, this breaks all spectral symmetries, including $\overline{\sigma_{++}} = \sigma_{+-} \neq \sigma_{++}$. Writing out 
\begin{align*}
	P_{\mathrm{rel}} = 1_{\sigma_{++}}(H) = 1_{[0,\infty)}(H_{\Re}) \; 1_{[0,\infty)}(H_{\Im})
\end{align*}
and using the symmetries~\eqref{my_way_simplified:eqn:example_1_symmetries_real_imaginary_part} of real and imaginary part, leads to the following relations 
\begin{align*}
	T \, P_{\mathrm{rel}} \, T^{-1} &= 1_{\overline{\sigma_{++}}}(H) = 1_{\sigma_{++}}(H^{\dagger}) 
	= P_{\mathrm{rel},\dagger}
	\\ 
	V_{\dagger} \, P_{\mathrm{rel}} \, V_{\dagger}^{-1} &= 1_{\overline{\sigma_{++}}}(H) = 1_{\sigma_{++}}(H^{\dagger}) 
	= P_{\mathrm{rel},\dagger} 
\end{align*}
between $P_{\mathrm{rel}}$ and a similarly defined projection 
\begin{align*}
	P_{\mathrm{rel},\dagger} &= 1_{\sigma_{++}}(H^{\dagger}) 
	= 1_{[0,\infty)}(H_{\Re}) \; 1_{[0,\infty)}(-H_{\Im})
	\\
	&= 1_{[0,\infty)}(H_{\Re}) \; 1_{(-\infty,0]}(H_{\Im})
	. 
\end{align*}
Clearly, $T$ and $V_{\dagger}$ \emph{separately} are no longer symmetries of $P_{\mathrm{rel}}$, but their product is! And since $T$ and $V_{\dagger}$ commute by assumption, the product $T_{\dagger} = V_{\dagger} \, T$ is \emph{still} an odd time-reversal symmetry of both, $P_{\mathrm{rel}}$ and $P_{\mathrm{rel},\dagger}$, 
\begin{align*}
	T_{\dagger} \, P_{\mathrm{rel}} \, T_{\dagger}^{-1} &= P_{\mathrm{rel}} 
	, 
	\\
	T_{\dagger} \, P_{\mathrm{rel},\dagger} \, T_{\dagger}^{-1} &= P_{\mathrm{rel},\dagger} 
	. 
\end{align*}
Hence, $Q$ possesses less symmetries than the specimen considered in the last subsection, 
\begin{subequations}\label{my_way_simplified:eqn:example_2_symmetries_P_rel}
	\begin{align}
		V_{\dagger} \, Q \, V_{\dagger}^{-1} &= Q_{\dagger} 
		\overset{\mathrm{def}}{=} (\id - P_{\mathrm{rel},\dagger}) - P_{\mathrm{rel},\dagger}
		, 
		\label{my_way_simplified:eqn:example_2_symmetries_P_rel:U_dagger}
		\\
		T \, Q \, T^{-1} &= Q_{\dagger} 
		,
		\label{my_way_simplified:eqn:example_2_symmetries_P_rel:TR}
		\\
		T_{\dagger} \, Q \, T_{\dagger}^{-1} &= Q 
		. 
		\label{my_way_simplified:eqn:example_2_symmetries_P_rel:TR_dagger}
	\end{align}
\end{subequations}
This is an important point: \emph{it is not the symmetries of $H$, but the symmetries and constraints of $P_{\mathrm{rel}}$ or, equivalently, the spectrally flattened hamiltonian $Q$ which matter.} And my choice of relevant states may break some or all of the symmetries that $H$ possesses. 

So let me play the classification game again: I have two projections $P_{\mathrm{rel}}$ and $P_{\mathrm{rel},\dagger}$, and each comes the odd time-reversal symmetry $T_{\dagger}$. However, these two projections are unitarily equivalent, equation~\eqref{my_way_simplified:eqn:example_2_symmetries_P_rel:U_dagger}, so topologically speaking, $P_{\mathrm{rel}}$ and $P_{\mathrm{rel},\dagger}$ are the same. Nevertheless, I may be able to define the relative index~\eqref{general:eqn:index_two_projections} for the pair of projections, which is also a topological quantity. At present, it is unclear whether this index is well-defined and how this relative index manifests itself physically. 

Specializing to the periodic case once more, I obtain only a \emph{single} class~AII vector bundle, and my classification is “half” of that given in the previous subsection. This corresponds to the point gap case in \cite{Kawabata_Shiozaki_Ueda_Sato:classification_non_hermitian_systems:2019} for class~AII, $\eta_+$: it is possible to deform the operator so that the relevant spectrum crosses the real or imaginary axis, there is no symmetry that would forbid this. I have summarized the findings in Table~\ref{my_way_simplified:table:example_1_point_gap}. 

\begin{table}
	\begin{center}
		\renewcommand{\arraystretch}{1.5}
		\begin{tabular}{c | c | c | c | c}
			\emph{Classification} & $d = 1$ & $d = 2$ & $d = 3$ & $d = 4$ \\ \hline \hline 
			Index~\eqref{general:eqn:index_two_projections} + $P_{\mathrm{rel}}$ class~AII & $\Z \oplus 0$ & $\Z \oplus \Z_2$ & $\Z \oplus \Z_2^4$ & $\Z \oplus \Z_2^{10} \oplus \Z$ \\ \hline
			\cite{Kawabata_Shiozaki_Ueda_Sato:classification_non_hermitian_systems:2019} for AII, $\eta_+$, $\mathrm{P}$ & $0$ & $\Z_2$ & $\Z_2$ & $\Z$ \\
		\end{tabular}
	\end{center}
	\caption{Classification obtained here for $\sigma_{\mathrm{rel}} = \sigma_{++} \neq \pm \sigma_{\mathrm{rel}} , \pm \overline{\sigma_{\mathrm{rel}}}$ in dimensions $d \leq 4$ and comparison with the point gap classification obtained in \cite[Table~IX]{Kawabata_Shiozaki_Ueda_Sato:classification_non_hermitian_systems:2019} for class~AII, $\eta_+$. In my classification, there is the possibility of an additional $\Z$-valued relative index~\eqref{general:eqn:index_two_projections} that stems from comparing the projections $P_{\mathrm{rel}}$ and $P_{\mathrm{rel},\dagger}$. The remaining contributions come from standard theory. One-dimensional class AII vector bundles are all trivial. In dimension $2$, the $\Z_2$-valued invariant is the Kane-Mele invariant for the class~AII vector bundle generated from $P_{\mathrm{rel}}$ alone. When $d = 2$, I likewise obtain the top Kane-Mele invariants as well as three other weak $\Z_2$-valued invariants. And for $d = 4$, the two $\Z$-valued invariants are the second Chern number and the relative index~\eqref{general:eqn:index_two_projections} for the two projections $P_{\mathrm{rel}}$ and $P_{\mathrm{rel},\dagger}$. Note that the classification obtained here includes both, strong and weak invariants whereas Kawabata et al.\ only give strong invariants. }
	\label{my_way_simplified:table:example_1_point_gap}
\end{table}
%
% paragraph asymmetrically_chosen_relevant_states (end)
\medskip

\noindent
\paragraph{Choosing the relevant states symmetrically with respect to reflections about the imaginary axis} % (fold)
\label{my_way_simplified:extracting_symmetries_constraints:symmetric_imaginary}
A second “symmetric” choice is to declare states above the real line to be relevant, 
\begin{align*}
	\sigma_{\mathrm{rel}} = \sigma_{++} \cup \sigma_{-+} = - \overline{\sigma_{\mathrm{rel}}} 
	, 
\end{align*}
and then proceed with the analysis. In this case the relevant projection
\begin{align*}
	P_{\mathrm{rel}} = 1_{\R}(H_{\Re}) \; 1_{[0,\infty)}(H_{\Im})
	= 1_{[0,\infty)}(H_{\Im})
\end{align*}
only depends on the imaginary part, and two of the symmetries flip the sign of $H_{\Im}$, I obtain two constraints and one odd time-reversal symmetry for the projection, 
\begin{align*}
	V_{\dagger} \, P_{\mathrm{rel}} \, V_{\dagger}^{-1} &= \id_{\Hil} - P_{\mathrm{rel}} 
	, 
	\\
	T \, P_{\mathrm{rel}} \, T^{-1} &= \id_{\Hil} - P_{\mathrm{rel}} 
	, 
	\\
	T_{\dagger} \, P_{\mathrm{rel}} \, T_{\dagger}^{-1} &= P_{\mathrm{rel}} 
	. 
\end{align*}
Because the $\dagger$-projection 
\begin{align*}
	P_{\mathrm{rel},\dagger} &= 1_{\sigma_{\mathrm{rel}}}(H^{\dagger}) = 1_{[0,\infty)}(-H_{\Im})
	\\
	&= \id_{\Hil} - P_{\mathrm{rel}} 
\end{align*}
coincides with the projection onto the orthogonal complement of the relevant states, $P_{\mathrm{rel}}$ and $P_{\mathrm{rel},\dagger}$ are again not independent projections. In principle, this may mean that a relative index~\eqref{general:eqn:index_two_projections} between the two projections may become relevant for the classification. 

Rephrasing the three symmetries in terms of the spectrally flattened hamiltonian $Q = - Q_{\dagger}$, 
\begin{align*}
	V_{\dagger} \, Q \, V_{\dagger}^{-1} &= - Q 
	, 
	\\
	T \, Q \, T^{-1} &= - Q 
	, 
	\\
	T_{\dagger} \, Q \, T_{\dagger}^{-1} &= Q 
	, 
\end{align*}
shows that I am classifying an operator with a chiral, an odd particle-hole and an odd time-reversal symmetry, \ie I am dealing with an operator from Cartan-Altland-Zirnbauer class~CII. While I am not aware of an exhaustive classification of class~CII, I know its strong invariants (\cf \eg \cite[Table~I]{Chiu_Teo_Schnyder_Ryu:classification_topological_insulators:2016}), and I see that this is in perfect agreement with the imaginary line gap classification in \cite[Table~IX]{Kawabata_Shiozaki_Ueda_Sato:classification_non_hermitian_systems:2019}; I have summarized the result in Table~\ref{my_way_simplified:table:example_1_imaginary_line_gap}. 

\begin{table}
	\begin{center}
		\renewcommand{\arraystretch}{1.5}
		\begin{tabular}{c | c | c | c | c | c}
			\emph{Classification} & $d = 0$ & $d = 1$ & $d = 2$ & $d = 3$ & $d = 4$ \\ \hline \hline 
			Index~\eqref{general:eqn:index_two_projections} + $P_{\mathrm{rel}}$ class~CII & $\Z \oplus 0$ & $\Z \oplus 2 \Z$ & $\Z \oplus 0$ & $\Z \oplus \Z_2$ & $\Z \oplus \Z_2$ \\ \hline
			\cite{Kawabata_Shiozaki_Ueda_Sato:classification_non_hermitian_systems:2019} for AII, $\eta_+$, $\mathrm{L_i}$ & $0$ & $2 \Z$ & $0$ & $\Z_2$ & $\Z_2$ \\
		\end{tabular}
	\end{center}
	\caption{Classification obtained here for $\sigma_{\mathrm{rel}} = \sigma_{++} \cup \sigma_{-+} = - \overline{\sigma_{\mathrm{rel}}}$ in dimensions $d \leq 4$ and comparison with the imaginary line gap classification obtained in \cite[Table~IX]{Kawabata_Shiozaki_Ueda_Sato:classification_non_hermitian_systems:2019} for class~AII, $\eta_+$. 
	My analysis shows I am dealing with an operator of Cartan-Altland-Zirnbauer class~CII, which is part of the canonical literature (\eg \cite[Table~I]{Chiu_Teo_Schnyder_Ryu:classification_topological_insulators:2016}). The agreement extends to the other dimensions $d = 5 , 6 , 7$ and thus, thanks to Bott periodicity, to all dimensions. 
	As before, the relative index~\eqref{general:eqn:index_two_projections} between the projections $P_{\mathrm{rel}}$ and $P_{\mathrm{rel},\dagger} = \id_{\Hil} - P_{\mathrm{rel}}$ may give an additional topological invariant. }
	\label{my_way_simplified:table:example_1_imaginary_line_gap}
\end{table}
%
% paragraph there_is_a_third_case (end)
\medskip

\noindent
\paragraph{Choosing the relevant states point-symmetrically} % (fold)
\label{par:choosing_the_relevant_states_point_symmetrically}
So far the classification obtained here is in perfect agreement with the classification of \cite{Zhou_Lee:non_hermitian_topological_classification:2019,Kawabata_Shiozaki_Ueda_Sato:classification_non_hermitian_systems:2019}. However, looking at the spectrum of the operator, there is one more case, namely when 
\begin{align*}
	\sigma_{\mathrm{rel}} = \sigma_{++} \cup \sigma_{--}
	= - \sigma_{\mathrm{rel}}
\end{align*}
is chosen point-symmetrically. This choice does not correspond to any line gap in \cite{Kawabata_Shiozaki_Ueda_Sato:classification_non_hermitian_systems:2019}. 

Repeating the symmetry analysis once more reveals that $P_{\mathrm{rel}}$ and $P_{\mathrm{rel},\dagger}$ satisfy all the same relations as in the case $\sigma_{\mathrm{rel}} = \sigma_{++} \cup \sigma_{-+}$ (Table~\ref{my_way_simplified:table:example_1_point_symmetric_sigma_rel}). So also here the system's topology is classified as class~CII — even though it is not sensible to call my choice of relevant spectrum as having an imaginary line gap in the sense of \cite{Kawabata_Shiozaki_Ueda_Sato:classification_non_hermitian_systems:2019}. 

\begin{table}
	\begin{center}
		\renewcommand{\arraystretch}{1.5}
		\begin{tabular}{c | c | c | c | c | c}
			\emph{Classification} & $d = 0$ & $d = 1$ & $d = 2$ & $d = 3$ & $d = 4$ \\ \hline \hline 
			Index~\eqref{general:eqn:index_two_projections} + $P_{\mathrm{rel}}$ class~CII & $\Z \oplus 0$ & $\Z \oplus 2 \Z$ & $\Z \oplus 0$ & $\Z \oplus \Z_2$ & $\Z \oplus \Z_2$ \\ \hline
			\cite{Kawabata_Shiozaki_Ueda_Sato:classification_non_hermitian_systems:2019} not applicable & —  & — & — & — & — \\ 
		\end{tabular}
	\end{center}
	\caption{Classification obtained here for $\sigma_{\mathrm{rel}} = \sigma_{++} \cup \sigma_{--} = - \sigma_{\mathrm{rel}}$ in dimensions $d \leq 4$. This point-symmetric case does not fit into the point gap/line gap classification scheme of \cite{Kawabata_Shiozaki_Ueda_Sato:classification_non_hermitian_systems:2019}. However, since $P_{\mathrm{rel}}$ possesses the same number and types of symmetries as in the case $\sigma_{\mathrm{rel}} = \sigma_{++} \cup \sigma_{-+}$, their classifications coincide. }
	\label{my_way_simplified:table:example_1_point_symmetric_sigma_rel}
\end{table}
%
% paragraph choosing_the_relevant_states_point_symmetrically (end)
% subsubsection example_1 (end)

\subsubsection{Example 2: a hamiltonian with point-symmetric spectrum} % (fold)
\label{my_way_simplified:non_hermitian_examples:2}
The first example was completely consistent with the classification of \cite{Zhou_Lee:non_hermitian_topological_classification:2019,Kawabata_Shiozaki_Ueda_Sato:classification_non_hermitian_systems:2019} — with the exception of the last case, which is not covered by the literature. The second example is more interesting: I consider another operator that merely has inversion-symmetric spectrum. This is the case if, say, $H$ could possess a chiral symmetry, 
\begin{align*}
	S \, H \, S^{-1} &= - H 
	, 
\end{align*}
and an even particle-hole${}^{\dagger}$ symmetry, 
\begin{align*}
	C_{\dagger} \, H \, C_{\dagger}^{-1} &= - H^{\dagger} 
	, 
\end{align*}
which I again assume to commute, 
\begin{align*}
	[C_{\dagger} , S] = 0
	. 
\end{align*}
Their product $T_{\dagger} = C_{\dagger} S$ is an even time-reversal${}^{\dagger}$ symmetry, 
\begin{align*}
	T_{\dagger} \, H \, T_{\dagger}^{-1} &= + H^{\dagger} 
	. 
\end{align*}
As before, it helps tremendously to express the symmetry action in terms of real and imaginary part operators: 
\begin{subequations}\label{my_way_simplified:eqn:example_2_symmetries_real_imaginary_part}
	\begin{align}
		S \, H \, S^{-1} = + H
		\; \; &\Longleftrightarrow \; \; 
		\begin{cases}
			S \, H_{\Re} \, S^{-1} &= - H_{\Re} \\
			S \, H_{\Im} \, S^{-1} &= - H_{\Im} \\
		\end{cases}
		, 
		\label{my_way_simplified:eqn:example_2_symmetries_real_imaginary_part:chiral}
		\\
		C_{\dagger} \, H \, C_{\dagger}^{-1} = + H^{\dagger}
		\; \; &\Longleftrightarrow \; \; 
		\begin{cases}
			C_{\dagger} \, H_{\Re} \, C_{\dagger}^{-1} &= - H_{\Re} \\
			C_{\dagger} \, H_{\Im} \, C_{\dagger}^{-1} &= - H_{\Im} \\
		\end{cases}
		,
		\label{my_way_simplified:eqn:example_2_symmetries_real_imaginary_part:PH_dagger}
		\\
		T_{\dagger} \, H \, T_{\dagger}^{-1} = + H^{\dagger}
		\; \; &\Longleftrightarrow \; \; 
		\begin{cases}
			T_{\dagger} \, H_{\Re} \, T_{\dagger}^{-1} &= + H_{\Re} \\
			T_{\dagger} \, H_{\Im} \, T_{\dagger}^{-1} &= + H_{\Im} \\
		\end{cases}
		\label{my_way_simplified:eqn:example_2_symmetries_real_imaginary_part:TR_dagger}
		. 
	\end{align}
\end{subequations}
The presence of $S$ and $C_{\dagger}$ symmetry both impose $\sigma(H) = - \sigma(H)$. Nevertheless, I shall continue to assume that $\sigma(H)$ splits into four spectral islands akin to Figure~\ref{intro:figure:highly_symmetric_spectrum} that come in two pairs as in Figure~\ref{my_way_simplified:figure:point_symmetric_spectrum}. 

\begin{figure}[t]
	\begin{centering}
		\resizebox{70mm}{!}{\includegraphics{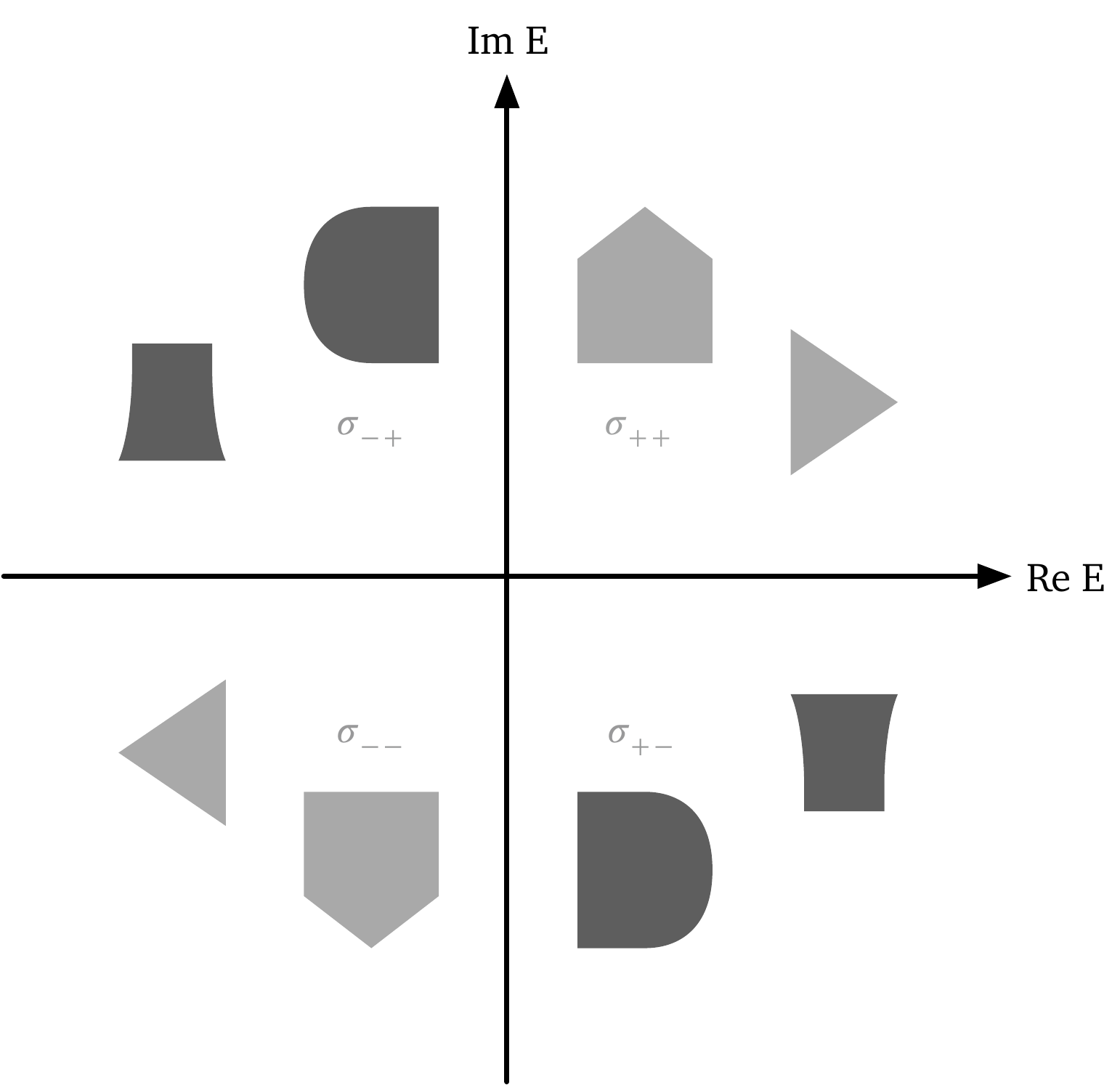}}
	\end{centering}
	\caption{Spectrum with point symmetry only. Very often such symmetries are due to the presence of discrete symmetries such as parity and particle-hole${}^{\dagger}$ symmetries. }
	\label{my_way_simplified:figure:point_symmetric_spectrum}
\end{figure}

There are three distinct cases, which I will analyze in turn below: I can choose $\sigma_{\mathrm{rel}}$ asymmetrically, point-symmetrically or reflection-symmetrically. The three topological classifications are summarized in Table~\ref{my_way_simplified:table:example_2_comparison_correct_classification} (\emph{not} Table~\ref{my_way_simplified:table:example_2_comparison_false_classification}). 

Let me compare that with the classification obtained by Kawabata et al. At first glance, it seems that this operator is — in the notation of \cite{Kawabata_Shiozaki_Ueda_Sato:classification_non_hermitian_systems:2019} — of class BDI${}^{\dagger}$ and their classifications, summarized in \cite[Table~V]{Kawabata_Shiozaki_Ueda_Sato:classification_non_hermitian_systems:2019}, are seemingly in perfect agreement. Unfortunately, BDI${}^{\dagger}$ is the false topological class to compare my classification to. 

\begin{table}
	\begin{center}
		\renewcommand{\arraystretch}{1.5}
		\begin{tabular}{c | c | c | c | c}
			\emph{Classification} & $d = 1$ & $d = 2$ & $d = 3$ & $d = 4$ \\ \hline \hline 
			$\sigma_{\mathrm{rel}} = \sigma_{++}$ & $0$ & $0$ & $0$ & $2 \Z$ \\ \hline
			\cite{Kawabata_Shiozaki_Ueda_Sato:classification_non_hermitian_systems:2019} for BDI${}^{\dagger}$, $\mathrm{P}$ & $0$ & $0$ & $0$ & $2 \Z$ \\ \hline \hline 
			$\sigma_{\mathrm{rel}} = \sigma_{++} \cup \sigma_{+-}$ & $\Z$ & $0$ & $0$ & $0$ \\ \hline
			\cite{Kawabata_Shiozaki_Ueda_Sato:classification_non_hermitian_systems:2019} for BDI${}^{\dagger}$, $\mathrm{L_r}$ & $\Z$ & $0$ & $0$ & $0$ \\ \hline \hline 
			$\sigma_{\mathrm{rel}} = \sigma_{++} \cup \sigma_{--}$ & $0$ & $0$ & $0$ & $2 \Z \oplus 2 \Z$ \\ \hline
			\cite{Kawabata_Shiozaki_Ueda_Sato:classification_non_hermitian_systems:2019} for BDI${}^{\dagger}$, $\mathrm{L_i}$ & $0$ & $0$ & $0$ & $2 \Z \oplus 2 \Z$ \\
		\end{tabular}
	\end{center}
	\caption{Comparison of the classifications obtained here for the three choices of relevant spectrum to class BDI${}^{\dagger}$ in \cite[Table~V]{Kawabata_Shiozaki_Ueda_Sato:classification_non_hermitian_systems:2019}. 
	While it seems BDI${}^{\dagger}$ is the correct class in \cite{Kawabata_Shiozaki_Ueda_Sato:classification_non_hermitian_systems:2019} to compare it to and it looks as if the two classifications agree mathematically, this is the \emph{wrong class} in \cite{Kawabata_Shiozaki_Ueda_Sato:classification_non_hermitian_systems:2019} to compare it to. The reason lies with differences how symmetries are labeled here and in \cite{Kawabata_Shiozaki_Ueda_Sato:classification_non_hermitian_systems:2019}. The correct comparison can be found in Table~\ref{my_way_simplified:table:example_2_comparison_correct_classification}. }
	\label{my_way_simplified:table:example_2_comparison_false_classification}
\end{table}

Rather, in the notation of \cite{Kawabata_Shiozaki_Ueda_Sato:classification_non_hermitian_systems:2019} the system is of class~D, $\mathcal{S}_+$, which admits a point- and a line-type subclass, but makes no distinction between real and imaginary line gaps. No doubt, it is easy to make this misclassification because of the subtle differences in notation between this work and \cite{Kawabata_Shiozaki_Ueda_Sato:classification_non_hermitian_systems:2019} (see Table~\ref{symmetries:table:overview_symmetries}). Indeed, in Kawabata et al.'s notation $S$ is a commuting sublattice symmetry, $C_{\dagger}$ a particle-hole symmetry (no $\dagger$) and $T$ a time-reversal${}^{\dagger}$ symmetry (\emph{with} $\dagger$). 

The classification scheme developed here deviates from Kawabata et al.'s: the spectrum suggests three essentially distinct rather than two choices; I can choose the spectrum asymmetrically, point-symmetrically or reflection-symmetrically. 

Our classifications only agree in one of three cases, \cf Table~\ref{my_way_simplified:table:example_2_comparison_correct_classification}, namely when I pick the spectrum reflection symmetrically. This corresponds to the line gap classification in \cite{Kawabata_Shiozaki_Ueda_Sato:classification_non_hermitian_systems:2019}. The asymmetric case does not agree with the point gap classification as one might have expected. And as before the point-symmetric choice for $\sigma_{\mathrm{rel}}$ falls outside Kawabata et al.'s classification scheme.  

\begin{table}
	\begin{center}
		\renewcommand{\arraystretch}{1.5}
		\begin{tabular}{c | c | c | c | c}
			\emph{Classification} & $d = 1$ & $d = 2$ & $d = 3$ & $d = 4$ \\ \hline \hline 
			$\sigma_{\mathrm{rel}} = \sigma_{++}$ & $0$ & $0$ & $0$ & $2 \Z$ \\ \hline
			\cite{Kawabata_Shiozaki_Ueda_Sato:classification_non_hermitian_systems:2019} for D, $\mathcal{S}_+$, $\mathrm{P}$ & $\Z$ & $0$ & $\Z$ & $0$ \\ \hline \hline 
			$\sigma_{\mathrm{rel}} = \sigma_{++} \cup \sigma_{+-}$ & $\Z$ & $0$ & $0$ & $0$ \\ \hline
			\cite{Kawabata_Shiozaki_Ueda_Sato:classification_non_hermitian_systems:2019} for D, $\mathcal{S}_+$, $\mathrm{L}$ & $\Z$ & $0$ & $0$ & $0$ \\ \hline \hline 
			$\sigma_{\mathrm{rel}} = \sigma_{++} \cup \sigma_{--}$ & $0$ & $0$ & $0$ & $2 \Z \oplus 2 \Z$ \\ \hline
			\cite{Kawabata_Shiozaki_Ueda_Sato:classification_non_hermitian_systems:2019} not applicable & — & — & — & — \\
		\end{tabular}
	\end{center}
	\caption{The table compares the classifications obtained here for the three choices of relevant spectrum to that obtained by Kawabata et al. for the symmetry class~D, $\mathcal{S}_+$ (\cf \cite[Table~VII]{Kawabata_Shiozaki_Ueda_Sato:classification_non_hermitian_systems:2019}). Only for the two reflection-symmetric cases where \eg $\sigma_{\mathrm{rel}} = \sigma_{++} \cup \sigma_{+-} = + \overline{\sigma_{\mathrm{rel}}}$ do they agree. }
	\label{my_way_simplified:table:example_2_comparison_correct_classification}
\end{table}
\medskip

\noindent
\paragraph{Asymmetrially chosen relevant spectrum} % (fold)
\label{par:asymmetrially_chosen_relevant_spectrum}
When I choose $\sigma_{\mathrm{rel}} = \sigma_{++}$, a quick glance at Table~\ref{symmetries:table:overview_symmetries} reveals that $S$ and $C_{\dagger}$ are broken, \ie they map 
\begin{align*}
	P_{\mathrm{rel}} &= 1_{[0,\infty)}(H_{\Re}) \; 1_{[0,\infty)}(H_{\Im})
\end{align*}
neither to itself, 
\begin{align*}
	P_{\mathrm{rel},\dagger} &= 1_{[0,\infty)}(H_{\Re}) \; 1_{(-\infty,0]}(H_{\Im})
\end{align*}
or their orthogonal complements. But $P_{\mathrm{rel}}$ and $P_{\mathrm{rel},\dagger}$ possesses the even time-reversal-type symmetry $T_{\dagger} = C_{\dagger} S$, 
\begin{align*}
	T_{\dagger} \, P_{\mathrm{rel}} \, T_{\dagger}^{-1} &= P_{\mathrm{rel}}
	, 
	\\
	T_{\dagger} \, P_{\mathrm{rel},\dagger} \, T_{\dagger}^{-1} &= P_{\mathrm{rel},\dagger}
	. 
\end{align*}
However, the states in $P_{\mathrm{rel},\dagger}$ are irrelevant and there is no relation between $P_{\mathrm{rel}}$ and its daggered counterpart. Therefore, it does not enter our classification. 

Consequently, I need to apply the classification of class~AI vector bundles, which is trivial in $d \leq 3$ and for $d = 4$ isomorphic class~AI vector bundles are classified by their second Chern number (\cf \cite[Theorem~1.6]{DeNittis_Gomi:AII_bundles:2014}). 
% paragraph asymmetrially_chosen_relevant_spectrum (end)
\medskip

\noindent
\paragraph{Point-symmetrically chosen relevant states} % (fold)
\label{par:point_symmetrically_chosen_relevant_states}
The second distinct case is when I pick the spectrum point symmetrically, \eg 
\begin{align*}
	\sigma_{\mathrm{rel}} = \sigma_{++} \cup \sigma_{--} = - \sigma_{\mathrm{rel}} 
	. 
\end{align*}
Then $P_{\mathrm{rel}}$ has a complete set of symmetries, 
\begin{align*}
	S \, P_{\mathrm{rel}} \, S^{-1} &= P_{\mathrm{rel}} 
	, 
	\\
	C_{\dagger} \, P_{\mathrm{rel}} \, C_{\dagger}^{-1} &= P_{\mathrm{rel}} 
	, 
	\\
	T_{\dagger} \, P_{\mathrm{rel}} \, T_{\dagger}^{-1} &= P_{\mathrm{rel}}
	. 
\end{align*}
So for the purpose of the classification, the action of $C_{\dagger}$ and $T_{\dagger}$ is that of two even time-reversal-type symmetries — just like in Example~1 when the states were chosen reflection-symmetrically. A classification of such systems has been obtained in \cite[Theorems~4.5 and 4.9]{DeNittis_Lein:symmetries_electromagnetism:2020}, where it was shown that such systems can be viewed as $2 \times \mathrm{AI}$. Consequently, up until dimension $3$, those systems as topologically trivial and in dimension $4$, the system is characterized by two second Chern numbers. 
% paragraph point_symmetrically_chosen_relevant_states (end)
\medskip

\noindent
\paragraph{Reflection-symmetrically chosen relevant states} % (fold)
\label{par:reflection_symmetrically_chosen_relevant_states}
Should I designate all of the spectrum to the right of the imaginary axis or above the real axis as relevant, \ie 
\begin{align*}
	\sigma_{\mathrm{rel}} = \sigma_{++} \cup \sigma_{+-} = + \overline{\sigma_{\mathrm{rel}}} 
\end{align*}
or $\sigma_{\mathrm{rel}} = \sigma_{++} \cup \sigma_{-+} = - \overline{\sigma_{\mathrm{rel}}}$, then $V$ and $C$ give us two constraints, 
\begin{align*}
	S \, P_{\mathrm{rel}} \, S^{-1} &= \id_{\Hil} - P_{\mathrm{rel}} 
	, 
	\\
	C_{\dagger} \, P_{\mathrm{rel}} \, C_{\dagger}^{-1} &= \id_{\Hil} - P_{\mathrm{rel}} 
	, 
\end{align*}
and the even time-reversal symmetry $T_{\dagger} = C_{\dagger} S$, 
\begin{align*}
	T_{\dagger} \, P_{\mathrm{rel}} \, T_{\dagger}^{-1} = P_{\mathrm{rel}}
	. 
\end{align*}
That means the relevant classification is that of class~BDI, because the spectrally flattened operator $Q$ possesses a chiral symmetry, an even particle-hole symmetry and an even time-reversal symmetry. This is consistent with the line gap classification of class~D, $\mathcal{S}_+$ of \cite{Kawabata_Shiozaki_Ueda_Sato:classification_non_hermitian_systems:2019} (see Table~\ref{my_way_simplified:table:example_2_comparison_correct_classification}). 
% paragraph reflection_symmetrically_chosen_relevant_states (end)
% subsubsection example_2 (end)
% subsection extension_to_certain_non_hermitian_cases (end)

\subsection{General classification of normal operators} % (fold)
\label{my_way_simplified:general_principles}
I now extract the general principles from the examples discussed in the previous section, starting with a list akin to equation~\eqref{my_way_simplified:eqn:example_1_symmetries_real_imaginary_part} that tells me how symmetries transform real and imaginary part operators.

\subsubsection{How symmetries and $\dagger$-symmetries manifest themselves} % (fold)
\label{my_way_simplified:general_principles:manifestation_symmetries}
When the hamiltonian is normal, equation~\eqref{diagonalizable_operators:eqn:naive_decomposition_real_imaginary_parts} splits $H = H_{\Re} + \ii H_{\Im}$ into the sum of two hermitian, \emph{commuting} operators. The presence of ($\dagger$)-symmetries of $H$ leads to symmetries of $H_{\Re}$ and $H_{\Im}$. Additional sign flips may occur in $H_{\Im}$ when the symmetry transformation is antilinear or connects $H$ with $H^{\dagger} = H_{\Re} - \ii H_{\Im}$. For linear symmetries, this gives me 
\begin{subequations}\label{my_way_simplified:eqn:linear_non_dagger}
	\begin{align}
		V \, H \, V^{-1} = + H
		\; \; &\Longleftrightarrow \; \; 
		\begin{cases}
			V \, H_{\Re} \, V^{-1} &= + H_{\Re} \\
			V \, H_{\Im} \, V^{-1} &= + H_{\Im} \\
		\end{cases}
		, 
		\label{my_way_simplified:eqn:linear_non_dagger:ordinary}
		\\
		S \, H \, S^{-1} = - H
		\; \; &\Longleftrightarrow \; \; 
		\begin{cases}
			S \, H_{\Re} \, S^{-1} &= - H_{\Re} \\
			S \, H_{\Im} \, S^{-1} &= - H_{\Im} \\
		\end{cases}
		\label{my_way_simplified:eqn:linear_non_dagger:chiral}
		, 
	\end{align}
\end{subequations}
while for antilinear symmetries, I get 
\begin{subequations}\label{my_way_simplified:eqn:antilinear_non_dagger}
	\begin{align}
		T \, H \, T^{-1} = + H
		\; \; &\Longleftrightarrow \; \; 
		\begin{cases}
			T \, H_{\Re} \, T^{-1} &= + H_{\Re} \\
			T \, H_{\Im} \, T^{-1} &= - H_{\Im} \\
		\end{cases}
		, 
		\label{my_way_simplified:eqn:antilinear_non_dagger:TR}
		\\
		C \, H \, C^{-1} = - H
		\; \; &\Longleftrightarrow \; \; 
		\begin{cases}
			C \, H_{\Re} \, C^{-1} &= - H_{\Re} \\
			C \, H_{\Im} \, C^{-1} &= + H_{\Im} \\
		\end{cases}
		\label{my_way_simplified:eqn:antilinear_non_dagger:PH}
		. 
	\end{align}
\end{subequations}
Analogously, I obtain the sign combinations for linear $\dagger$-symmetries, 
\begin{subequations}\label{my_way_simplified:eqn:linear_dagger}
	\begin{align}
		V_{\dagger} \, H \, V_{\dagger}^{-1} = + H^{\dagger}
		\; \; &\Longleftrightarrow \; \; 
		\begin{cases}
			V_{\dagger} \, H_{\Re} \, V_{\dagger}^{-1} &= + H_{\Re} \\
			V_{\dagger} \, H_{\Im} \, V_{\dagger}^{-1} &= - H_{\Im} \\
		\end{cases}
		, 
		\label{my_way_simplified:eqn:linear_dagger:ordinary}
		\\
		S_{\dagger} \, H \, S_{\dagger}^{-1} = - H^{\dagger}
		\; \; &\Longleftrightarrow \; \; 
		\begin{cases}
			S_{\dagger} \, H_{\Re} \, S_{\dagger}^{-1} &= - H_{\Re} \\
			S_{\dagger} \, H_{\Im} \, S_{\dagger}^{-1} &= + H_{\Im} \\
		\end{cases}
		\label{my_way_simplified:eqn:linear_dagger:chiral}
		, 
	\end{align}
\end{subequations}
and their antilinear siblings, 
\begin{subequations}\label{my_way_simplified:eqn:antilinear_dagger}
	\begin{align}
		T_{\dagger} \, H \, T_{\dagger}^{-1} = + H^{\dagger}
		\; \; &\Longleftrightarrow \; \; 
		\begin{cases}
			T_{\dagger} \, H_{\Re} \, T_{\dagger}^{-1} &= + H_{\Re} \\
			T_{\dagger} \, H_{\Im} \, T_{\dagger}^{-1} &= + H_{\Im} \\
		\end{cases}
		, 
		\label{my_way_simplified:eqn:antilinear_dagger:TR}
		\\
		C_{\dagger} \, H \, C_{\dagger}^{-1} = - H^{\dagger}
		\; \; &\Longleftrightarrow \; \; 
		\begin{cases}
			C_{\dagger} \, H_{\Re} \, C_{\dagger}^{-1} &= - H_{\Re} \\
			C_{\dagger} \, H_{\Im} \, C_{\dagger}^{-1} &= - H_{\Im} \\
		\end{cases}
		\label{my_way_simplified:eqn:antilinear_dagger:PH}
		. 
	\end{align}
\end{subequations}
As usual, the four antilinear maps~\eqref{my_way_simplified:eqn:antilinear_non_dagger} and \eqref{my_way_simplified:eqn:antilinear_dagger} come in two flavors, the even and the odd variety depending on whether they square to $+ \id_{\Hil}$ or $- \id_{\Hil}$. 

To simplify notation, I introduce the sign $\epsilon_{\Re,\Im} = \pm 1$ determined from
\begin{align*}
	U \, H_{\Re,\Im} \, U^{-1} &= \epsilon_{\Re,\Im} H_{\Re,\Im} 
	. 
\end{align*}
These sign combinations come in pairs, \eg a TR${}^{\dagger}$ symmetry has the same sign combination as an ordinary, commuting symmetry. Of course, these symmetries are not having the same effect, the first one is antilinear, the second one linear. Moreover, these sign flips explain why the presence of symmetries of $H$ manifest themselves as symmetries in the spectrum (\cf Table~\ref{symmetries:table:overview_symmetries}). 

Similarly, I can compute the symmetry conditions on the unitary phase $V_H = H \, \sabs{H}^{-1}$ and the absolute value $\sabs{H} = \sqrt{H \, H^{\dagger}}$. All symmetries necessarily need to commute with $\sabs{H}$: for non-$\dagger$ symmetries the condition $U \, H \, U^{-1} = \pm H$ manifests itself as 
\begin{align*}
	U \, \bigl ( H \, H^{\dagger} \bigr ) \, U^{-1} &= U \, H \, U^{-1} \, U \, H^{\dagger} \, U^{-1} 
	\\
	&= (\pm 1)^2 \, H \, H^{\dagger} 
	= H \, H^{\dagger} 
	, 
\end{align*}
and therefore also preserves the absolute value, $U \, \sabs{H} \, U^{-1} = + \sabs{H}$. Diagonalizability is not needed in the above computation. By rewriting the phase operator as the product 
\begin{align*}
	V_H = H \, \sabs{H}^{-1} 
	, 
\end{align*}
I see that (anti)commutativity of $U$ with $H$ leads to (anti)commutativity with the phase operator $V_H$. 

For $\dagger$-symmetries, the discussion is more subtle and diagonalizability enters in the derivation: I can repeat the above computation and arrive at 
\begin{align*}
	U_{\dagger} \, \bigl ( H \, H^{\dagger} \bigr ) \, U_{\dagger}^{-1} &= H^{\dagger} \, H
	, 
\end{align*}
where $H$ and its adjoint have traded places on the right. However, by assumption $H$ is normal, \ie it commutes with $H^{\dagger}$, and I can reshuffle them as I see fit to get once more 
\begin{align*}
	U_{\dagger} \, \sabs{H} \, U_{\dagger}^{-1} &= + \sabs{H}
	. 
\end{align*}
In conclusion, for normal operators $\dagger$-symmetries of $H$ again commute with the modulus operator $\sabs{H}$, and $H$ (anti)commutes with $U_{\dagger}$ if and only if its phase $V_H$ does. 
% subsubsection how_symmetries_and_dagger_symmetries_manifest_themselves (end)

\subsubsection{Symmetries lead to relations amongst the spectral projections} % (fold)
\label{my_way_simplified:general_principles:symmetries_relations_spectral_projections}
Relations~\eqref{my_way_simplified:eqn:linear_non_dagger}–\eqref{my_way_simplified:eqn:antilinear_dagger} allow me to enumerate symmetries and constraints for spectral projections, including $P_{\mathrm{rel}}$. In all of my examples, I am able to replace $\sigma_{\mathrm{rel}}$ with a finite or semi-infinite square in the complex plane, 
\begin{align*}
	\Omega = \Omega_{\Re} \times \Omega_{\Im}
	= [\lambda_0 , \lambda_1] \times \ii [\mu_0 , \mu_1] 
	\subseteq \C
	. 
\end{align*}
That simplifies the discussion, spectral projections~\eqref{diagonalizable_operators:eqn:P_rel_product_set} for products transform as 
\begin{align}
	U \, 1_{\Omega}(H) \, U^{-1} &= 1_{\Omega_{\Re}}(\epsilon_{\Re} H_{\Re}) \; 1_{\Omega_{\Im}}(\epsilon_{\Im} H_{\Im}) 
	\notag 
	\\
	&= 1_{\epsilon_{\Re} \Omega_{\Re}}(H_{\Re}) \; 1_{\epsilon_{\Im} \Omega_{\Im}}(H_{\Im})
	. 
	\label{my_way_simplified:eqn:transformation_PVM}
\end{align}
Depending on the sign combination $(\epsilon_{\Re},\epsilon_{\Im})$, the symmetry $U$ maps the spectral projection for the set $\Omega$ onto a spectral projection for the set $\pm \Omega$ or $\pm \overline{\Omega}$. I emphasize that the arguments extend directly to arbitrary Borel sets in the complex plane (\cf Appendix~\ref{appendix:functional_calculus_normal_operators}). 

Equation~\eqref{diagonalizable_operators:eqn:P_rel_product_set} (and also Appendix~\ref{appendix:diagonalizable_operators:facts}) tells us that the spectral projections of $H$ and its adjoint are related through equation~\eqref{diagonalizable_operators:eqn:P_rel_dagger_as_spectral_projection_of_H}. The latter further implies that if $\Omega \cap \sigma(H)$ is symmetric with respect to reflections about the real axis, \ie when $\Omega \cap \sigma(H) = \overline{\Omega \cap \sigma(H)}$, the corresponding spectral projections coincide, $1_{\Omega}(H) = 1_{\Omega}(H^{\dagger})$. 
% subsubsection symmetries_lead_to_relations_amongst_the_spectral_projections (end)

\subsubsection{Deriving symmetries and constraints for the relevant states} % (fold)
\label{my_way_simplified:general_principles:deriving_symmetries_and_constraints}
Once I designate a subset $\sigma_{\mathrm{rel}}$ of the energy or frequency spectrum as being physically relevant, I obtain the projection $P_{\mathrm{rel}} = 1_{\sigma_{\mathrm{rel}}}(H)$ onto the corresponding states. The presence of ($\dagger$-)symmetries in $H$ will lead to the presence of symmetries and constraints of $P_{\mathrm{rel}}$. I call $U$ a symmetry of $P_{\mathrm{rel}}$ if and only if it is a linear or antilinear bounded map with bounded inverse that commutes with $P_{\mathrm{rel}}$, 
\begin{align}
	U \, P_{\mathrm{rel}} \, U^{-1} &= P_{\mathrm{rel}}
	. 
	\label{my_way_simplified:eqn:symmetry}
\end{align}
Similarly, I call $U$ a $\dagger$-symmetry if and only if it connects $P_{\mathrm{rel}}$ with $P_{\mathrm{rel},\dagger}$ from equation~\eqref{diagonalizable_operators:eqn:P_rel_dagger}, 
\begin{align}
	U_{\dagger} \, P_{\mathrm{rel}} \, U_{\dagger}^{-1} &= P_{\mathrm{rel},\dagger}
	. 
	\label{my_way_simplified:eqn:dagger_symmetry}
\end{align}
Constraints connect $P_{\mathrm{rel}}$ with $\id_{\Hil} - P_{\mathrm{rel},(\dagger)}$ and similarly come in two flavors, either as an ordinary constraint, 
\begin{align}
	U \, P_{\mathrm{rel}} \, U^{-1} &= \id_{\Hil} - P_{\mathrm{rel}}
	. 
	\label{my_way_simplified:eqn:constraint}
\end{align}
or a $\dagger$-constraint, 
\begin{align}
	U_{\dagger} \, P_{\mathrm{rel}} \, U_{\dagger}^{-1} &= \id_{\Hil} - P_{\mathrm{rel},\dagger}
	. 
	\label{my_way_simplified:eqn:dagger_constraint}
\end{align}
Which — if any — of these symmetries and constraints are present depends on the symmetries of $H$ and the set of relevant states $\sigma_{\mathrm{rel}}$. This analysis generalizes the arguments from Section~\ref{my_way_simplified:non_hermitian_examples}. Not all symmetries and constraints need to be immediately obvious. For example, the $C_{\dagger} = T U_{\dagger}$ symmetry from Section~\ref{my_way_simplified:extracting_symmetries_constraints:asymmetric} was preserved even though $T$ and $U_{\dagger}$ separately were broken. 

For a given relevant projection $P_{\mathrm{rel}}$ and its sibling $P_{\mathrm{rel},\dagger}$, I can construct two spectrally flattened operators, namely 
\begin{align*}
	Q &= \id_{\Hil} - 2 P_{\mathrm{rel}} 
	, 
	\\
	Q_{\dagger} &= \id_{\Hil} - 2 P_{\mathrm{rel},\dagger} 
	. 
\end{align*}
These give an alternative characterization of my symmetries and constraints: antilinear symmetries of the type~\eqref{my_way_simplified:eqn:symmetry} give rise to time-reversal symmetries of $Q$, \ie I have $T \, Q \, T^{-1} = + Q$. Similarly, a constraint becomes a chiral symmetry or a particle-hole symmetry of $Q$, depending on whether it is linear or antilinear. 

The list of ($\dagger$-)symmetries and ($\dagger$-)constraints then leads to a usually incomplete list of topological invariants that are supported. A linear constraint~\eqref{my_way_simplified:eqn:constraint} indicates the presence of a class~AIII winding number. An odd antilinear symmetry~\eqref{my_way_simplified:eqn:symmetry} suggests that Kane-Melé-type invariants enter the classification. Indeed, this was the case for the examples studied in Section~\ref{my_way}. I used a conjunctive here, because at present we do not fully understand the cases when symmetries and constraints are simultaneously present in their normal and $\dagger$ variety. For example, there exists no \emph{complete} list of topological invariants. 

In conclusion, the classification of diagonalizable operators can be recast as a problem of classifying an orthogonal projection or a pair of orthogonal projections with symmetries and constraints connecting them. As we have seen in a few of the examples, this means I can not only compute the topological classification, but I also already know how to compute at least some of the invariants. 
% subsubsection deriving_symmetries_and_constraints_for_the_relevant_states (end)
% subsection general_principles (end)

\subsection{Topological invariants} % (fold)
\label{my_way_simplified:topological_invariants}
At this point I decide against pursuing a complete zoology of normal operators (which as I shall argue in Section~\ref{my_way} extends to generic diagonalizable operators), and comparing that with the 38 classes and gap-type subclasses of \cite{Kawabata_Shiozaki_Ueda_Sato:classification_non_hermitian_systems:2019,Zhou_Lee:non_hermitian_topological_classification:2019}. Nevertheless, there are two generic situations, which I think merit a few more comments.

\subsubsection{Cases where $P_{\mathrm{rel}}$ has no $\dagger$-symmetries and $\dagger$-constraints} % (fold)
\label{general:case_of_no_dagger_symmetries}
Here, I need to study the classification of the hermitian operators $P_{\mathrm{rel}}$ or $Q$ with a given set of symmetries and constraints. I reckon that standard techniques, including \cite{Kennedy_Zirnbauer:Bott_periodicity_Z2_symmetric_ground_states:2016,DeNittis_Gomi:AI_bundles:2014,DeNittis_Gomi:AII_bundles:2014,DeNittis_Gomi:AIII_bundles:2015} could be used to exhaustively describe the topology in these cases. So even though the initial problem is non-hermitian, the projections I arrive at in the end \emph{are} hermitian and standard theory for hermitian operators applies; the non-hermitian nature of the problem seems to play no role for the topological classification. 
% subsubsection case_of_no_dagger_symmetries (end)

\subsubsection{Cases where $P_{\mathrm{rel}}$ has some $\dagger$-symmetries and/or $\dagger$-constraints} % (fold)
\label{general:dagger_symmetries}
Suppose my setting is such that $P_{\mathrm{rel}}$ comes furnished with a $\dagger$-symmetry~\eqref{my_way_simplified:eqn:symmetry}. Then the $\dagger$-symmmetry $U_{\dagger}$ relates the two projections $P_{\mathrm{rel}}$ and $P_{\mathrm{rel},\dagger}$ with one another. And for this case, it is well-known that provided certain technical conditions are satisfied (\cf \cite[Proposition~2.4]{Avron_Seiler_Simon:charge_deficiency_charge_transport_index_formulas_projections:1994}) I can define a topological invariant 
\begin{align}
	\mathrm{Index} \bigl ( P_{\mathrm{rel}},P_{\mathrm{rel},\dagger} \bigr ) = - \mathrm{Ind} \bigl ( P_{\mathrm{rel}} \, U_{\dagger} \, P_{\mathrm{rel}} \bigr ) 
	\label{general:eqn:index_two_projections}
\end{align}
that classifies one projection relative to another. Here, the $\mathrm{Index}$ map on the left is the index of two projections that is formally defined as 
\begin{align*}
	\mathrm{Index}(P,Q) &= \dim \, \ker \, \bigl ( Q - (\id_{\Hil} - P) \bigr ) 
	\, + \\
	&\quad 
	- \dim \, \ker \, \bigl ( P - (\id_{\Hil} - Q) \bigr ) 
	, 
\end{align*}
and $\mathrm{Ind}$ on the right is the \emph{Fredholm index}
\begin{align*}
	\mathrm{Ind}(F) = \dim \ker F - \dim \ker F^{\dagger} 
	. 
\end{align*}
Both are invariant under continuous deformations by definition, and are known to be topological invariants in certain situations; that includes certain models for systems exhibiting the Quantum Hall Effect (\cf \cite[Section~6]{Avron_Seiler_Simon:charge_deficiency_charge_transport_index_formulas_projections:1994}). 

It stands to reason that they are topological invariants also for certain classes of diagonalizable topological insulators. To the best of my knowledge, these have not yet been considered in this context, though. 

Naturally, the presence of other ($\dagger$-)symmetries and ($\dagger$-)\linebreak constraints needs to be taken into account. 
% subsection case_with_dagger_symmetries_novel_topological_invariants (end)
% subsection topological_invariants (end)

\subsection{Relation to point and line gap classifications of \cite{Kawabata_Shiozaki_Ueda_Sato:classification_non_hermitian_systems:2019}} % (fold)
\label{my_way_simplified:relation_to_Kawabata_et_al}
One initial datum in the topological classification of Kawabata et al. \cite{Kawabata_Shiozaki_Ueda_Sato:classification_non_hermitian_systems:2019} is the gap type, \ie whether I am dealing with a point gap, a generic line gap, a real line gap or an imaginary line gap. 

The point gap case corresponds to choosing a relevant part of the spectrum that lacks any symmetry, 
\begin{align*}
	\sigma_{\mathrm{rel}} \neq - \sigma_{\mathrm{rel}} , \; \pm \overline{\sigma_{\mathrm{rel}}} 
	. 
\end{align*}
The real line gap means I deem all states to the left of the imaginary axis as relevant, \ie I pick 
\begin{align*}
	\sigma_{\mathrm{rel}} = \bigl \{ E \in \sigma(H) \; \; \vert \; \; \Re E \leq 0 \bigr \} 
	. 
\end{align*}
Likewise, an imaginary line gap corresponds to choosing 
\begin{align*}
	\sigma_{\mathrm{rel}} = \bigl \{ E \in \sigma(H) \; \; \vert \; \; \Im E \leq 0 \bigr \} 
	. 
\end{align*}
Of course, I could have equivalently picked states to the right of the imaginary axis or above the real line, respectively; this just amounts to replacing $P_{\mathrm{rel}}$ with $\id_{\Hil} - P_{\mathrm{rel}}$. 

Generic line gaps similarly split the complex plane in half, where the relevant states are those lying below or above the line; by convention on where the spectral gap resides this line must run through the origin. In contrast to the two previous line gaps the vector $n_{\Re} + \ii n_{\Im}$ that spans the line need not be parallel to the imaginary or real axis. The relevant states are those that lie below or above the line, 
\begin{align*}
	\sigma_{\mathrm{rel}} = \Bigl \{ E_{\Re} + \ii E_{\Im} \in \sigma(H) \; \; \big \vert \; \; \pm \bigl ( n_{\Re} \, E_{\Re} + n_{\Im} \, E_{\Im} \bigr ) \leq 0 \Bigr \} 
	. 
\end{align*}
I have summarized all of these cases in Table~\ref{my_way_simplified:table:gap_types_relevant_spectrum}. 

The second example covered in Section~\ref{my_way_simplified:non_hermitian_examples:2} has two cases which fall into the (generic) line gap classification, namely when $\sigma_{\mathrm{rel}}$ consists of the states above the real axis \emph{or} to the right of the imaginary axis; both will lead to the same classification. That is because the operators possesses a symmetry which makes the spectrum point-symmetric and therefore connects $P_{\mathrm{rel}}$ to its complement $\id_{\Hil} - P_{\mathrm{rel}}$. 

When $H$ comes with symmetries that lead to symmetries in the spectrum, then these sets inherit these symmetries. I point once more to the examples discussed in Section~\ref{my_way_simplified:non_hermitian_examples}. 

One last word regarding spectral flattening. For simple spectral gaps (as in Figure~\ref{intro:figure:highly_symmetric_spectrum}, but not the nested Cs as in Figure~\ref{38_fold_recap:figure:nested_Cs_spectrum}) my choice of spectrally flattened hamiltonian~\eqref{my_way_simplified:eqn:spectrally_flattened_hamiltonian} coincides with the one obtained from the procedure outlined in \cite[Figure~2]{Kawabata_Shiozaki_Ueda_Sato:classification_non_hermitian_systems:2019} only for the point gap and real line gap case. When the line gap is imaginary, they differ by a factor $\pm \ii$, which is immaterial for their topological classification. 

\begin{table}
	\begin{centering}
		\newcolumntype{A}{>{\centering\arraybackslash\normalsize}m{30mm}}
		\newcolumntype{B}{>{\centering\arraybackslash\normalsize}m{40mm}}
		% \newcolumntype{C}{>{\centering\arraybackslash\normalsize}m{10mm}}
		\renewcommand{\arraystretch}{1.5}
		\begin{tabular}{A | B}
			\emph{Gap type} & \emph{Relevant states} \\ \hline \hline 
			Generic line gap & above/below line \\ \hline
			Real line gap & $\pm \Re E \leq 0$ \\ \hline 
			Imaginary line gap & $\pm \Im E \leq 0$ \\ \hline 
			Point gap & asymmetrically chosen \\ 
		\end{tabular}
	\end{centering}
	\caption{Correspondence between gap type and relevant spectrum $\sigma_{\mathrm{rel}}$. My definition is consistent with the Fermi projection for hermitian systems, although mathematically, one could have equivalently chosen $\Re E \geq 0$ and $\Im E \geq 0$, respectively. }
	\label{my_way_simplified:table:gap_types_relevant_spectrum}
\end{table}
%
% subsection relation_to_point_and_line_gap_classifications_of_cite_kawabata_shiozaki_ueda_sato_classification_non_hermitian_systems_2019 (end)
% section alternative_approach_symmetries_and_constraints_of_relevant_states (end)
%!TEX root = /Users/max/Dropbox/research/topological classification non-selfadjoint hamiltonians/paper/non_hermitian_classification.tex
\section[Extending the classification of $P_{\mathrm{rel}}$ to the general case]{Extending the classification \linebreak of $P_{\mathrm{rel}}$ to the general case} % (fold)
\label{my_way}
The last piece of the puzzle is to check whether replacing the simplifying assumption of normality (Assumption~\ref{my_way_simplified:assumption:simplifying_assumption}) with diagonalizability (Assumption~\ref{intro:assumption:diagonalizability}) changes anything as far as the topological classification is concerned. Fortunately, the answer is no and the purpose of this section is to explain to the reader why. 

From this moment on let me operate under the original Assumption~\ref{intro:assumption:diagonalizability} from the introduction. When the ($\dagger$-)symmetries are not (anti)unitary with respect to the biorthogonal scalar product, the Hilbert space the symmetries would like to live in is different from the Hilbert space the operator feels most comfortable in. Let us explore some of the ramifications together.

\subsection[Mismatch of geometry and focussing on algebraic properties]{Mismatch of geometry and \linebreak focussing on algebraic properties} % (fold)
\label{my_way:geometry_vs_algebra}
This mismatch of geometries, which emerges from the two choices of scalar products, is well-known in the context of non-hermitian topological insulators. For example, Schomerus carefully works out the consequences this mismatch can have in \cite{Schomerus:non_reciprocal_response_non_hermitian_metamaterials:2020}. While the discussion is framed under the rubrik of “non-orthogonality”, I think it is more apt to speak of a \emph{mismatch} of geometries — I can always pick a scalar product that is compatible with \emph{either} the symmetries \emph{or} the hamiltonian, just not both simultaneously. 

That tension cannot be resolved, unless I simply forgo geometry altogether. Instead of working with operators on \emph{Hilbert} spaces, I discard the scalar product and just think of Banach spaces, that is normed, complete vector spaces. By definition of diagonalizability the similarity transform $G$ that relates the biorthogonal scalar product~\eqref{diagonalizable_operators:eqn:weighted_scalar_product} to the original scalar product is bounded and has a bounded inverse. Therefore, the norms $\scpro{\varphi}{\varphi}^{\nicefrac{1}{2}}$ and $\scppro{\varphi}{\varphi}^{\nicefrac{1}{2}}$
are equivalent (\cf my discussion in Section~\ref{diagonalizable_operators:topological_classification_algebraic}), and $\bigl ( \Hil , \scpro{\, \cdot \,}{\, \cdot \,} \bigr )$ and $\bigl ( \Hil , \scppro{\, \cdot \,}{\, \cdot \,} \bigr )$ agree as \emph{Banach} spaces. 

Rather than think of symmetries as (anti)unitaries, I can regard them as bounded, invertible (anti)linear maps that square to $\pm \id_{\Hil}$. Indeed, this exact same reasoning is often used in reverse: Kuiper's Theorem \cite{Kuiper:homotopy_type_unitary_group:1965}, for example, states that in class~A working with unitaries is the same is as working with bounded \emph{invertible} operators whose inverses are bounded, \ie that $\mathcal{B}(\Hil)^{-1} = \mathrm{GL}(\Hil)$ can be deformation retracted to $\mathcal{U}(\Hil)$. To obtain the point gap classification Kawabata et al.\ homotopically deform $H \in \mathcal{B}(\Hil)^{-1}$ to a unitary operator $\hat{H} \in \mathcal{U}(\Hil)$. Indeed, all ingredients for the topological classification of diagonalizable operators can be rephrased just using algebraic constructs: unitary becomes invertible with bounded inverse; the spectrum as a set can be traced to the \emph{invertibility} of $H - E$; commutativity, which characterizes diagonalizable operators via 
\begin{align*}
	\bigl [ H_{\Re} , H_{\Im} \bigr ] = 0 
	, 
\end{align*}
is entirely algebraic; the defining relations of projections, $P^2 = P$, and spectrally flattened hamiltonians, $Q^2 = \id_{\Hil}$, are algebraic. Put another way, it stands to reason that the topological classification of generic diagonalizable operators only depends on the \emph{algebraic} structure and not the geometric structure. 

The attentive reader will have noticed a gap in my line of argumentation: at least in case $\dagger$-symmetries are present, \ie symmetries which relate $H$ to $H^{\dagger}$, I need to use the adjoint, which is tied to the initially given scalar product and therefore decidedly not algebraic. The solution is to use the cartesian decomposition of $H = H_{\Re} + \ii H_{\Im}$ and replace the adjoint with the algebraic relation $H^{\dagger} = W^{-1} \, (H_{\Re} - \ii H_{\Im}) \, W$. 
% subsection geometry_vs_algebra (end)

\subsection{Continuous, symmetry- and gap-preserving deformations of $H$ lead to continuous deformations of $P_{\mathrm{rel}}$ and $P_{\mathrm{rel},\dagger}$} % (fold)
\label{my_way:homotopies}
One of the reasons I have had to exclude non-diagonalizable operators from the topological classification is that continuous symmetry- and gap-preserving deformations $\lambda \mapsto H(\lambda)$ (with respect to the norm topology) do \emph{not} lead to continuous deformations of $P_{\mathrm{rel}}(\lambda)$ (\cf my discussion in Section~\ref{why_diagonalizability_matters:spectral_projections_spectrally_flattened_hamiltonians}). Here I will take a moment to show that these deficiencies are cured once I impose diagonalizability, \ie that $P_{\mathrm{rel}}(\lambda)$ and $P_{\mathrm{rel},\dagger}(\lambda)$ inherit the continuity from $H(\lambda)$. In what follows, I will tacitly assume that all deformations $H(\lambda)$ preserve diagonalizability, the spectral gap and all relevant symmetries. 

A sensible approach would be to look at functional calculus more broadly, since it is one of 5 equivalent characterizations of diagonalizability given in Theorem~\ref{diagonalizable_operators:thm:characterizations_diagonalizable_operators}. For example, if I wanted to split $H = H_{\Re} + \ii H_{\Im}$ into real and imaginary parts, I would need to involve the biorthogonal scalar product~\eqref{diagonalizable_operators:eqn:weighted_scalar_product}. This, in turn, involves the \emph{$\lambda$-dependent} weight operator operator $W(\lambda) = G(\lambda)^{\dagger} G(\lambda)$, \ie for different values of $\lambda$, I would have to use a different scalar product. Proving the continuity of \eg the real part 
\begin{align*}
	H(\lambda) \mapsto H_{\Re}(\lambda) = \frac{1}{2} \Bigl ( H(\lambda) + W(\lambda)^{-1} \, H(\lambda)^{\dagger} \, W(\lambda) \Bigr )
\end{align*}
would now depend on the continuity of adjoining with $W(\lambda)$. This is not as simple as it looks since $G(\lambda)$ is not uniquely determined; in fact, concatenating any $\lambda$-dependent unitary $V(\lambda)$ to make $G'(\lambda) = V(\lambda) \, G(\lambda)$ leads to the exact same operator $W(\lambda)$. 

Fortunately, I have offered several equivalent definitions of $P_{\mathrm{rel}}$, and the complex integral~\eqref{intro:eqn:definition_P_rel} allows for a more direct approach. Likewise, $P_{\mathrm{rel},\dagger}$ can be expressed as a contour integral after replacing $H$ with $H^{\dagger}$. 

But let me study $P_{\mathrm{rel}}$ first. The contour integral has two variable components, the resolvent operator $\bigl ( H(\lambda) - z \bigr )^{-1}$ and the contour that encloses the relevant part $\sigma_{\mathrm{rel}}(\lambda)$ of the spectrum. I conjecture that the spectrum of diagonalizable operators is inner \emph{and} outer semicontinuous with respect to \emph{diagonalizable} perturbations; I will investigate this point in a future work \cite{Lein_Lenz:diagonalizable_operators:2020}. Proceeding under the assumption that this is indeed true, this will guarantee that on the one hand spectrum inside $\sigma_{\mathrm{rel}}(\lambda)$ cannot suddenly disappear; nor does the gap between $\sigma_{\mathrm{rel}}(\lambda)$ and the remainder of the spectrum suddenly collapse. To show continuity of $P_{\mathrm{rel}}(\lambda)$ for an arbitary parameter value $\lambda_0$, I pick a contour $\Gamma \bigl ( \sigma_{\mathrm{rel}}(\lambda_0) \bigr )$. Then at least in a small neighborhood of $\lambda_0$, the contour $\Gamma \bigl ( \sigma_{\mathrm{rel}}(\lambda_0) \bigr )$ encloses $\sigma_{\mathrm{rel}}(\lambda)$ and \emph{only} $\sigma_{\mathrm{rel}}(\lambda)$ also for $\lambda \approx \lambda_0$. Consequently, in the vicinity of $\lambda_0$, I can express
\begin{align}
	P_{\mathrm{rel}}(\lambda) = \frac{\ii}{2 \pi} \int_{\Gamma(\sigma_{\mathrm{rel}}(\lambda_0))} \dd z \, \bigl ( H(\lambda) - z \bigr )^{-1} 
	\label{my_way:eqn:P_rel_lambda_contour_integral}
\end{align}
as a contour integral with respect to a \emph{fixed} contour. The resolvent inherits the continuity of $H(\lambda)$ in the parameter, and from that I conclude that also $P_{\mathrm{rel}}(\lambda)$ is continuous in a neighborhood of $\lambda_0$. Since the value $\lambda_0$ was arbitrary, that shows continuity for as long as $\sigma_{\mathrm{rel}}(\lambda)$ is separated by a gap from the remainder of the spectrum. 

The necessary modifications for $P_{\mathrm{rel},\dagger}$ are straightforward: \emph{a priori} I need to assume that $\sigma_{\mathrm{rel}}(\lambda) \cap \sigma \bigl ( H(\lambda)^{\dagger} \bigr )$ is separated from the remainder $\sigma \bigl ( H(\lambda)^{\dagger} \bigr ) \setminus \sigma_{\mathrm{rel}}(\lambda)$ by a gap. So let me proceed under the assumption that this is so. Given that the spectrum of $H$ and $\sigma(H^{\dagger}) = \overline{\sigma(H)}$ are related by complex conjugation, the latter condition translates to $\overline{\sigma_{\mathrm{rel}}(\lambda)} \cap \sigma \bigl ( H(\lambda) \bigr )$ being gapped from the rest of the spectrum of $H$. The adjoint operation $H \mapsto H^{\dagger}$ is norm continuous (\cf \cite[Theorem~VI.3~(e)]{Reed_Simon:M_cap_Phi_1:1972}), and $H^{\dagger}$ is diagonalizable exactly when $H$ is (Lemma~\ref{appendix:diagonalizable_operators:lem:useful_facts}~(3)). Consequently, the adjoint of any continuous deformation $H(\lambda)$ is another continuous deformation $H(\lambda)^{\dagger}$ of a diagonalizable operator. So my arguments for $P_{\mathrm{rel}}(\lambda)$ apply to $P_{\mathrm{rel},\dagger}(\lambda)$ as well after replacing $H(\lambda)$ with its adjoint, provided I have a spectral gap. 

All of these arguments are compatible with the presence of ($\dagger$-)symmetries and ($\dagger$-)constraints. For example, a time-reversal symmetry $T$ of $H(\lambda)$ transforms the resolvent operator to 
\begin{align*}
	T \, \bigl ( H(\lambda) - z \bigr )^{-1} \, T^{-1} &= \bigl ( T \, H(\lambda) \, T^{-1} - \bar{z} \bigr )^{-1} 
	\\
	&= \bigl ( H(\lambda) - \bar{z} \bigr )^{-1} 
	\\
	&= \bigl ( H(\lambda) - \bar{z} \bigr )^{-1} 
	. 
\end{align*}
Note that equation~\eqref{my_way:eqn:P_rel_lambda_contour_integral} has a purely imaginary prefactor whose sign gets flipped when commuting it with $T$. Symmetries of the hamiltonian only become symmetries or constraints of the relevant projection if the relevant spectrum has the appropriate symmetry. Then one can choose a contour compatible with the symmetries. 
% subsection homotopies_of_diagonalizable_operators (end)

\subsection[Symmetries of the hamiltonian and the relevant projection]{Symmetries of the hamiltonian \linebreak and the relevant projection} % (fold)
\label{my_way:symmetries_hamiltonian}
The presence of symmetries now leads to relations between $H = H_{\Re} + \ii H_{\Im}$ and its biorthogonal adjoint $H^{\ddagger} = H_{\Re} - \ii H_{\Im}$. Evidently, the symmetry relations that only involve $H$, \eqref{my_way_simplified:eqn:linear_non_dagger} and \eqref{my_way_simplified:eqn:antilinear_non_dagger}, are untouched. 

The two linear $\dagger$-symmetries can be rephrased as 
\begin{widetext}
	\begin{subequations}\label{my_way:eqn:linear_dagger}
		\begin{align}
			V_{\dagger} \, H \, V_{\dagger}^{-1} = + W \, H^{\ddagger} \, W^{-1}
			\; \; &\Longleftrightarrow \; \; 
			\begin{cases}
				V_{\dagger} \, H_{\Re} \, V_{\dagger}^{-1} &= + W \, H_{\Re} \, W^{-1} \\
				V_{\dagger} \, H_{\Im} \, V_{\dagger}^{-1} &= - W \, H_{\Im} \, W^{-1} \\
			\end{cases}
			, 
			\label{my_way:eqn:linear_dagger:ordinary}
			\\
			S_{\dagger} \, H \, S_{\dagger}^{-1} = - W \, H^{\ddagger} \, W^{-1} 
			\; \; &\Longleftrightarrow \; \; 
			\begin{cases}
				S_{\dagger} \, H_{\Re} \, S_{\dagger}^{-1} &= - W \, H_{\Re} \, W^{-1} \\
				S_{\dagger} \, H_{\Im} \, S_{\dagger}^{-1} &= + W \, H_{\Im} \, W^{-1} \\
			\end{cases}
			\label{my_way:eqn:linear_dagger:chiral}
			. 
		\end{align}
	\end{subequations}
	%
% \end{widetext}
%
Similarly, their antilinear siblings also involve conjugating with $W$, 
%
% \begin{widetext}
	%
	\begin{subequations}\label{my_way:eqn:antilinear_dagger}
		\begin{align}
			T_{\dagger} \, H \, T_{\dagger}^{-1} = + W \, H^{\ddagger} \, W^{-1}
			\; \; &\Longleftrightarrow \; \; 
			\begin{cases}
				T_{\dagger} \, H_{\Re} \, T_{\dagger}^{-1} &= + W \, H_{\Re} \, W^{-1} \\
				T_{\dagger} \, H_{\Im} \, T_{\dagger}^{-1} &= + W \, H_{\Im} \, W^{-1} \\
			\end{cases}
			, 
			\label{my_way:eqn:antilinear_dagger:TR}
			\\
			C_{\dagger} \, H \, C_{\dagger}^{-1} = - W \, H^{\ddagger} \, W^{-1} 
			\; \; &\Longleftrightarrow \; \; 
			\begin{cases}
				C_{\dagger} \, H_{\Re} \, C_{\dagger}^{-1} &= - W \, H_{\Re} \, W^{-1} \\
				C_{\dagger} \, H_{\Im} \, C_{\dagger}^{-1} &= - W \, H_{\Im} \, W^{-1} \\
			\end{cases}
			\label{my_way:eqn:antilinear_dagger:PH}
			. 
		\end{align}
	\end{subequations}
\end{widetext}
What is more, functional calculus is also compatible with these symmetries, which is important when I want to infer symmetries of spectral projections. The first thing to note is that $H$ is diagonalizable exactly when $H^{\dagger}$ is (\cf Lemma~\ref{appendix:diagonalizable_operators:lem:useful_facts}~(3)). Consequently, $H^{\dagger}$ has a functional calculus. Secondly, for any bounded invertible map $V$ with bounded inverse $V^{-1}$ the spectral projections of $V \, H \, V^{-1}$ and $H$ are related by the similarity transform $V$, 
\begin{align*}
	1_{\Lambda} \bigl ( V \, H \, V^{-1} \bigr ) = V \, 1_{\Lambda}(H) \, V^{-1} 
	, 
\end{align*}
where $\Lambda \subseteq \C$ is any Borel set in the complex plane (Lemma~\ref{appendix:diagonalizable_operators:lem:similarity_tranform_functional_calculus}). Compared with equation~\eqref{diagonalizable_operators:eqn:P_rel_dagger_as_spectral_projection_of_H}, the relation between the spectral projections of $H^{\dagger}$ and $H$ is augmented by the similarity transform $W$, 
\begin{align*}
	1_{\Lambda}(H^{\dagger}) &= 1_{\Lambda} \bigl ( W \, H^{\ddagger} \, W^{-1} \bigr ) 
	= W \, 1_{\Lambda}(H^{\ddagger}) \, W^{-1} 
	\\
	&= W \, 1_{\overline{\Lambda}}(H)  \, W^{-1} 
	. 
\end{align*}
%
% subsection symmetries_of_the_hamiltonian (end)

\subsection[Symmetries and constraints of the projection onto the relevant states]{Symmetries and constraints of \linebreak the projection onto the relevant states} % (fold)
\label{my_way:symmetries_constraints_projection}
The readers can hopefully identify the pattern: the equations from Section~\ref{my_way_simplified} that contain $H^{\dagger}$ need to be augmented by adjoining with $W$. Starting with the daggered projection~\eqref{diagonalizable_operators:eqn:P_rel_dagger}, I instead get  
\begin{align}
	P_{\mathrm{rel},\dagger} = W \, 1_{\overline{\sigma_{\mathrm{rel}}}}(H) \, W^{-1} 
	. 
	\label{my_way:eqn:daggered_projection}
\end{align}
The daggered symmetry and constraint conditions that relate $P_{\mathrm{rel}}$ with $P_{\mathrm{rel},\dagger}$ are identical, \ie I still retain equations~\eqref{my_way_simplified:eqn:dagger_symmetry} and \eqref{my_way_simplified:eqn:dagger_constraint}, albeit for the modified $\dagger$-projection $P_{\mathrm{rel},\dagger}$ from equation~\eqref{my_way:eqn:daggered_projection}. 

The symmetries and constraints can equally be expreessed in terms of the spectrally flattened hamiltonian $Q = \id_{\Hil} - 2 P_{\mathrm{rel}}$ and its $\dagger$-counterpart 
\begin{align*}
	Q_{\dagger} = \id_{\Hil} - 2 P_{\mathrm{rel},\dagger} 
	, 
\end{align*}
whose definition is identical to that in Section~\ref{my_way_simplified} except that I insert \eqref{my_way:eqn:daggered_projection} as the $\dagger$-projection. 
% subsection symmetries_and_constraints_of_the_projection_onto_the_relevant_states (end)

\subsection[The topological classification for periodic operators is not affected]{The topological classification \linebreak for periodic operators is not affected} % (fold)
\label{my_way:periodic_operators_vector_bundles}
At least once I impose mild conditions on $W$, its presence does not affect the topological classification and indeed, for the purposes of the topological classification an emergent $\dagger$-symmetry like for example 
\begin{align}
	U_{\dagger} \, P_{\mathrm{rel}} \, U_{\dagger}^{-1} &= P_{\mathrm{rel},\dagger}
	\notag \\
	&= W \, 1_{\overline{\sigma_{\mathrm{rel}}}}(H) \, W^{-1}
	\label{my_way:eqn:emergent_dagger_symmetry_P_rel}
\end{align}
is just as good as the symmetry for $P_{\mathrm{rel}} = 1_{\sigma_{\mathrm{rel}}}(H)$ without $W$, 
\begin{align*}
	U \, 1_{\sigma_{\mathrm{rel}}}(H) \, U^{-1} = 1_{\overline{\sigma_{\mathrm{rel}}}}(H) 
	. 
\end{align*}
It is tempting to combine the two operators to $U' = W^{-1} \, U$, which now satisfies the above equation after replacing the (anti)unitary $U$ with the (anti)linear similarity transform $U'$. But I need to take a little more care. 

Let me spell out the details for the periodic case, where I can reach into the toolbox of vector bundle theory. The mild assumption I have referred to earlier is: 
\begin{assumption}\label{my_way:assumption:regulatity_periodic_case}
	We suppose that $H = \int_{\BZ}^{\oplus} \dd k \, H(k)$ and therefore $W = \int_{\BZ}^{\oplus} \dd k \, W(k)$ are periodic, and $H(k)$ as well as $W(k)$ depend on Bloch momentum $k$ in a \emph{continuous} fashion. 
\end{assumption}
At least for periodic tight-binding operators $H(k)$ this assumptions is almost always satisfied in practice. Exceptions do happen \eg due to conical intersections at $k = 0$ and $E = 0$ that are characteristic for Maxwell-type operators, which describe certain classical waves (\cf \cite{DeNittis_Lein:Schroedinger_formalism_classical_waves:2017} and \cite[Section~3.2]{DeNittis_Lein:adiabatic_periodic_Maxwell_PsiDO:2013}). However, usually periodic hamiltonians physicists encounter are even \emph{analytic}; for the purpose of topological classifications, though, continuity suffices (\cf the discussion in \cite[Section~II.F]{DeNittis_Lein:exponentially_loc_Wannier:2011}).

\subsubsection{The Bloch vector bundle} % (fold)
\label{my_way:periodic_operators_vector_bundles:bloch_vector_bundle}
Thanks to the above Assumption~\ref{my_way:assumption:regulatity_periodic_case} and the spectral gap, I can define the so-called \emph{Bloch vector bundle} 
\begin{align*}
	\mathcal{E}(P_{\mathrm{rel}}) : \bigsqcup_{k \in \BZ} \ran P_{\mathrm{rel}}(k) \overset{\pi}{\longrightarrow} \BZ 
\end{align*}
over the $d$-dimensional Brillouin torus $\BZ$; for a precise mathematical definition, I refer to \eg \cite[Section~IV]{DeNittis_Lein:exponentially_loc_Wannier:2011}. 

In the absence of any symmetries, \ie class~A, $\mathcal{E}(P_{\mathrm{rel}})$ and the analogously defined $\mathcal{E} \bigl ( W \, P_{\mathrm{rel}} \, W^{-1} \bigr )$ are isomorphic vector bundles: the continuous map $k \mapsto W(k)$ can now be interpreted as a vector bundle isomorphism 
\begin{align}
	\bfig
		\node Bloch_bundle(-500,0)[\mathcal{E}(P_{\mathrm{rel}})]
		\node W_Bloch_bundle(500,0)[\mathcal{E} \bigl ( W \, P_{\mathrm{rel}} \, W^{-1} \bigr )]
		\node Brillouin_torus(0,-600)[\BZ]
		\arrow[Bloch_bundle`W_Bloch_bundle;W]
		\arrow[Bloch_bundle`Brillouin_torus;\pi]
		\arrow[W_Bloch_bundle`Brillouin_torus;\pi_W]
	\efig
	\label{my_way:eqn:equivalence_vector_bundles}
\end{align}
as it depends continuosuly on $k$ and isormorphically maps the fiber $\ran P_{\mathrm{rel}}(k)$ onto the fiber
\begin{align*}
	\ran \bigl ( W(k) \, P_{\mathrm{rel}}(k) \, W(k)^{-1} \bigr ) = W(k) \, \bigl [ \ran P_{\mathrm{rel}}(k) \bigr ] 
\end{align*}
over the same base point $k$. The classification of class~A hermitian topological insulators now translates to classifying complex vector bundles up to isomorphism. The resulting equivalence classes, \ie topological phases, are characterized by the rank and Chern classes, albeit not necessarily completely (\cf \cite[Section~V.G]{DeNittis_Lein:exponentially_loc_Wannier:2011} for a counterexample). The converse conclusion nevertheless holds: as isomorphic vector bundles, they lie in the same topological phase and all topological invariants, that is rank and Chern classes, of $\mathcal{E}(P_{\mathrm{rel}})$ and $\mathcal{E} \bigl ( W \, P_{\mathrm{rel}} \, W^{-1} \bigr )$ necessarily agree. 

This line of argumentation tells me that for the purpose of topological classification the relation $P_{\mathrm{rel}} = P_{\mathrm{rel},\dagger}$ in the case where $H$ is normal and $\sigma_{\mathrm{rel}} = \overline{\sigma_{\mathrm{rel}}}$ (Section~\ref{my_way_simplified}) is just as good as $P_{\mathrm{rel}} = W^{-1} \, P_{\mathrm{rel},\dagger} \, W$. 
% subsubsection the_bloch_vector_bundle (end)

\subsubsection{Dealing with symmetries and $\dagger$-symmetries} % (fold)
\label{my_way:periodic_operators_vector_bundles:dagger_symmetries}
Once I add symmetries of $P_{\mathrm{rel}}$ into the mix, making these arguments precise is more involved if the works on classes~AI, AII and AIII \cite{DeNittis_Gomi:AI_bundles:2014,DeNittis_Gomi:AII_bundles:2014,DeNittis_Gomi:AIII_bundles:2015} are any indication. For each of these cases, I have to clarify what I mean by “vector bundle with symmetries” and make precise when two vector bundles with symmetries are equivalent. I will not attempt to venture into the details here and refer the readers to the aforementioned works by De~Nittis and Gomi. 

The presence of $\dagger$-symmetries lead to relations between $\mathcal{E}(P_{\mathrm{rel}})$ and the vector bundle $\mathcal{E}(P_{\mathrm{rel},\dagger}) \cong \mathcal{E} \bigl ( 1_{\overline{\sigma_{\mathrm{rel}}}}(H) \bigr )$: if $U_{\dagger}$ is linear, it means the vector bundles 
\begin{align*}
	\mathcal{E}(P_{\mathrm{rel}}) &\cong \mathcal{E} \bigl ( U_{\dagger} \, P_{\mathrm{rel}} \, U_{\dagger}^{-1} \bigr ) 
	= \mathcal{E}(P_{\mathrm{rel},\dagger}) 
	\\
	&\cong \mathcal{E} \bigl ( 1_{\overline{\sigma_{\mathrm{rel}}}}(H) \bigr )
\end{align*}
are isomorphic in the sense of class~A (denoted with $\cong$). That is because the presence of the $\dagger$-symmetry leads to the fiber-wise relation 
\begin{align}
	U_{\dagger} \, P_{\mathrm{rel}}(k) \, U_{\dagger}^{-1} = W(k) \, P_{\mathrm{rel},\dagger}(k) \, W(k)^{-1} 
	, 
	\label{my_way:eqn:fiber_wise_relation_vector_bundles}
\end{align}
which can equivalently written as $U'(k) \, P_{\mathrm{rel}}(k) \, U'(k)^{-1} = P_{\mathrm{rel},\dagger}(k)$ for $U'(k) = W(k)^{-1} \, U_{\dagger}(k)$. 

When $U_{\dagger}$ is antilinear, I can adapt the arguments from \cite[Section~V.C]{DeNittis_Lein:exponentially_loc_Wannier:2011}: usually antilinear symmetries flip the sign of $k$, so I need to replace $k$ by $-k$ on the right-hand side in the fiber-wise relation~\eqref{my_way:eqn:fiber_wise_relation_vector_bundles}. As a result the new fiber-wise relation leads to 
\begin{align*}
	\mathcal{E}^{\dagger}(P_{\mathrm{rel}}) \cong f^* \bigl ( \mathcal{E}(P_{\mathrm{rel},\dagger}) \bigr )
\end{align*}
being isomorphic in the sense of complex vector bundles. On the left $\mathcal{E}^{\dagger}(P_{\mathrm{rel}})$ is the conjugate vector bundle where all transition functions are replaced by their complex conjugates (\cf \cite[Chapter~14]{Milnor_Stasheff:characteristic_classes:1974}). And on the right, I am considering the pullback bundle with respect to the function $f : k \mapsto -k$ that flips the sign of momentum. Intuitively, to construct the pullback bundle $f^* \bigl ( \mathcal{E}(P_{\mathrm{rel},\dagger}) \bigr )$ I glue the fibers $\ran P_{\mathrm{rel},\dagger}(k)$ together in a mirror universe. 

Adapting the arguments and the computation in the proof of \cite[Theorem~5.4]{DeNittis_Lein:exponentially_loc_Wannier:2011}, I arrive at the following relation between the Chern classes: 
\begin{align}
	c_n \bigl ( \mathcal{E}(P_{\mathrm{rel}}) \bigr ) = (-1)^n \, c_n \bigl ( \mathcal{E}(P_{\mathrm{rel},\dagger}) \bigr ) 
	, 
	&&
	n \in \N
	\label{my_way:eqn:relation_Chern_classes_P_rel_P_rel_dagger}
\end{align}
Put another way, all even Chern classes agree whereas all odd Chern classes are equal in magnitude, but have opposite sign. The sign flip comes from the relation 
\begin{align*}
	c_n \bigl ( \mathcal{E}^{\dagger}(P_{\mathrm{rel},\dagger}) \bigr ) = (-1)^n \, c_n \bigl ( \mathcal{E}(P_{\mathrm{rel},\dagger}) \bigr ) 
\end{align*}
between the Chern classes of a vector bundle and its conjugate bundle (\cf \cite[Lemma~14.9]{Milnor_Stasheff:characteristic_classes:1974}). 

In case the relevant spectrum $\sigma_{\mathrm{rel}} = \overline{\sigma_{\mathrm{rel}}}$ is chosen symmetrically, the relevant projections $P_{\mathrm{rel},\dagger} = W^{-1} \, P_{\mathrm{rel}} \, W$ are related by a similarity transform and lead to isomorphic complex vector bundles. Thus, their Chern numbers all agree, 
\begin{align*}
	c_n \bigl ( \mathcal{E}(P_{\mathrm{rel}}) \bigr ) = c_n \bigl ( \mathcal{E}(P_{\mathrm{rel},\dagger}) \bigr ) 
	. 
\end{align*}
When combined with equation~\eqref{my_way:eqn:relation_Chern_classes_P_rel_P_rel_dagger}, I deduce that all \emph{odd} Chern classes must vanish. In low dimension, $d \leq 3$, this means the presence of a time-reversal${}^{\dagger}$ symmetry forces the Bloch bundle $\mathcal{E}(P_{\mathrm{rel}})$ to be trivial as a complex vector bundle. 

Of course, the above arguments do not preclude the absence of topological phenomena even in low dimension, only that the first Chern class is not a useful topological invariant. Other topological invariants may become relevant, though. Consequently, absent any other symmetries I should regard the vector bundle not as a complex (\ie class~A) vector bundle, but as a class~AI ($U_{\dagger}^2 = + \id_{\Hil}$) or class~AII ($U_{\dagger}^2 = - \id_{\Hil}$) vector bundle, which have been constructed and classified in \cite{DeNittis_Gomi:AI_bundles:2014,DeNittis_Gomi:AII_bundles:2014}. Provided the time-reversal symmetry connects fibers at $k$ and $-k$, in low dimension ($d \leq 3$) class~AI vector bundles over the torus are all trivial (\cf \cite[Theorem~1.6]{DeNittis_Gomi:AI_bundles:2014}); and class~AII vector bundles of $\dim \BZ \leq 4$ are characterized by Kane-Melé-type invariants (\cf \cite[Theorems~1.5 and 1.7]{DeNittis_Gomi:AII_bundles:2014}). 
% subsubsection dealing_with_dagger_symmetries (end)

\subsubsection{Dealing with ($\dagger$-)constraints} % (fold)
\label{my_way:periodic_operators_vector_bundles:constraints}
The presence $W$ in constraints is not changing anything either. That is because by the  the complex vector bundles $\mathcal{E}(\id_{\Hil} - P_{\mathrm{rel}})$ and 
\begin{align*}
	\mathcal{E} \bigl ( W \, (\id_{\Hil} - P_{\mathrm{rel}}) \, W^{-1} \bigr ) \cong \mathcal{E}(\id_{\Hil} - P_{\mathrm{rel}})
\end{align*}
are isomorphic, and the constraint $U \, P_{\mathrm{rel}} \, U^{-1} = \id_{\Hil} - P_{\mathrm{rel}}$ is equivalent to the constraint
\begin{align}
	U \, P_{\mathrm{rel}} \, U^{-1} = W \, (\id_{\Hil} - P_{\mathrm{rel}}) \, W^{-1}
	\label{my_way:eqn:constraint_condition_with_similarity_transform}
\end{align}
\emph{with} $W$ for the purpose of topological classifications. Absent any other symmetries it stands to reason that projections with a constraint of the form~\eqref{my_way:eqn:constraint_condition_with_similarity_transform} are classified as class~AIII vector bundles \cite{DeNittis_Gomi:AIII_bundles:2015}, \emph{provided} $W(k)$ is at least continuous in Bloch momentum $k$. 
% subsubsection dealing_with_dagger_constraints (end)

\subsubsection{The relative index of two projections} % (fold)
\label{my_way:periodic_operators_vector_bundles:relative_index}
There \emph{is} one place where I cannot get rid of $W$, namely for the relative index~\eqref{general:eqn:index_two_projections} for two projections. Usually the relative index is well-defined if $W$ were \emph{unitary} rather than just bounded invertible (and other, technical conditions on the projections are satisfied). But here, $W$ is only bounded invertible, and a more careful analysis is necessary to ensure the index is well-defined and a topological invariant. Even if it were well-defined, I still would have to prove that for a given topological class this relative index of projections can be non-zero (a math problem) and manifests itself in experiment (a physics problem). 

In that case the relative index~\eqref{general:eqn:index_two_projections} would retain some information on the geometry of the system. The definition of the index is entirely algebraic and could be non-zero even when $W = \id_{\Hil}$, \ie when $H$ is normal with respect to the scalar product that makes symmetries (anti)unitary. 

One important point I want to impress upon the reader is that the Hilbert space structure in my arguments is not needed. It suffices that the vector bundle isomorphisms are implemented fiberwise by bounded \emph{invertible} maps rather than unitaries, for example. And symmetries are likewise implemented by bounded, (anti)linear, invertible maps that square to $\pm \id_{\Hil}$. Conversely, any complex vector bundle can be equipped with a family of scalar products on each of the fibers to make it into a hermitian vector bundle (\cf \cite[Proposition~1.2]{Hatcher:vector_bundles_K_theory:2009}).
% subsubsection the_relative_index_of_two_projections (end)
% subsection topological_classification_for_the_periodic_case_is_not_affected (end)

\subsection{Extension to disordered system via $K$-theory} % (fold)
\label{my_way:K_theory}
While my hands-on arguments with Bloch vector bundles only apply to periodic operators, it stands to reason that they extend to disordered systems as well, at least in the weak disorder limit. The standard approach here is to use $K$-theory, \eg \cite{Thiang:K_theoretic_classification_topological_insulators:2016,Alldridge_Max_Zirnbauer:bulk_boundary_correspondences_free_fermion_topological_phases:2020} or \cite{Prodan_Schulz_Baldes:complex_topological_insulators:2016} for the two complex classes; twisted equivariant $K$-theory \cite{Freed_Moore:twisted_equivariant_matter:2013,Gomi:twisted_equivariant_K_theory:2017} only applies to periodic systems. And in principle, it should be possible to use $K$-groups for oblique projections $P = P^2$ (sometimes referred to as idempotents) and bounded invertible operators with bounded inverses rather than unitaries (\cf \eg Chapters~III and IV in \cite{Blackadar:K_theory:2008}); \emph{that strongly suggests that I can perform the topological classification of non-hermitian operators only on the basis of algebraic rather than geometric data (the latter being derived from a scalar product).} 

Unfortunately, given the plethora of approaches to $K$-theory — I could start from vector bundles, twisted crossed product $C^*$-algebras or von Neumann algebras, include equivariants twists, etc. — it would not seem a wise investment to focus on one particular flavor and give the readers all the necessary details, only to shift the conversation to a technical comparison of different $K$-theoretical frameworks. Nevertheless, broadly speaking any $K$-theoretic treatment must be consistent with a vector-bundle-theoretic approach when the operators are periodic; however, $K$-theory may be used to \emph{extend} this classification to more general systems that lack periodicity \cite{Bourne_Prodan:Chern_numbers_aperiodic_systems:2018}. 
% subsection extension_to (end)
% section extending_the_alternative_approach_to_the_general_case (end)
%!TEX root = /Users/max/Dropbox/research/topological classification non-selfadjoint hamiltonians/paper/non_hermitian_classification.tex
\section{Discussion and comparison with literature} % (fold)
\label{discussion}
As this article and others (see \eg \cite{DeNittis_Gomi:K_theoretic_classification_operators_on_Krein_spaces:2019,Bliokh_Leykam_Lein_Nori:topological_classification_homogeneous_electromagnetic_media:2019,Yang_Schnyder_Hu_Chiu:Fermion_doubling_theorems:2019,Yang_Chiu_Fang_Hu:Jones_polynomial_knot_transitions_hermitian_non_hermitian_topological_semimetals:2020,Wojcik_Sun_Bzdusek_Fan:topological_classification_non_hermitian_hamiltonians:2020}) have shown, the classification of non-hermitian operators is still by no means well-understood and the last word has yet to be spoken. This article improves our understanding of three key aspects: 
\begin{enumerate}[(1)]
	\item I have proposed an algorithmic classification procedure that starts with input from physics: after selecting what states are physically relevant, I just need to turn the crank (\cf Sections~\ref{my_way_simplified} and \ref{my_way}). At the end, the problem is reduced to the classification of (pairs of) projections with symmetries and constraints. 
	\item Maintaining diagonalizability of operators is crucial in order to ensure continuity of projections and unitaries, which enter the classification here and in the literature (\cf Sections~\ref{diagonalizable_operators} and \ref{why_diagonalizability_matters}). 
	\item It seems that the classification of diagonalizable non-hermitian operators is based solely on algebraic properties (\eg the spectrum and (anti)commutativity of certain operators) rather than geometric properties (\ie quantities derived from scalar products). 
\end{enumerate}

\subsection{Comparison with the literature} % (fold)
\label{discussion:literature}
While I have not attempted to perform an exhaustive classification and derive a “complete” zoology of non-hermitian operators, I have attempted to propose a generic scheme and shown how to implement it for some example operators. My examples from Section~\ref{my_way_simplified:non_hermitian_examples} give some indication on how it compares with the literature, though, in particular the works \cite{Kawabata_Shiozaki_Ueda_Sato:classification_non_hermitian_systems:2019,Zhou_Lee:non_hermitian_topological_classification:2019}, which broke new ground in our field. 

Generally, the classification procedure here seems to be more general that those two works: when the relevant spectrum is chosen point symmetrically, my algorithm can classify the system in a straightforward fashion even though it does not fit into the point gap/line gap scheme of Kawabata et al. Beyond that, our classifications seem to agree only partially: ignoring the case $\sigma_{\mathrm{rel}} = - \sigma_{\mathrm{rel}}$, the first example from Section~\ref{my_way_simplified:non_hermitian_examples:1} is in perfect agreement with \cite{Kawabata_Shiozaki_Ueda_Sato:classification_non_hermitian_systems:2019}; the second example (Section~\ref{my_way_simplified:non_hermitian_examples:2}) is only in partial agreement though. It is not clear to me whether this is due to a scientific typo on either end as is suggested by the result (compare the classifications in Tables~\ref{my_way_simplified:table:example_2_comparison_false_classification} and \ref{my_way_simplified:table:example_2_comparison_correct_classification}) or a genuine disagreement between our methods. 

Of course, I have not chosen these two examples randomly. I wanted to obtain cases that I could classify using only existing theory. This reveals another weakness in our current understanding of topological insulators — for only very few cases do we have an \emph{exhaustive} classification in terms of topological invariants; exhaustive means that we have a complete list of topological invariants that uniquely label each topological phase. At present even in the Cartan-Altland-Zirnbauer classification, of the ten classes, only 4 are well-understood, namely class~A \cite{Hatcher:vector_bundles_K_theory:2009,Grauert:analytische_Faserungen:1958,DeNittis_Lein:exponentially_loc_Wannier:2011}, AI \cite{DeNittis_Lein:exponentially_loc_Wannier:2011,DeNittis_Gomi:AI_bundles:2014}, AII \cite{DeNittis_Gomi:AII_bundles:2014} and AIII \cite{DeNittis_Gomi:AIII_bundles:2015}. Only for those do we have proofs that we have obtained a complete list of topological invariants. Even then our knowledge is either limited to lower-dimensional spaces (typically $d \leq 4$) or subject to additional conditions (like the stable rank condition for class~A). 

Newer works \cite{Yang_Chiu_Fang_Hu:Jones_polynomial_knot_transitions_hermitian_non_hermitian_topological_semimetals:2020,Wojcik_Sun_Bzdusek_Fan:topological_classification_non_hermitian_hamiltonians:2020} as well as my main results indicate the classifications results obtained in \cite{Kawabata_Shiozaki_Ueda_Sato:classification_non_hermitian_systems:2019,Zhou_Lee:non_hermitian_topological_classification:2019} are insufficient and incomplete in two ways: first of all, Wojcik et al.'s classification via homotopy theory \cite{Wojcik_Sun_Bzdusek_Fan:topological_classification_non_hermitian_hamiltonians:2020} links the topological classification to (non-abelian!) braid groups; in the same vein other works have started to use knot theory \cite{Yang_Chiu_Fang_Hu:Jones_polynomial_knot_transitions_hermitian_non_hermitian_topological_semimetals:2020} to characterize certain topological properties of the system derived from the Fermi surface. I will discuss this aspect in more detail in Section~\ref{discussion:beyond_k_theory} below. 

The second point is not just directly related to \cite{Wojcik_Sun_Bzdusek_Fan:topological_classification_non_hermitian_hamiltonians:2020}, but also the second major aspect of this article, and that is the issue of diagonalizability that has not seen sufficient study. 
% subsection comparison_with_the_literature (end)

\subsection[The diagonalizability assumption: further research is needed]{The diagonalizability assumption: \linebreak further research is needed} % (fold)
\label{discussion:diagonalizability}
Like many articles on topological insulators, I need to involve more math than your average work from theoretical physics. That is because many of the significant contributions — including \eg \cite{Kawabata_Shiozaki_Ueda_Sato:classification_non_hermitian_systems:2019,Shiozaki_Sato_Gomi:band_topology_3d_crystallographic_groups:2018} — arise from collaborations between theoretical and mathematical physicists. The mathematical tools often have yet to be developed and one needs to know the inner workings to use them correctly. 

The issue of diagonalizability belongs in this category: to the best of my knowledge, there exists no universally agreed upon definition of diagonalizable operator on infinite-dimensional Hilbert spaces. Some identify normal operators as diagonalizable (by insisting that one can find a diagonalizing similarity transform that is unitary). Other works insist on pure point spectrum so that the operator has a complete set of proper eigenvectors. Neither are general enough to treat diagonalizable periodic operators that have continuous spectrum due to (non-constant) energy bands. 

The definition of diagonalizability I have given here is to my knowledge new, and I will explore aspects like perturbations of diagonalizable operators in a future work. I expect that diagonalizable operators have all the nice properties of normal operators, \ie they behave just like hermitian operators but may have complex spectrum. Two important ramifications that are of immediate relevance to the topological classification are that spectra should depend continuously on the perturbation parameter and as a consequence, spectral projections are well-defined and continuous in the perturbation parameter; both are \emph{false} for generic non-hermitian operators (\cf Section~\ref{why_diagonalizability_matters}). 

For these reasons, it seems that existing classifications, including \cite{Kawabata_Shiozaki_Ueda_Sato:classification_non_hermitian_systems:2019,Zhou_Lee:non_hermitian_topological_classification:2019,DeNittis_Gomi:K_theoretic_classification_operators_on_Krein_spaces:2019} only apply to \emph{diagonalizable} operators rather than generic non-hermitian operators. At the very least this point should be addressed specifically and explicitly (\eg by resolving the issues mentioned in Section~\ref{why_diagonalizability_matters:gap_Kawabata_classification}). 

I am by no means alone in singling out this sticking point: \eg \cite{Wojcik_Sun_Bzdusek_Fan:topological_classification_non_hermitian_hamiltonians:2020} identifies the line where $H$ has a Jordan block as \emph{the} topological obstacle in their classification. And other works (\eg \cite{Yang_Chiu_Fang_Hu:Jones_polynomial_knot_transitions_hermitian_non_hermitian_topological_semimetals:2020}) also specifically address this point. 
% subsection the_diagonalizability_assumption (end)

\subsection{Deriving bulk-boundary correspondences for non-hermitian systems} % (fold)
\label{discussion:bulk_boundary_correspondences}
Just like \eg \cite{Gong_et_al:topological_phases_non_hermitian_systems:2018,Kawabata_Shiozaki_Ueda_Sato:classification_non_hermitian_systems:2019,Zhou_Lee:non_hermitian_topological_classification:2019} the present work is solely concerned with the \emph{bulk} classification. At the end of the day, this is only the first step towards proving bulk-boundary correspondences~\eqref{intro:eqn:bulk_boundary_correspondence}. Given that the platforms to realize non-hermitian operators often involve classical waves, the physical observables are typically the boundary modes themselves; by preparing wave packets with specific $k$-values and frequencies, the whole boundary $k$-space can be swept by \eg varying the incident angle of a laser relative to the surface normal. In contrast, measuring the transverse conductivity only gives us the net number of edge modes in the Quantum Hall Effect. 

That then leaves the “mathematical” bulk-boundary correspondence, the second equality $T_{\mathrm{bdy}} = f(T_{\mathrm{bulk}})$ in equation~\eqref{intro:eqn:bulk_boundary_correspondence}, which needs to be derived. How does the present work advance the state-of-the-art in this respect?

\subsubsection[Persistence of hermitian topological phenomena in certain non-hermitian systems]{Persistence of hermitian topological phenomena \linebreak in certain non-hermitian systems} % (fold)
\label{discussion:bulk_boundary_correspondences:persistence_hermitian_topological_effects}
It is useful to distinguish topological phenomena that are “non-hermitian versions of topological phenomena in hermitian systems” from \emph{bona fide} non-hermitian topological phenomena that have no hermitian analog. And at least for topological phenomena of the first category, it stands to reason that existing techniques to derive bulk-boundary correspondences can be applied directly. My formalism has the advantage that I can easily make this vague distinction mathematically precise. 

The reason for this is simple: after the first step in my construction I obtain an orthogonal projection $P_{\mathrm{rel}}$, and this projection “no longer remembers” whether the relevant states are associated with real spectrum of a hermitian operator or with spectrum of a diagonalizable operator with complex spectrum. And if $P_{\mathrm{rel}}$ possesses no $\dagger$-symmetries and $\dagger$-constraints, the classification is \emph{identical} to that of a hermitian operator — namely the spectrally flattened hamiltonian $Q = \id_{\Hil} - 2 P_{\mathrm{rel}}$. The similarity is closest if there exists a scalar product with respect to which the symmetry operators are (anti)unitary and $H$ is normal. However, the arguments in Section~\ref{my_way} prove that under mild conditions (specifically Assumption~\ref{my_way:assumption:regulatity_periodic_case}) the classification extends verbatim from normal to \emph{diagonalizable} operators. 

In these circumstances, I can just apply existing techniques to $P_{\mathrm{rel}}$ and/or $Q$. The formalism developed by Schulz-Baldes and Prodan applies to the two complex classes, class~A and class~AIII (\cf \cite[Chapter~7]{Prodan_Schulz_Baldes:complex_topological_insulators:2016}), for instance. Non-hermitian, diagonalizable systems of this kind exhibit topological phenomena with hermitian analogs; and it is for this reason, I call them non-hermitian \emph{generalizations} of hermitian topological phenomena. An example is the theoretically predicted analog of the Quantum Hall Effect in magnonic crystals
\cite{Shindou_et_al:chiral_magnonic_edge_modes:2013,Lein_Sato:topological_classification_magnons:2019}. 
% subsubsection persistence_of_hermitian_topological_in_certain_non_hermitian_systems (end)

\subsubsection{\emph{Bona fide} non-hermitian topological phenomena} % (fold)
\label{discussion:bulk_boundary_correspondences:bona_fide_non_hermitian_topological_phenomena}
In contrast, there are topological phenomena with no hermitian counterpart. That occurs in systems where $\dagger$-symmetries and/or $\dagger$-constraints emerge, which relate $P_{\mathrm{rel}}$ to $P_{\mathrm{rel},\dagger}$ and/or $\id_{\Hil} - P_{\mathrm{rel},\dagger}$. For those systems, our community needs to develop new techniques for proving bulk-boundary correspondences. While there are flexible “meta techniques” such as the Six-Term Exact Sequence approach that have been used to great effect \cite{Kellendonk_Richter_Schulz-Baldes:edge_currents_Chern_numbers_quantum_Hall:2002,Prodan_Schulz_Baldes:complex_topological_insulators:2016,Leung_Prodan:bulk_boundary_correspondence_magneto_electric_effect:2020}, it stands to reason that adapting them to \eg more general $K$-theories is not straightforward and will likely involve hard mathematical work. Nevertheless, this is absolutely necessary if we truly want to understand non-hermitian topological phenomena. 

To give one fascinating example: a recent paper \cite{Bliokh_Leykam_Lein_Nori:topological_classification_homogeneous_electromagnetic_media:2019} has proposed that the presence and polarization of electromagnetic interface modes between “metals” ($\sgn \eps = - \sgn \mu$) and “dielectrics” ($\sgn \eps = + \sgn \mu$) can be explained via two bulk-boundary correspondences; similar topological phenomena have been found in other classical wave equations \cite{Bliokh_Nori:transverse_spin_surface_waves_acoustic_metamaterials:2019,Leykam_Bliokh_Nori:topologial_electromagnetic_edge_modes_slab_waveguides:2020}. If I put my mathematical physicist's hat on, I would be more cautious and say these are \emph{conjectures} of bulk-boundary correspondences. Preliminary research shows that the relevant bulk operators are of class~AI and class~${\mathrm{D}^{\dagger}} \simeq \mathrm{AI}$. So the bulk operators are, at least as far as existing theory is concerned, topologically trivial. Yet, the interface formed between two different, seemingly topologically trivial systems is topologically \emph{non}-trivial. The paper proposes a bulk classification, which does not seem to fit the mold of any of the current classification schemes. Finding the mechanism and formalizing the mathematical principles would allow us to systematically predict novel topological phenomena with no analogs in hermitian systems. 
% subsubsection bona_fide_non_hermitian_topological_phenomena (end)

\subsubsection{Dependence on boundary conditions} % (fold)
\label{discussion:bulk_boundary_correspondences:boundary_conditions}
One last big issue in non-hermitian systems is the question whether and how bulk-boundary correspondences depend on the choice of boundary conditions; this question is also relevant for certain hermitian continuum systems. There are cases where boundary conditions seem to break bulk-boundary correspondences \cite{Graf_Jud_Tauber:bulk_boundary_correspondence_boundary_conditions:2020}. Boundary conditions may sometimes also break bulk symmetries, \eg in the language of \cite{DeNittis_Lein:symmetries_electromagnetism:2020} if a dielectric electromagnetic medium with time-reversal symmetry $T_1 = (\sigma_1 \otimes \id) \, C$ is terminated by a perfect electric conductor (\ie we choose PEC boundary conditions), then these boundary conditions break $T_1$ time-reversal symmetry. 

Another direction that has seen a lot of attention in the physics community are works comparing systems with open and periodic boundary conditions (\eg \cite{Yokomizo_Murakami:non_Bloch_band_theory_non_hermitian_systems:2019} or \cite{Bergholtz_Budich_Kunst:review_exceptional_points:2020} for a current review) and connected phenomena like the non-hermitian skin effect \cite{Yao_Wang:edge_states_topological_invariatnts_non_hermitian_systems:2018,Okuma_Kawabata_Shiozaki_Sato:topological_origin_non_hermitian_skin_effect:2020}. 
% subsubsection dependence_on_boundary_conditions (end)
% subsection deriving_bulk_boundary_correspondences_for_non_hermitian_systems (end)

\subsection[Going beyond $K$-theory: utilizing the theory of braids, knots and weaves]{Going beyond $K$-theory: \linebreak utilizing the theory of braids, knots and weaves} % (fold)
\label{discussion:beyond_k_theory}
Within the last two, three years researchers have begun looking beyond $K$-theory to classify the topology of physical systems. On the one hand, this has become necessary, because even in some simple non-hermitian systems, an exhaustive classification can only be classified in terms of \emph{non-abelian} (non-commutative) groups \cite{Wojcik_Sun_Bzdusek_Fan:topological_classification_non_hermitian_hamiltonians:2020}. And given that all $K$-groups are necessarily abelian, at least some aspects of the systems's topology cannot be captured by a $K$-theoretic classification. 

On the other hand, going beyond $K$-theory could open the door to new topological phenomena. There are several works \cite{Yang_Schnyder_Hu_Chiu:Fermion_doubling_theorems:2019,Yang_Chiu_Fang_Hu:Jones_polynomial_knot_transitions_hermitian_non_hermitian_topological_semimetals:2020} that apply knot theory to periodic systems. These characterize certain topological features of the Fermi surface and derived quantities. Perhaps other structures beyond knots \cite{Murasugi:knot_theory:1996} such as weaves \cite{Grishanov_Meshkov_Omelchenko:topology_textile_structures_1:2009,Grishanov_Meshkov_Omelchenko:topology_textile_structures_2:2009,Mahmoudi:classification_weaves:2020} or other structures can be obtained by entangling (energy level sets of) energy bands like threads with one another. While it is not yet clear whether and in what ways those topological features manifest themselves in experiment, this is clearly a very promising avenue to explore. 
% subsection going_beyond_k_theory_utilizing_the_theory_of_knots_and_braids (end)

\subsection{Other classifications of certain non-hermitian operators} % (fold)
\label{discussion:other_symmetries}
Classification problems in mathematics are as rare as grains of sand on a beach. So choosing the right one is important. And there is usually a trade-off: I could impose less assumptions and assume less structure, which leads to a coarser, but more general classification; or I could do the opposite, make more assumptions and obtain a finer classification. For instance, the classification of pseudohermitian (Krein-hermitian) systems in \cite{DeNittis_Gomi:K_theoretic_classification_operators_on_Krein_spaces:2019} is finer than that of \cite{Zhou_Lee:non_hermitian_topological_classification:2019,Kawabata_Shiozaki_Ueda_Sato:classification_non_hermitian_systems:2019} for this reason; a second example are topological insulators with crystalline symmetries \cite{Gomi:topological_classification_crystallographic_point_groups_2d:2017,Shiozaki_Sato_Gomi:band_topology_3d_crystallographic_groups:2018}. 

That being said, my results here suggest that (anti)unitarity of symmetries is not important in the setting of \cite{Zhou_Lee:non_hermitian_topological_classification:2019,Kawabata_Shiozaki_Ueda_Sato:classification_non_hermitian_systems:2019}. More precisely, the assumption of (anti)unitarity can be relaxed to bounded with bounded inverse. Since diagonalizable operators $H = H_{\Re} + \ii H_{\Im}$ are exactly those that can be split into two \emph{commuting}, hermitian operators $H_{\Re} = H_{\Re}^{\ddagger}$ and $H_{\Im} = H_{\Im}^{\ddagger}$, usual, “non-$\dagger$” symmetries $U_j$, $j = 1 , \ldots$, are of the form 
\begin{align*}
	U_j \, H \, U_j^{-1} = \pm W_j^{-1} \, H \, W_j
\end{align*}
where $W_j \in \mathcal{B}(\Hil)^{-1}$ is a similarity transform. 

$\dagger$-symmetries are those that relate the hamiltonian $H$ to $H^{\ddagger} = H_{\Re} - \ii H_{\Im}$, namely 
\begin{align*}
	U_{\dagger,j} \, H \, U_{\dagger,j}^{-1} = \pm W_{\dagger,j}^{-1} \, H^{\ddagger} \, W_{\dagger,j}
	, 
\end{align*}
where again $W_{\dagger,j} \in \mathcal{B}(\Hil)^{-1}$ is a similarity transform. In case the operators $W_j$ and $W_{\dagger,j}$ are “nice enough”, \eg when $H$ and the similarity transforms are periodic, and their fiber operators $H(k)$, $W_j(k)$ and $W_{\dagger,j}(k)$ are continuous in Bloch momentum $k$, the arguments of Section~\ref{my_way:periodic_operators_vector_bundles} apply verbatim. 

Of course, in general, these symmetries need not (anti)commute with one another, so the situation is more general than that considered in \cite{Kawabata_Shiozaki_Ueda_Sato:classification_non_hermitian_systems:2019} even when the operators are all (anti)unitary. 

There is also another situation that is currently not well-understood: what if we consider interfaces between topological insulators of different classes? That is the situation at metal-dielectric interfaces between homogeneous electromagnetic media; the electromagnetic surface modes have been shown to be topological \cite{Bliokh_Leykam_Lein_Nori:topological_classification_homogeneous_electromagnetic_media:2019} since their presence is explained by a bulk-boundary correspondence. The relevant bulk operators are hermitian and of class~AI on the dielectric side, and anti- as well as pseudohermitian and of class~$\mathrm{D}^{\dagger} \simeq \mathrm{AI}$ on the metallic side. The current state-of-the-art \cite{Zhou_Lee:non_hermitian_topological_classification:2019,Kawabata_Shiozaki_Ueda_Sato:classification_non_hermitian_systems:2019} predicts that the bulk systems are topologically trivial. Nevertheless, if I sandwich two topologically trivial bulk systems from different classes, I still get topologically protected interface modes. Here, the pertinent factor seems to be the change in the fundamental nature of the geometric structure — from “Riemannian” to “Minkowskian” — which seems to be at the heart of this topological phenomenon. 
% subsection other_symmetries (end)
% section discussion (end)

\section*{Acknowledgements} % (fold)
\label{sec:acknowledgements}
The authors has been supported by JSPS through a Wakate~B (grant number 16K17761) and a Kiban~C grant (grant number 20K03761) as well as a Fusion grant from the WPI-AIMR. The author thanks Chris Bourne, Ching-Kai Chiu, Giuseppe De~Nittis, Shanhui Fan, Flore Kunst, Koji Sato and Casey Wojcik for their helpful input and encouragement during private discussions. 
% section acknowledgements (end)

%!TEX root = /Users/max/Dropbox/research/topological classification non-selfadjoint hamiltonians/paper/non_hermitian_classification.tex
%
\begin{appendix}
	\section[Relation between biorthogonal calculus and the weighted scalar product]{Relation between biorthogonal calculus \linebreak and the weighted scalar product} % (fold)
	\label{appendix:biorthogonal_calculus_weighted_Hilbert_spaces}
	The biorthogonal calculus that is commonly used in the physics community is a cumbersome way of using a weighted scalar product. This appendix will show the equivalence of the two.

	\subsection{A diagonalizable, but not obviously normal $2 \times 2$ matrix} % (fold)
	\label{appendix:biorthogonal_calculus_weighted_Hilbert_spaces:2x2_matrix}
	My example starts with our choice of eigenvalues, $1$ and $\ii$. I choose $g_1 = (1,0)^{\mathrm{T}}$ and $g_2 = (1,1)^{\mathrm{T}}$ as the corresponding eigenvectors. When I collect this information into matrix form, I arrive at 
	\begin{align*}
		D &= \left (
		\begin{matrix}
			1 & 0 \\
			0 & \ii \\
		\end{matrix}
		\right )
		,
		\\
		G^{-1} &= \left (
		\begin{matrix}
			1 & 1 \\
			0 & 1 \\
		\end{matrix}
		\right )
		. 
	\end{align*}
	The $2 \times 2$ matrix I am then interested in is obtained by similarity transform, 
	\begin{align*}
		H &= G^{-1} \, D \, G 
		% \\
		% &
		= \left (
		\begin{matrix}
			1 & 1 \\
			0 & 1 \\
		\end{matrix}
		\right ) \left (
		\begin{matrix}
			1 & 0 \\
			0 & \ii \\
		\end{matrix}
		\right ) \left (
		\begin{matrix}
			1 & -1 \\
			0 & 1 \\
		\end{matrix}
		\right )
		% \\
		% &
		= \left (
		\begin{matrix}
			1 & -1 + \ii \\
			0 & \ii \\
		\end{matrix}
		\right )
		. 
	\end{align*}
	Once I compute the usual hermitian adjoint of $H$, I can easily convince myself that $H \, H^{\dagger}$ and $H^{\dagger} \, H$ disagree. 

	Indeed, respect to the usual, Euclidean scalar product 
	\begin{align*}
		\scpro{\varphi}{\psi}_{\C^2} &= \overline{\varphi_1} \, \psi_1 + \overline{\varphi_2} \, \psi_2 
	\end{align*}
	the two eigenvectors, \ie the column vectors of $G$, are not orthonormal to one another. So $H$ is not normal with respect to the \emph{Euclidean} scalar product. 

	But after choosing an adapted scalar product 
	\begin{align}
		\scppro{\varphi}{\psi} &= \scpro{G \varphi}{G \psi}_{\C^2} 
		, 
		\label{appendix:biorthogonal_calculus_weighted_Hilbert_spaces:2x2_matrix:eqn:G_scalar_product}
	\end{align}
	I can make them orthonormal by \emph{definition}; as the notation suggests, this is nothing but the biorthogonal scalar product. That is because $G$ maps the eigenvectors $g_j$ of $H$ onto the canonical basis vectors $e_j$, which makes $g_1$ and $g_2$ orthonormal, 
	\begin{align*}
		\sscppro{g_j}{g_k} &= \scpro{G g_j \, }{G g_k}_{\C^2}
		= \sscpro{e_j}{e_k}_{\C^2} = \delta_{jk} 
		. 
	\end{align*}
	I can define a hermitian adjoint 
	\begin{align}
		A^{\ddagger} &= (G G^{\dagger})^{-1} \, A^{\dagger} \, (G G^{\dagger}) 
		\label{appendix:biorthogonal_calculus_weighted_Hilbert_spaces:2x2_matrix:eqn:G_adjoint}
		\\
		&= G^{-1} \, \bigl ( G \, A \, G^{-1} \bigr )^{\dagger} \, G
		\notag 
	\end{align}
	with respect to $\scppro{\, \cdot \,}{\, \cdot \,}$, \ie the matrix which satisfies 
	\begin{align*}
		\bscppro{A^{\ddagger} \varphi \, }{ \, \psi} &= \bscppro{\varphi \, }{ \, A \psi} 
		. 
	\end{align*}
	Because the hermitian adjoint 
	\begin{align*}
		H^{\dagger} &= \bigl ( G^{-1} \, D \, G \bigr )^{\dagger} 
		= G^{\dagger} \, \overline{D} \, \bigl ( G^{-1} \bigr )^{\dagger} 
	\end{align*}
	is diagonalized by conjugating with 
	\begin{align*}
		G^{\dagger} &= \left (
		\begin{matrix}
			1 & 0 \\
			-1 & 1 \\
		\end{matrix}
		\right )
		, 
	\end{align*}
	and the eigenvectors of $H^{\dagger}$ are just the two column vectors of this matrix. The conventional “bihermitian” approach now suggests to look at the operators 
	\begin{align*}
		\sopro{\psi_{R,1}}{\psi_{L,1}} &= \left (
		\begin{matrix}
			1 \\ 
			0 \\
		\end{matrix}
		\right ) \, \left (
		\begin{matrix}
			1 & -1 \\
		\end{matrix}
		\right ) 
		= \left (
		\begin{matrix}
			1 & -1 \\
			0 & 0 \\
		\end{matrix}
		\right )
		, 
		\\
		\sopro{\psi_{R,\ii}}{\psi_{L,-\ii}} &= \left (
		\begin{matrix}
			1 \\ 
			1 \\
		\end{matrix}
		\right ) \, \left (
		\begin{matrix}
			0 & 1 \\
		\end{matrix}
		\right ) 
		= \left (
		\begin{matrix}
			0 & 1 \\
			0 & 1 \\
		\end{matrix}
		\right )
		. 
	\end{align*}
	Because the two vectors are eigenvectors of $H$ and $H^{\dagger}$ to complex conjugate eigenvalues, they square to themselves, \ie they are (potentially oblique) projections, 
	\begin{align*}
		\bigl ( \sopro{\psi_{R,E}}{\psi_{L,\bar{E}}} \bigr )^2 = \sopro{\psi_{R,E}}{\psi_{L,\bar{E}}}
		. 
	\end{align*}
	With respect to the usual, Euclidean scalar product, these two projections are \emph{not} hermitian. But if we instead use the scalar product~\eqref{appendix:biorthogonal_calculus_weighted_Hilbert_spaces:2x2_matrix:eqn:G_scalar_product} and the corresponding hermitian adjoint~\eqref{appendix:biorthogonal_calculus_weighted_Hilbert_spaces:2x2_matrix:eqn:G_adjoint}, we can confirm with ease that these operators are indeed $G$-hermitian, 
	\begin{align*}
		\sopro{\psi_{R,E}}{\psi_{L,\bar{E}}}^{\ddagger} &= \soppro{\psi_{R,E}}{\psi_{R,E}}^{\ddagger} 
		\\
		&= \soppro{\psi_{R,E}}{\psi_{R,E}} 
		= \sopro{\psi_{R,E}}{\psi_{L,\bar{E}}}
		. 
	\end{align*}
	%
	% subsection a_diagonalizable_but_not_obviously_normal_2_times_2_matrix (end)

	\subsection[Generalization to $\C^N$ and infinite-dimensional Hilbert spaces]{Generalization to $\C^N$ and \linebreak infinite-dimensional Hilbert spaces} % (fold)
	\label{sub:generalization_to_c_n_and_infinite_dimensional_hilbert_spaces}
	These arguments evidently generalize to $N \times N$ matrices and operators on separable, infinite-dimensional vector spaces that admit a complete set of eigenvectors. All I need to do is \emph{declare} the eigenvectors of $H$ to be orthogonal to each other and have unit length by mapping them onto canonical basis vectors, 
	\begin{align*}
		G^{-1} e_n = \varphi_{R,n} 
		, 
	\end{align*}
	where $e_n = (\delta_{jn})_{j = 1 , \ldots , N}$ and $N$ is either finite or $\infty$. The column vectors of $G^{-1} = \sum_{n = 1}^N \sopro{\varphi_{R,n}}{e_n}$ are nothing but the right-eigenvectors and the column vectors of $G^{\dagger}$ are the left-eigenvectors. 

	The fact that left- and right-eigenvectors sum to the identity follows directly from 
	\begin{align*}
		\id &= \sum_{n = 1}^N \sopro{e_n}{e_n} 
		= \sum_{n = 1}^N \sopro{G \varphi_{R,n}}{G \varphi_{R,n}}
		\\
		&= \sum_{n = 1}^N \soppro{\varphi_{R,n}}{\varphi_{R,n}} 
		= \sum_{n = 1}^N \sopro{\varphi_{R,n}}{\varphi_{L,n}}
		. 
	\end{align*}
	Incorporating $G$ into the scalar product yields another sesquilinear form that satisfies all the axioms of a scalar product, \ie $\scppro{\varphi}{\psi} = \scpro{G \varphi}{G \psi}$ \emph{is} a scalar product. And with respect to this scalar product, all of the rank-$1$ projections 
	\begin{align*}
		\sopro{\varphi_{R,n}}{\varphi_{L,n}}^{\ddagger} &= \soppro{\varphi_{R,n}}{\varphi_{R,n}}^{\ddagger}
		\\
		&= \soppro{\varphi_{R,n}}{\varphi_{R,n}}
		= \sopro{\varphi_{R,n}}{\varphi_{L,n}}
	\end{align*}
	are hermitian with respect to $\sscppro{\, \cdot \,}{\, \cdot \,}$. That means the $\sopro{\varphi_{R,n}}{\varphi_{L,n}}$ are a collection of $\sscppro{\, \cdot \,}{\, \cdot \,}$-orthogonal projections. 

	Once more I can check that $H$ is diagonalized by $G$, and I can check as in the $2 \times 2$-matrix case that $H$ is normal with respect to the scalar product $\scppro{\varphi}{\psi}$. 

	When the spectrum does not just consist of eigenvalues, making this argument is trickier as it is not clear what diagonalizable precisely means in a mathematical sense. Certainly, in some cases, I can still use the biorthogonal calculus, \eg when I am dealing with periodic tight-binding operators. After Bloch-Floquet decomposition, I am left with a matrix-valued function of $k$, and the spectrum of matrices evidently consists solely of eigenvalues. Then the above arguments can be adapted, although now $G = G(k)$ must also be a function of $k$. 

	However, when \eg disorder is present, I cannot adapt the above construction by hand, there is no simple way to make explicit use of the biorthogonal calculus. Nevertheless, as long as 
	\begin{align*}
		\id_{\Hil} &= \int_{\sigma(H)} \dd \sopro{\psi_{R,E}}{\psi_{L,\bar{E}}} 
		\\
		&= \int_{\sigma(H)} \dd \soppro{\psi_{R,E}}{\psi_{R,E}} 
	\end{align*}
	holds true, this defines a projection-valued measure on the complex plane; $\soppro{\psi_{R,E}}{\psi_{R,E}}$ is a formal expression, mathematically speaking it should be replaced by the projection-valued measure (\cf Definition~\ref{appendix:functional_calculus_normal_operators:defn:projection_valued_measure} below). 
	
	By a choice of scalar product, I can make this projection-valued measure hermitian. From the projection-valued measure I can REconstruct two hermitian operators, which I will dutifully denote with 
	\begin{align*}
		H_{\Re} = \int_{\C} (\Re E) \; \dd \soppro{\psi_{R,E}}{\psi_{R,E}}
	\end{align*}
	and a similarly defined $H_{\Im}$. Because the projection-valued measures commute with one another, $H_{\Re}$ and $H_{\Im}$ commute. Thus, the total operator $H = H_{\Re} + \ii H_{\Im}$ is normal. 

	This decomposition also shows directly that the spectra of $H$ and $H^{\dagger}$ are related by complex conjugation, 
	\begin{align*}
		\sigma(H^{\dagger}) = \overline{\sigma(H)}
		, 
	\end{align*}
	and that the spectral projections are similarly related, 
	\begin{align*}
		1_{\sigma_{\mathrm{rel}}}(H) = 1_{\overline{\sigma_{\mathrm{rel}}}}(H^{\dagger}) 
		. 
	\end{align*}
	This expression will be defined in the next section of this Appendix. 
	% subsection generalization_to_c_n_and_infinite_dimensional_hilbert_spaces (end)
	% section relation_between_biorthogonal_calculus_and_using_weighted_hilbert_spaces (end)

	\section{Functional calculus for normal operators} % (fold)
	\label{appendix:functional_calculus_normal_operators}
	Normal operators admit a functional calculus, \ie a systematic way to assign an operator $f(H)$ to a suitable function $f$. The most prominent examples are the time-evolution $\e^{- \ii t H}$ where $f(E) = \e^{- \ii t E}$ and spectral projections, which arise from functional calculus for the indicator functions 
	\begin{align*}
		f(E) = 1_{\Omega}(E) = 
		\begin{cases}
			1 & E \in \Omega \\
			0 & \mbox{else} \\
		\end{cases}
		. 
	\end{align*}
	The set $\Omega \subseteq \C$ is comprised of the relevant energies or frequencies. 
	
	The collection of spectral projections gives rise to the so-called projection-valued measure, which makes expressions like $\dd P(E) = \dd \sopro{\psi_E}{\psi_E}$ mathematically rigorous. The following definition is a straightforward extension of the hermitian case (see \eg \cite[Chapter~3.1]{Teschl:quantum_mechanics:2009}). 
	\begin{definition}[Projection-valued measure]\label{appendix:functional_calculus_normal_operators:defn:projection_valued_measure}
		Let $\mathfrak{B}$ be the Borel $\sigma$-algebra on $\C$. Then a projection-valued measure is a map from the Borel $\sigma$-algebra to the orthogonal projections, 
		\begin{align*}
			P : \mathfrak{B} \longrightarrow \mathcal{B}(\Hil)
			, 
			\quad 
			\Omega \mapsto P(\Omega)
			= P(\Omega)^2 = P(\Omega)^{\dagger}
			, 
		\end{align*}
		such that the following two conditions hold: 
		\begin{enumerate}[(a)]
			\item $P(\C) = \id_{\Hil}$
			\item If $\Omega = \bigcup_n \Omega_n$ is the union of mutually disjoint sets, $\Omega_n \cap \Omega_j = \emptyset$ for all $n \neq j$, then $\sum_n P(\Omega_n) \psi = P(\Omega) \psi$ holds for all $\psi \in \Hil$ (strong $\sigma$-additivity). 
		\end{enumerate}
	\end{definition}
	The first condition is nothing but the well-known completeness condition, 
	\begin{align*}
		P(\C) = \int_{\C} \dd P(E) = \int_{\C} \dd \sopro{\psi_E}{\psi_E}
		= \id_{\Hil}
		. 
	\end{align*}
	The range of $P(\Omega)$ are the states of energies/frequencies contained in the set $\Omega \subseteq \C$. These defining properties imply among other things that 
	\begin{align*}
		P(\Omega) \, P(\Lambda) = P(\Omega \cap \Lambda) 
		= P(\Lambda) \, P(\Omega)
		. 
	\end{align*}
	Depending on the approach to functional calculus, one could either construct the functional calculus from the projection-valued measure or the other way around. I will start with the projection-valued measure, which can be constructed using functional calculus for \emph{hermitian} operators (again, \cf \cite[Chapter~3.1]{Teschl:quantum_mechanics:2009}): I start by splitting $H = H_{\Re} + \ii H_{\Im}$ into real and imaginary parts, 
	\begin{align*}
		H_{\Re} &= \frac{1}{2} \bigl ( H + H^{\dagger} \bigr ) 
		,
		\\
		H_{\Im} &= \frac{1}{\ii 2} \bigl ( H - H^{\dagger} \bigr ) 
		.
	\end{align*}
	By their very definition, real and imaginary part operators are hermitian. And importantly, since $[H , H^{\dagger}] = 0$, the two commute with one another as well. 
	
	Consequently, $H_{\Re} = H_{\Re}^{\dagger}$ and $H_{\Im} = H_{\Im}^{\dagger}$ admit a functional calculus \cite[Theorem~3.1]{Teschl:quantum_mechanics:2009}, which gives meaning to $f(H_{\Re,\Im})$ for any bounded Borel function $f : \R \longrightarrow \C$ on $\R$. Initially, I can make sense of $P(\Omega)$ for product sets $\Omega = \Omega_{\Re} \times \Omega_{\Im}$ that are also Borel: I define the associated projection as 
	\begin{align*}
		P(\Omega) &\overset{\mathrm{def}}{=} 1_{\Omega_{\Re}}(H_{\Re}) \, 1_{\Omega_{\Im}}(H_{\Im}) 
		\\
		&= 1_{\Omega_{\Im}}(H_{\Im}) \, 1_{\Omega_{\Re}}(H_{\Re}) 
		. 
	\end{align*}
	As product sets are a base for the topology given by Borel sets on $\C$, this definition extends to arbitrary Borel sets on $\C$. Long story short, this gives a mathematically rigorous definition of expressions like 
	\begin{align*}
		f(H) = \int_{\C} f(E) \; \dd P(E) 
	\end{align*}
	as well for suitable functions $f : \C \longrightarrow \C$. 
	\begin{theorem}[Functional calculus for normal operators]\label{appendix:functional_calculus_normal_operators:thm:functional_calculus}
		Let $H$ be a normal, bounded operator on a Hilbert space $\Hil$, and suppose $\mu \in \C$ is a scalar and $f , g \in \mathrm{Bo}_{\mathrm{b}}(\C,\C)$ bounded Borel functions. Then the map 
		\begin{align*}
			\mathrm{Bo}_{\mathrm{b}}(\R,\C) \ni f \mapsto f(H) \in \mathcal{B}(\Hil)
		\end{align*}
		has the following properties: 
		\begin{enumerate}[(1)]
			\item $f \mapsto f(H)$ is a $\ast$-homomorphism, \ie 
			\begin{align*}
				(f + \mu g)(H) &= f(H) + \mu \, g(H)
				,
				\\
				(f \, g)(H) &= f(H) \, g(H) 
				,
				\\
				1_{\C}(H) &= \id_{\Hil}
				,
				\\
				f(H)^{\dagger} &= \bar{f}(H)
				. 
			\end{align*}
			\item If $H \psi = E \psi$, then $f(H) \psi = f(E) \, \psi$.
			\item $f \geq 0$ $\Rightarrow$ $f(H) \geq 0$
		\end{enumerate}
	\end{theorem}
	Also diagonalizable operators admit a functional calculus, which is what I will talk about next. 
	% section functional_calculus_for_normal_operators (end)

	\section{Diabonalizable operators} % (fold)
	\label{appendix:diabonalizable_operators}
	While for matrices there is an unambiguous definition of diagonalizability, there is no universally accepted definition in the mathematics literature for operators on infinite-dimensional Hilbert spaces. Some authors require the operator to possess a basis of proper eigenvectors. As a result, the operator has to be compact and must possess pure point spectrum. That is far too restrictive for my purposes. In contrast, unless all bands are flat, periodic operators have at least continuous spectrum coming from the energy bands. 
	
	Others equate diagonalizability with \emph{unitary} or orthogonal diagonalizability, which singles out normal operators. This class is also unnecessarily small. I will opt for a generalization that can handle continuous spectrum, yet is consistent with the definition of matrices.

	\subsection{Characterizations of diagonalizability} % (fold)
	\label{apprendix:characterizations_diagonalizability}
	To accommodate operators with continuous spectrum and not limit myself to normal operators I declare operators diagonalizable if they are normal after a similarity transform. Just like in the case of matrices, a \emph{similarity transform} $G \in \mathcal{B}(\Hil)$ is a bounded operator with bounded inverse $G^{-1} \in \mathcal{B}(\Hil)$; I will abbreviate this class of operators with $\mathcal{B}(\Hil)^{-1}$, although also $\mathrm{GL}(\Hil)$ is commonly used. 
	\begin{definition}[Diagonalizable operator]
		A bounded operator $H \in \mathcal{B}(\Hil)$ on a Hilbert space is called diagonalizable if there exists a similarity transform $G \in \mathcal{B}(\Hil)^{-1}$ for which 
		\begin{align}
			G \, H \, G^{-1} = \int_{\C} E \, \dd P(E)
			\label{appendix:diagonalizable_operators:eqn:spectral_decomposition}
		\end{align}
		admits a spectral decomposition where $P(E)$ is a projection-valued measure on $\C$. 
	\end{definition}
	The flip side of having no established definition is that I cannot point to the literature and then solely focus on the physics. Instead, I will need to establish certain relevant mathematical facts myself. 
	\begin{theorem}
		The following are equivalent characterizations of diagonalizability: 
		\begin{enumerate}[(1)]
			\item $H$ is diagonalizable. 
			\item There exists a similarity transform $G \in \mathcal{B}(\Hil)^{-1}$ so that $G \, H \, G^{-1}$ is normal. 
			\item There exists a similarity transform $G \in \mathcal{B}(\Hil)^{-1}$ so that $G \, H \, G^{-1}$ admits a functional calculus $f \mapsto f \bigl ( G \, H \, G^{-1} \bigr )$, \ie a systematic way to associate an operator $f(H)$ to suitable functions $f : \C \longrightarrow \C$ (\cf Appendix~\ref{appendix:functional_calculus_normal_operators}). 
		\end{enumerate}
	\end{theorem}
	\begin{proof}
		“(1) $\Longrightarrow$ (2):” $P(\Omega)$ that enters equation~\eqref{appendix:diagonalizable_operators:eqn:spectral_decomposition} is a projection-valued measure. By definition projection-valued measures $P(\Omega)^2 = P(\Omega) = P(\Omega)^{\dagger}$ take values in the \emph{orthogonal} projections. Therefore, we can express the adjoint 
		\begin{align*}
			\bigl ( G \, H \, G^{-1} \bigr )^{\dagger} &= \int_{\C} \overline{E} \, \dd P(E)
		\end{align*}
		in terms of the same projection-valued measure. The commutator of $G \, H \, G^{-1}$ and its adjoint vanishes, because the spectral projections commute, 
		\begin{align*}
			P(\Omega) \, P(\Lambda) = P (\Omega \cap \Lambda) = P(\Lambda) \, P(\Omega) 
			. 
		\end{align*}
		“(2) $\Longrightarrow$ (3):” Since $G \, H \, G^{-1} = H_{\Re} + \ii H_{\Im}$ is a normal operator, Theorem~\ref{appendix:functional_calculus_normal_operators:thm:functional_calculus} tells us that there exists a functional calculus and a projection-valued measure for $G \, H \, G^{-1}$. 
		
		The projection-valued measure now gives rise to a functional calculus: by mimicking the construction in Appendix~\ref{appendix:functional_calculus_normal_operators} or \cite[Chapter~3.1]{Teschl:quantum_mechanics:2009} I obtain a functional calculus that associates an operator to each bounded  Borel function $f : \C \longrightarrow \C$. 
		\smallskip
		
		\noindent
		“(3) $\Longrightarrow$ (1):” If I am given a functional calculus, then it defines a projection-valued measure via $P(\Omega) = 1_{\Omega} \bigl ( G \, H \, G^{-1} \bigr )$. By approximating the function $f(E) = E$ by linear combinations of step functions on $\sigma(H)$, I can approximate 
		\begin{align*}
			f \bigl ( G \, H \, G^{-1} \bigr ) = G \, H \, G^{-1} 
			= \int_{\C} E \, \dd P(E) 
		\end{align*}
		by simple functions and then take the limit. The limit then gives the above integral~\eqref{appendix:diagonalizable_operators:eqn:spectral_decomposition}, which is nothing but the spectral decomposition. 
	\end{proof}
	I can alternatively give a more geometric interpretation of the similarity transform along the lines of the $2 \times 2$ matrix example from Appendix~\ref{appendix:biorthogonal_calculus_weighted_Hilbert_spaces:2x2_matrix}. 
	\begin{theorem}\label{appendix:diagonalizable_operators:thm:weighted_characterization}
		The following are equivalent characterizations of diagonalizability: 
		\begin{enumerate}[(1)]
			\item $H$ is diagonalizable. 
			\item $H$ is normal with respect to a suitably chosen scalar product $\scppro{\, \cdot \,}{\, \cdot \,}$ on the vector space $\Hil$, \ie $[H , H^{\ddagger}] = 0$. 
			\item $H$ admits a functional calculus $f \mapsto f(H)$, \ie a systematic way to associate an operator $f(H)$ to suitable functions $f : \C \longrightarrow \C$ (\cf Appendix~\ref{appendix:functional_calculus_normal_operators}).
		\end{enumerate}
	\end{theorem}
	\begin{proof}
		“(1) $\Longrightarrow$ (2):” The proof is quite similar to that of the previous theorem, this time we just incorporate the similarity transform $G$ that makes $G \, H \, G^{-1}
		$ normal into a second scalar product 
		\begin{align*}
			\scppro{\varphi}{\psi} \overset{\mathrm{def}}{=} \bscpro{G \varphi \, }{ \, G \psi} 
			= \bscpro{\varphi \, }{ \, (G^{\dagger} G) \, \psi} 
		\end{align*}
		and argue that $H$ is normal with respect to it. Of course, this scalar product is just the biorthogonal scalar product. 
		
		The weighted adjoint ${}^{\ddagger}$ is given by equation~\eqref{diagonalizable_operators:eqn:G_adjoint}. A quick computation confirms that $H$ and $H^{\ddagger}$ commute: after adding $\id_{\Hil} = G \, G^{-1} = G^{\dagger} \, (G^{\dagger})^{-1}$ where necessary, I can factor $G$ and its inverse out of the commutator, 
		\begin{align*}
			0 &= \Bigl [ G \, H \, G^{-1} \, , \, \bigl ( G \, H \, G^{-1} \bigr )^{\dagger} \Bigr ]
			\\
			&= G \, H \, G^{-1} \, (G^{\dagger})^{-1} \, H^{\dagger} \, G^{\dagger} \, G \, G^{-1} 
			+ \\
			&\quad 
			- G \, G^{-1} \, (G^{\dagger})^{-1} \, H^{\dagger} \, G^{\dagger} \, G \, H \, G^{-1}
			. 
			% \\
			% &= G \, \Bigl (  \Bigr ) \, G^{-1}
			% G \, [H , H^{\ddagger}] \, G^{-1} &= \bigl [ G \, H \, G^{-1} \, , \, G \, H^{\ddagger} \, G^{-1} \bigr ]
			% , 
		\end{align*}
		Then I plug in \eqref{diagonalizable_operators:eqn:G_adjoint} for $H^{\ddagger} = (G G^{\dagger})^{-1} \, H^{\dagger} \, (G G^{\dagger})$,  
		\begin{align*}
			0 &= G \, \bigl [ H , H^{\ddagger} \bigr ] \, G^{-1} 
			, 
		\end{align*}
		which is just as good as $\bigl [ H , H^{\ddagger} \bigr ] = 0$ since $G$ and its inverse are bounded. Hence, $H$ and $H^{\ddagger}$ commute. This is the definition of normality with respect to $\scppro{\, \cdot \,}{\, \cdot \,}$. 
		\smallskip
		
		\noindent
		“(2) $\Longrightarrow$ (3):” Normal operators admit a functional calculus by  Theorem~\ref{appendix:functional_calculus_normal_operators:thm:functional_calculus}. 
		\smallskip
		
		\noindent
		“(3) $\Longrightarrow$ (1):” Functional calculus allows me to recover the projection-valued measure via $P(\Omega) = 1_{\Omega}(H)$, where the latter is the characteristic function for the  Borel set $\Omega \subseteq \C$. Because $H$ is bounded, I can pick any function that satisfies $f(E) = E$ on the spectrum of $H$ and is made into a bounded function by modifying it outside of $\sigma(H)$ in a measurable way. In that case I recover the spectral decomposition of the operator via functional calculus, 
		\begin{align*}
			f(H) &= \int_{\C} f(E) \, \dd P(E)
			= \int_{\sigma(H)} f(E) \, \dd P(E)
			\\
			&= \int_{\sigma(H)} E \, \dd P(E)
			= \int_{\C} E \, \dd P(E) 
			= H 
			. 
		\end{align*}
		To be consistent with the notation used in this paper, I should use ${}^{\ddagger}$ for the scalar product that makes the spectral projections orthogonal, $P(\Omega)^{\ddagger} = P(\Omega)$. 
	\end{proof}
	I will be using two decompositions of diagonalizable operators in many places of the main body of the text: 
	\begin{theorem}[Cartesian and polar decomposition]~
		\begin{enumerate}[(1)]
			\item $H$ is diagonalizable if and only if it is possible to write 
			\begin{align*}
				H = H_{\Re} + \ii H_{\Im} 
			\end{align*}
			for two hermitian operators $H_{\Re,\Im} = H_{\Re,\Im}^{\ddagger}$ that \emph{commute}, $[H_{\Re} , H_{\Im}] = 0$. 
			\item $H \in \mathcal{B}(\Hil)^{-1}$ is diagonalizable with bounded inverse if and only if there exist a $\sscppro{\, \cdot \,}{\, \cdot \,}$-unitary $V_H$ and a hermitian, strictly positive operator $\sabs{H} = \sabs{H}^{\ddagger}$ so that 
			\begin{align*}
				H = V_H \, \sabs{H} 
			\end{align*}
			holds and the two operators \emph{commute}, $[V_H , \sabs{H}] = 0$. 
		\end{enumerate}
	\end{theorem}
	\begin{proof}
		\begin{enumerate}[(1)]
			\item Suppose $H$ is diagonalizable. Then $H$ is normal with respect to the weighted scalar product~\eqref{diagonalizable_operators:eqn:weighted_scalar_product}, \ie $[H , H^{\ddagger}] = 0$. Therefore, I can define real and imaginary part operators with respect to the adjoint ${}^{\ddagger}$, 
			\begin{align*}
				H_{\Re} &= \frac{1}{2} \bigl ( H + H^{\ddagger} \bigr ) 
				= H_{\Re}^{\ddagger}
				, 
				\\
				H_{\Im} &= \frac{1}{\ii 2} \bigl ( H - H^{\ddagger} \bigr ) 
				= H_{\Im}^{\ddagger}
				. 
			\end{align*}
			These two operators commute, because $H$ and $H^{\ddagger}$ do. 
			
			Conversely, if I am given a decomposition $H = H_{\Re} + \ii H_{\Im}$ in terms of two commuting, hermitian operators, then $H^{\ddagger} = H_{\Re} - \ii H_{\Im}$ holds true. Clearly, $H$ and its adjoint commute exactly when real and imaginary part operators do. 
			\item Suppose $H$ is diagonalizable and has a bounded inverse. Then we can define $\sabs{H}$ using functional calculus for any function that satisfies $g(E) = \sabs{E}$ on $\sigma(H)$. The operator $\sabs{H}$ is \emph{strictly} positive as $g$ is strictly positive on $\sigma(H)$. That is because $H \in \mathcal{B}(\Hil)^{-1}$ implies the spectrum $\sigma(H)$ is gapped away from $0$, \ie for a ball $B_R(0)$ centered at $0$ of sufficiently small radius $R > 0$ I have $\sigma(H) \cap B_R(0) = \emptyset$. 
			
			The phase $V_H = f(H)$ can be defined via the function 
			\begin{align*}
				f(E) = 
				\begin{cases}
					\nicefrac{E}{\sabs{E}} & E \neq 0 \\
					0 & E = 0 \\
				\end{cases}
				. 
			\end{align*}
			Because $H$ is assumed invertible, we know $0 \not\in \sigma(H)$ does not lie in the spectrum and $f$ is invertible with a bounded inverse on $\sigma(H)$. The inverse of $f$ is just its complex conjugate, $f^{-1} = \overline{f}$. 
			
			That not only shows that 
			\begin{align*}
				V_H^{\ddagger} &= f(H)^{\ddagger} 
				= \overline{f}(H) 
				= f^{-1}(H) \\
				&= V_H^{-1} 
			\end{align*}
			is unitary with respect to $\scppro{\, \cdot \,}{\, \cdot \,}$, but also that the two operators commute, 
			\begin{align*}
				V_H \, \sabs{H} &= f(H) \, g(H) 
				= (f g)(H) 
				= (g f)(H) 
				\\
				&= g(H) \, f(H) = \sabs{H} \, V_H 
				. 
			\end{align*}
			Conversely, suppose I can write $H = V_H \, \sabs{H}$ as the product of a $\scppro{\, \cdot \,}{\, \cdot \,}$-unitary $V_H$ and a strictly positive, hermitian operator $\sabs{H} = \sabs{H}^{\ddagger}$, which mutually commute. Then since $\sabs{H}$ commutes with $V_H$ if and only if it commutes with $V_H^{-1}$, I can write the adjoint operators as 
			\begin{align*}
				H^{\ddagger} &= \bigl ( V_H \, \sabs{H} \bigr )^{\ddagger} 
				= \sabs{H}^{\ddagger} \, V_H^{\ddagger} 
				\\
				&= \sabs{H} \, V_H^{-1} 
				= V_H^{-1} \, \sabs{H} 
				. 
			\end{align*}
			The commutator now has to vanish, since $V_H$ and its inverse annihilate one another, 
			\begin{align*}
				[H , H^{\ddagger}] &= \sabs{H} \, V_H \, V_H^{-1} \, \sabs{H} - \sabs{H} \, V_H^{-1} \, V_H \, \sabs{H} 
				\\
				&= \sabs{H}^2 - \sabs{H}^2 
				= 0 
				, 
			\end{align*}
			and I have verified that $H$ is normal. 
		\end{enumerate}
	\end{proof}
	%
	% subsection characterizations_of_diagonalizability (end)

	\subsection{Useful facts about diagonalizable operators} % (fold)
	\label{appendix:diagonalizable_operators:facts}
	There are a few facts about diagonalizable operators I will use in the paper. 
	\begin{lemma}\label{appendix:diagonalizable_operators:lem:useful_facts}
		Suppose $V \in \mathcal{B}(\Hil)^{-1}$ is a similarity transform and $H$ is diagonalizable. Then the following holds true: 
		\begin{enumerate}[(1)]
			\item $V \, H \, V^{-1}$ is diagonalizable. 
			\item $K \, H \, K$ is diagonalizable, where $K$ is any complex conjugation on $\Hil$. 
			\item $H^{\dagger}$ is diagonalizable. 
			\item Real and imaginary parts of $V \, H \, V^{-1}$ are related to those of $H$ via the similarity transform $V$, 
			\begin{align*}
				\bigl ( V \, H \, V^{-1} \bigr )_{\Re,\Im} &= V \, H_{\Re,\Im} \, V^{-1} 
				. 
			\end{align*}
			\item The functional calculi of $H$ and $V \, H \, V^{-1}$ are related by the similarity transform $V$, that is 
			\begin{align*}
				f \bigl ( V \, H \, V^{-1} \bigr ) &= V \, f(H) \, V^{-1} 
			\end{align*}
			holds for all bounded Borel functions. The same holds true when replacing $V$ by $V K$. 
		\end{enumerate}
	\end{lemma}
	\begin{proof}
		\begin{enumerate}[(1)]
			\item Since $H$ is diagonalizable, there exists a similarity transform $G$ so that 
			\begin{align*}
				G \, H \, G^{-1} &= \int_{\C} E \, \dd P(E) 
			\end{align*}
			holds true. But then this immediately implies 
			\begin{align*}
				\widetilde{G} \, V \, H \, V^{-1} \, \widetilde{G}^{-1} &= \int_{\C} E \, \dd P(E) 
				, 
			\end{align*}
			for $\widetilde{G} = G \, V^{-1}$, that is, $V \, H \, V^{-1}$ is diagonalizable as well. 
			\item Since $H$ is diagonalizable, the operator possesses a functional calculus (\cf Theorem~\ref{appendix:diagonalizable_operators:thm:weighted_characterization}~(3)) and the projection-valued measure 
			\begin{align*}
				P(\Omega) := 1_{\Omega}(H)
			\end{align*}
			for $H$ can be recovered from it. Now I define the family of operators 
			\begin{align*}
				\overline{P}(\Omega) := K \, 1_{\overline{\Omega}}(H) \, K 
			\end{align*}
			indexed by Borel sets $\Omega$. I will show that $\overline{P}(\Omega)$ is the projection-valued measure for $K \, H \, K$. 
			
			Clearly, this defines yet another projection-valued measure: $\overline{P}(\Omega)^2 = \overline{P}(\Omega)$, completeness and strong $\sigma$-additivity follow directly from the definition. The only open question is orthogonality. Let $\sscppro{\varphi}{\psi} = \scpro{\varphi}{W \, \psi}$ be a scalar product with respect to which the projection-valued measure $P(\Omega) = P(\Omega)^{\ddagger}$ for $H$ is orthogonal. Then a straightforward computation confirms that $\overline{P}(\Omega)$ is hermitian with respect to the scalar product 
			\begin{align*}
				\scpro{\varphi}{\psi}_{\overline{W}} \overset{\mathrm{def}}{=} \bscpro{\varphi}{K \, W \, K \psi}
				. 
			\end{align*}
			And from this projection-valued measure, I can construct the operator 
			\begin{align*}
				\overline{H} &= \int_{\C} E \, \dd \overline{P}(E) 
				, 
			\end{align*}
			which by its very definition is diagonalizable. By approximating $f(E) = E$ on $\sigma(H) \cup \overline{\sigma(H)}$ with simple functions, I can make the following formal manipulations rigorous: 
			\begin{align*}
				K \, H \, K &= \int_{\C} \bar{E} \; K \, \dd P(E) \, K 
				\\
				&= \int_{\C} E \; \dd \bigl (K \, P(\bar{E}) \, K \bigr )
				\\
				&= \int_{\C} E \; \dd \overline{P}(E)
			\end{align*}
			That gives us an explicit diagonalization of $K \, H \, K$. 
			\item Since $H$ is diagonalizable, $H$ commutes with its adjoint $H^{\ddagger} = H^{\dagger_W} = W^{-1} \, H^{\dagger} \, W$ by Theorem~\ref{appendix:diagonalizable_operators:thm:weighted_characterization} where I have abbreviated $W = G^{\dagger} G = W^{\dagger}$ for convenience. Taking the $\scpro{\, \cdot \,}{\, \cdot \,}$-adjoint of the commutator yields the operator with which $H^{\dagger}$ commutes, 
			\begin{align*}
				0 &= \Bigl ( \bigl [ H \, , \, H^{\dagger_W} \bigr ] \Bigr )^{\dagger}
				= \bigl ( H \, W^{-1} \, H^{\dagger} \, W - W^{-1} \, H^{\dagger} \, W \, H \bigr )^{\dagger} 
				\\
				&= W^{\dagger} \, H \, (W^{-1})^{\dagger} \, H^{\dagger} - H^{\dagger} \, W^{\dagger} \, H \, (W^{-1})^{\dagger} 
				\\
				&= W \, H \, W^{-1} \, H^{\dagger} - H^{\dagger} \, W \, H \, W^{-1}
				\\
				&= \bigl [ W \, H \, W^{-1} \, , \, H^{\dagger} \bigr ]
				.  
			\end{align*}
			However, the operator 
			\begin{align*}
				W \, H \, W^{-1} &= \bigl ( H^{\dagger} \bigr )^{\dagger_{W^{-1}}} 
			\end{align*}
			is nothing but the weighted adjoint of $H^{\dagger}$ with respect to a weighted scalar product 
			\begin{align*}
				\scpro{\varphi}{\psi}_{W^{-1}} &= \scpro{\varphi \, }{ \, W^{-1} \psi} 
			\end{align*}
			with the inverse weight $W^{-1}$. Thus, we obtain real and imaginary parts, 
			\begin{align*}
				(H^{\dagger})_{\Re} &= \frac{1}{2} \Bigl ( H^{\dagger} + \bigl ( H^{\dagger} \bigr )^{\dagger_{W^{-1}}} \Bigr )
				, 
				\\
				(H^{\dagger})_{\Im} &= \frac{1}{\ii 2} \Bigl ( H^{\dagger} - \bigl ( H^{\dagger} \bigr )^{\dagger_{W^{-1}}} \Bigr )
				, 
			\end{align*}
			which commute with one another. Consequently, also the adjoint operator $H^{\dagger}$ is diagonalizable. 
			\item This follows directly from the definition of real and imaginary parts with respect to the $G$-weighed scalar product from (2) and the explicit expression of the $V G$-weighted adjoint as $V \, H^{\ddagger} \, V^{-1}$. 
			\item This follows from (4), the definition of the projection-valued measure in terms of real and imaginary parts as well as the definition of $f(H)$ via the projection-valued measure (\cf Theorem~\ref{appendix:diagonalizable_operators:thm:weighted_characterization}). 
		\end{enumerate}
	\end{proof}
	Lastly, there is a close connection between the projection-valued measures of $H$ and $H^{\dagger}$, which I will exploit. 
	\begin{lemma}\label{appendix:diagonalizable_operators:lem:similarity_tranform_functional_calculus}
		Assume $H$ is bounded and diagonalizable. Then the following holds true: 
		\begin{enumerate}[(1)]
			\item For any Borel set $\Omega \subseteq \C$ the spectral projections of $H$ and its adjoint $H^{\ddagger}$ are related by 
			\begin{align*}
				1_{\Omega}(H^{\ddagger}) = 1_{\overline{\Omega}}(H)
				, 
			\end{align*}
			where ${}^{\ddagger}$ is the adjoint with respect to any scalar product that makes $H$ normal. 
			\item Let $V$ be a similarity transform for which 
			\begin{align*}
				V \, H \, V^{-1} &= \pm H^{\dagger}
			\end{align*}
			holds true. Then for any Borel set $\Omega \subseteq \C$ the spectral projections are related by 
			\begin{align*}
				V \, 1_{\Omega}(H) \, V^{-1} &= 1_{\Omega}(\pm H^{\dagger})
				= 1_{\overline{\Omega}}(\pm H) 
				= 1_{\pm \overline{\Omega}}(H) 
				. 
			\end{align*}
			\item The statement~(2) holds true also if we replace $V$ by $V \, K$, where $K$ is a complex conjugation on $\Hil$. 
			\item $\sigma(H^{\dagger}) = \sigma(H^{\ddagger}) = \overline{\sigma(H)}$
		\end{enumerate}
	\end{lemma}
	\begin{proof}
		\begin{enumerate}[(1)]
			\item Since product sets are a neighborhood basis of the Borel $\sigma$-algebra of $\C$, I may assume without loss of generality that $\Omega = \Omega_{\Re} \times \Omega_{\Im}$ is the product of two Borel sets of $\R$. For such sets, I can verify the claim directly using the cartesian decomposition $H^{\ddagger} = H_{\Re} - \ii H_{\Im}$, 
			\begin{align*}
				1_{\Omega}(H^{\ddagger}) &= 1_{\Omega_{\Re}}(H_{\Re}) \; 1_{\Omega_{\Im}}(-H_{\Im}) 
				\\
				&= 1_{\Omega_{\Re}}(H_{\Re}) \; 1_{-\Omega_{\Im}}(H_{\Im}) 
				\\
				&= 1_{\overline{\Omega}}(H)
				. 
			\end{align*}
			\item Since $H$ is diagonalizable, then by Lemma~\ref{appendix:diagonalizable_operators:lem:useful_facts}~(2) so is $H^{\dagger} = W \, H^{\ddagger} \, W^{-1}$. Here I have introduced the shorthand $W = G^{\dagger} \, G$; observe that also $W \in \mathcal{B}(\Hil)^{-1}$ is a similarity transform. Thanks to Lemma~\ref{appendix:diagonalizable_operators:lem:useful_facts}~(3) real and imaginary part operators of $H^{\dagger}$ and $W \, H^{\dagger} \, W^{-1}$ are related by the similarity transform $W$. Consequently, $W$ relates also their spectral projections, and combined with (1), I obtain the claim. 
			\item The proof is identical to (2). 
			\item The spectrum $\sigma(H)$ coincides with the support of the projection-valued measure. Therefore, the claim follows from directly (1). Moreover, given that 
			\begin{align*}
				H^{\ddagger} - z &= W^{-1} \, \bigl ( H^{\dagger} - z \bigr ) \, W 
			\end{align*}
			is invertible exactly when $H^{\dagger} - z$ is, the spectra $\sigma(H^{\dagger}) = \sigma(H^{\ddagger})$ of the two adjoints agree. 
		\end{enumerate}
	\end{proof}
	%
	% subsection similarity_transforms_of_normal_operators (end)
	% section diabonalizable_operators (end)

	\section{Inner and outer continuity of spectra} % (fold)
	\label{appendix:continuity_spectra}
	Suppose the not necessarily hermitian operator $H(\lambda)$ continuously depends on a parameter $\lambda$ and $\sigma \bigl ( H(\lambda) \bigr ) \subset \C$ is its $\lambda$-dependent spectrum. 
	
	Generally, \emph{spectrum may not suddenly appear out of nowhere}, at least if the perturbation is weak enough. That is, if I take any compact set $K \subset \C$ that lies entirely inside a spectral gap of $\sigma \bigl ( H(\lambda_0) \bigr )$ at $\lambda = \lambda_0$, \ie $\sigma \bigl ( H(\lambda_0) \bigr ) \cap K = \emptyset$, then there exists an interval around $\lambda_0$ such that 
	\begin{align*}
		\sigma \bigl ( H(\lambda) \bigr ) \cap K = \emptyset 
	\end{align*}
	holds on $(\lambda_0 - \delta , \lambda_0 + \delta )$. This is referred to as \emph{outer} or \emph{upper continuity.} 
	
	Not all perturbations of operators have \emph{inner} or \emph{lower continuous} spectrum $\sigma \bigl ( H(\lambda) \bigr )$, though. Intuitively speaking, \emph{spectrum} of operators \emph{may not suddenly disappear} when perturbed. I call $\sigma \bigl ( H(\lambda) \bigr )$ inner continuous if and only if for any open set $O \subseteq \C$ so that at $\lambda = \lambda_0$ the intersection $\sigma \bigl ( H(\lambda_0) \bigr ) \cap O \neq \emptyset$ is non-empty, there exists an interval $(\lambda_0 - \delta , \lambda_0 + \delta)$ around $\lambda_0$ so that 
	\begin{align*}
		\sigma \bigl ( H(\lambda) \bigr ) \cap O \neq \emptyset
	\end{align*}
	remains non-empty in that interval. 
	
	As a general fact, the spectrum of arbitrary perturbations of operators are only outer continuous (\cf \cite[Theorem~3.1]{Kato:perturbation_theory:1995}). However, the spectra need \emph{not} be inner continuous; Kato gives an explicit counterexample on \cite[p.~210]{Kato:perturbation_theory:1995}. 
	
	Perturbations in the space of normal operators \cite[Proposition~1]{Bellissard:Lipschitz_continuity_spectra_Hofstadter:1994} lead to spectra that are known to be inner \emph{and} outer continuous. It stands to reason that this extends to \emph{all diagonalizable} operators as well. 
	\begin{conjecture}\label{appendix:continuity_spectra:thm:continuity_spectra}
		Let $[0,1] \ni \lambda \mapsto H(\lambda)$ be a continuous path in the set of \emph{diagonalizable} operators. Then the spectrum $\sigma \bigl ( H(\lambda) \bigr )$ is inner and outer semicontinuous in $\lambda$. 
	\end{conjecture}
	I will investigate this point in a future work. 
	
	One last note on Kato's counterexample: he constructs a perturbation of the shift operator on $\Z$ whose spectra 
	\begin{align*}
		\sigma \bigl ( H(\lambda) \bigr ) = 
		\begin{cases}
			\Sone & \lambda \neq 0 \\
			\mathbb{D}^1 & \lambda = 0 \\
		\end{cases}
	\end{align*}
	are either the circle line $\Sone$ or the closed unit disc $\mathbb{D}^1$ in the complex plane. At first glance, it would seem that the spectrum is not outer continuous at $\lambda = 0$ (even though it is). Indeed, for any $\lambda_0 > 0$ the scaled unit disc $\eps \mathbb{D}^1$, $0 < \eps < 1$, lies inside $\Sone$, \ie $\sigma \bigl ( H(\lambda_0) \bigr ) \cap \eps \mathbb{D}^1 = \emptyset$. And this remains true for all $\lambda$ from the open neighborhood $(0,\lambda_0)$. Of course, this argument also applies verbatim when $\lambda_0 < 0$. At the point where the spectrum changes, $\lambda = 0$, the initial \emph{assumption} is violated, 
	\begin{align*}
		\sigma \bigl ( H(0) \bigr ) \cap \eps \mathbb{D}^1 = \eps \mathbb{D}^1 \neq \emptyset
		, 
	\end{align*}
	which resolves the apparent contradiction. 
	
	However, the spectrum is \emph{not} inner continuous at $\lambda = 0$: all spectrum in the interior of the unit disc is unstable, any slight perturbation ($\lambda \neq 0$) will make it disappear. 
	% section inner_and_outer_continuity_of_spectra (end)
\end{appendix}
%

%%% end content %%% (end)

\bibliographystyle{apsrev4-2}
\bibliography{/Users/max/Library/texmf/tex/latex/max/bibliography}

\end{document}